\documentclass[11pt]{article}

\usepackage{amssymb,amsmath,amsthm,amsfonts,dsfont}
\usepackage[T1]{fontenc}
\usepackage{lmodern}

\usepackage[colorlinks,hypertexnames=false]{hyperref}
\usepackage{authblk}
\usepackage{tabularx}
\hypersetup{
	pdfstartview={FitH},
	pdfnewwindow=true,
	colorlinks=true,
	linkcolor=blue,
	citecolor=blue,
	filecolor=blue,
	urlcolor=blue}
\usepackage[capitalise,noabbrev]{cleveref}
\usepackage[all]{hypcap}
\usepackage{url}
\usepackage{graphicx}
\usepackage[font=small]{caption}
\usepackage[subrefformat=parens,labelformat=parens]{subcaption}
\usepackage{float}
\usepackage{footnote}
\usepackage{enumitem}
\usepackage{verbatim}
\usepackage[margin=1in]{geometry}
\usepackage{multicol}

\newcolumntype{Y}{>{\centering\arraybackslash}X}

\newtheorem{theorem}{Theorem}
\newtheorem{lemma}{Lemma}
\newtheorem{definition}[lemma]{Definition}

\newtheorem{corollary}[lemma]{Corollary}
\theoremstyle{remark}

\theoremstyle{plain}

\newcommand{\ket}[1]{\lvert#1\rangle}
\newcommand{\bra}[1]{\langle#1\rvert}
\newcommand{\ketbra}[1]{\lvert#1 \rangle \mskip-2mu \langle #1\rvert}

\usepackage{mathtools}
\usepackage{multirow} 
\usepackage{makecell}
\usepackage{tocloft}

\def\cO{{\cal O}}

\title{Optimal quantum simulation of linear non-unitary dynamics}

\author[]{Guang Hao Low\thanks{Corresponding author: \href{mailto:guanghaolow@google.com}{guanghaolow@google.com}} }
\author[]{Rolando D. Somma}
\affil[]{Google Quantum AI, Venice, CA 90291, United States}
\begin{document}

\maketitle
\begin{abstract}
We present a quantum algorithm for simulating the time evolution generated by any bounded, time-dependent operator $-A$ with non-positive logarithmic norm, thereby serving as a natural generalization of the Hamiltonian simulation problem.
Our method generalizes the recent Linear-Combination-of-Hamiltonian-Simulation (LCHS) framework.
In instances where $A$ is time-independent,
we provide a block-encoding of the evolution operator $e^{-At}$ with $\mathcal{O}\big(t\log\frac{1}{\epsilon})$ queries to the block-encoding oracle for $A$. 
We also show how the normalized evolved state can be prepared with $\mathcal{O}(1/\|e^{-At}\ket{\vec{u}_0}\|)$ queries to the oracle that prepares the normalized initial state $\ket{\vec{u}_0}$.
These complexities are optimal in all parameters and improve the error scaling over prior results.
Furthermore, we  show that any improvement of our approach exceeding a constant factor of approximately 3 is infeasible.
For general time-dependent operators $A$, we also prove that a uniform trapezoidal rule on our LCHS construction yields exponential convergence, leading to simplified quantum circuits with improved gate complexity compared to prior nonuniform-quadrature methods.
\end{abstract}
\tableofcontents
\section{Introduction}
The efficient simulation of exponentially large dynamical systems is one of the most promising applications of quantum computing.
In particular, the case of unitary evolution generated by Hermitian operators, commonly known as the Hamiltonian simulation problem, was a main motivation for quantum computing~\cite{Feynman1982} and has attracted the greatest interest. 
Following a series of results~\cite{Lloyd1996universal,Aharonov2003Adiabatic,Childs2010,Berry2015Hamiltonian,Berry2014Exponential}, an algorithm~\cite{Low2016HamSim} with optimal quantum query complexity with respect to all parameters -- namely, evolution time $t$, error $\epsilon$, and matrix norm $\alpha$~\cite{Low2016Qubitization} -- was discovered for the time-independent case.
These asymptotic improvements in the query setting have proven instrumental in reducing the cost of classically intractable simulation problems of scientific and industrial relevance~\cite{Babbush2018encoding,vonBurg2020carbon,Lee2020hypercontraction,low2025fast}.

However, results with optimal quantum query complexity are not known for the more general case of linear non-unitary dynamics. This is generated by the time-dependent, inhomogeneous linear differential equation
\begin{align}\label{eq:inhomogeneous_differential_equation}
\frac{\mathrm{d}}{\mathrm{d}t}\vec{u}(t)=-A(t)\vec{u}(t)+\vec{b}(t),\quad \vec{u}(0)\doteq\vec{u}_0.
\end{align}
Non-unitary dynamics appears in 
other exciting application areas, including dissipative nonlinear dynamics by Carleman linearization~\cite{Liu2023DissipativeNonlinear,Liu2023ReactionDiffusion}, partial differential equations~\cite{Childs2021PDE}, and linear matrix equations~\cite{Somma2025LinearMatrix}, but exploration is hindered by a more limited understanding of how to best approach~\cref{eq:inhomogeneous_differential_equation}.
Within the context of quantum algorithms, 
the goal is to prepare a quantum state that approximates the normalized $\ket{\vec{u}(t)}\doteq \vec{u}(t)/|\vec{u}(t)|$ for some $t>0$, which has the formal solution
\begin{align}\label{eq:inhomogeneous_solution}
\vec{u}(t)=U_0(t)\vec{u}_0+\int_0^tU_{s}(t)\vec{b}(s)\mathrm{d}s,\quad U_s(t)\doteq \mathcal{T}e^{-\int_s^tA(s')\mathrm{d}s'},
\end{align}
in terms of the time-ordering operator $\mathcal{T}$.
In the special case of anti-Hermitian $A(t)$ and a zero forcing term $\vec{b}(t)=0$, \cref{eq:inhomogeneous_differential_equation} is the famous Schr\"odinger equation that underlies quantum mechanics and the Hamiltonian simulation problem.
In the quantum query setting,~\cref{eq:inhomogeneous_solution} is solved by queries to a unitary $\textsc{Prep}_{\ket{\vec{u}_0}}$ that prepares the initial state $\ket{\vec{u}_0}$ from a computational basis state, and a unitary~\cref{def:block_encoding} that block-encodes the operator $A$ (see~\cref{def:block_encoding_td} for the time-dependent case), which is a well-established access model that generalizes sparse matrices with efficiently computable entries~\cite{Aharonov2003Adiabatic}.
\begin{definition}[Block-encoding~\cite{Low2016Qubitization}]\label{def:block_encoding}
The $n$-qubit matrix $A\in\mathbb{C}^{2^n\times 2^n}$ is block-encoded with normalization factor $\alpha_A\ge\|A\|$ by the $n+a$-qubit unitary operator $\textsc{Be}[A/\alpha_A]\in\mathbb{C}^{2^{n+a}\times 2^{n+a}}$ if
\begin{align}
(\bra{0}^{a}\otimes \mathcal{I})\textsc{Be}\left[\frac{A}{\alpha_A}\right](\ket{0}^{a}\otimes \mathcal{I})=\frac{A}{\alpha_A}.
\end{align}
The state $\ket 0^{a}$ is the all-zero state of the $a$ qubits.
\end{definition}

\begin{table}[t]
    \centering\small
    \begin{tabularx}{\textwidth}{c|Y|c|c|c}
    \hline\hline
         \multirow{2}{*}{Result}&  \multirow{2}{*}{Innovation} & \multirow{2}{*}{Conditions on $A$} & \multicolumn{2}{c}{Query complexity $\mathcal{O}(\cdot)$}
         \\
         &&&$\textsc{Be}[A/\alpha_A]$& $\textsc{Prep}_{\ket{\vec{u}_0}}$
         \\
         \hline\hline
         \cite{Low2016HamSim,Low2016Qubitization}&QSP+Qubitization&Anti-Hermitian& $\alpha_A t+\log\frac{1}{\epsilon}$ & 1
         \\
         \hline
         \cite{Berry2017Differential}&Taylor series+QLSP&\multirow{2}{*}{$\mathrm{Re}[\Lambda]\succeq 0$}& \multicolumn{2}{c}{$u_\text{r}\kappa\alpha_A t\;\mathrm{polylog}\big(\frac{\kappa\alpha_A tu_\text{r}}{\epsilon}\big)$ }
         \\
         \cite{Low2024Eigenvalue,Low2024QLSP}&Faber series+QLSP& & $\kappa\frac{|\vec{u}_0|}{|\vec{u}(t)|}(\alpha_A t+\log\frac{\kappa}{\epsilon})\log\frac{\kappa}{\epsilon}$ & $\frac{|\vec{u}_0|}{|\vec{u}(t)|}\|e^{-A(\cdot)}\|_{L^\infty}$      
         \\
         \hline
         \cite{Berry2024Dyson}&Dyson series+QLSP&\multirow{3}{*}{$A+A^\dagger\succeq0$}& $u_\text{r}\alpha_A t\log\frac{\alpha_A t}{\epsilon}\log\frac{1}{\epsilon}$ & $u_\text{r}\alpha_A t\log\frac{1}{\epsilon}$   
         \\
         \cite{an2023lchsoptimal}&LCHS&& 
         $u_{\text{r}}\alpha_A t\log^{1+\mathcal{O}(1)}\frac{1}{\epsilon}$ 
         & $u_{\text{r}}$
         \\
         This work&Optimal LCHS&& 
          $u_{\text{r}}\alpha_A t\log\frac{1}{\epsilon}$ 
         & $u_{\text{r}}$
         \\
     \hline\hline
    \end{tabularx}
    \caption{Query complexities of various methods that prepare $\ket{\vec u(t)}$ to additive error $u_\text{r}\epsilon$, where $u_\text{r}\doteq {\|\vec{u}\|_{L^\infty}}/{|\vec{u}(t)|}$ for the homogeneous case $\vec b=0$ and time-independent $A$. The term QSP refers to quantum signal processing~\cite{Low2016HamSim}
    and QLSP refers to reductions to the quantum linear systems problem~\cite{Harrow2009,Costa2021linearsystems,Low2024QLSP}.
    For LCHS methods, $u_\text{r}= {|\vec{u}_0}|/{|\vec{u}(t)|}\doteq u_{\text{lchs}}$. 
}
    \label{tab:time_independent_comparison}
\end{table}

In~\cref{tab:time_independent_comparison},
we compare quantum algorithms
for the general case of non-Hermitian $A$ with the results of this work. 
To compare algorithmic advances, we focus on the homogeneous time-independent case where $\vec{u}(t)=e^{-At}\ket{\vec{u}_0}$.
These algorithms require some `stability' condition on $A$, such as non-negative eigenvalues
$\mathrm{Re}[\Lambda ]\succeq0$~\cite{Berry2014Differential,Berry2017Differential,Childs2020Spectral,Low2024Eigenvalue} in the eigendecomposition $A=S\Lambda S^{-1}$.
In this setting, all known results are suboptimal by some logarithmic factor of time $t$, error $\epsilon$, or condition number $\kappa=\|S\|\|S^{-1}\|$, which is a measure of matrix non-normality.
Under the stronger condition where $-A$ has non-positive logarithmic norm or $L\doteq\frac{1}{2}(A+A^\dagger)\succeq0$~\cite{Krovi2022LinearDE}, the recent Linear-Combination-of-Hamiltonian (LCHS)~\cite{an2023lchs,an2023lchsoptimal} approach achieves better scaling, in particular optimal scaling in $u_\text{lchs}\doteq|\vec{u}_0|/|\vec{u}(t)|$ queries to $\textsc{Prep}_{\ket{\vec{u}_0}}$.
However, the query complexity to $\textsc{Be}[A/\alpha_A]$ scales as $\mathcal{O}(\alpha_At\log^{1+\mathcal{O}(1)}\frac{1}{\epsilon})$, which is suboptimal in error.
Furthermore, a no-go result~\cite[Proposition 7]{an2023lchsoptimal} suggested that achieving the optimal scaling $\mathcal{O}(\alpha_A t\log\frac{1}{\epsilon})$ was impossible within the existing LCHS framework.
Moreover, known lower bounds are additive in each parameter, in contrast to the multiplicative upper bound.
We note that scaling with other matrix- or state-dependent parameters is possible, but these alternate approaches~\cite{Krovi2022LinearDE,Fang2023Marching,Low2024Eigenvalue} are incomparable and out-of-scope for this work. 

LCHS in prior art~\cite{an2023lchs,an2023lchsoptimal} is based on a kernel function $\hat{f}$ whose inverse Fourier transform $f$ equals exponential decay for $x\ge0$, but is allowed to be arbitrary for $x<0$. Crucially, for any $A=L+iH$, where $L\succeq0$ and $H$ are Hermitian, this may be lifted to non-unitary matrix-valued exponential decay in a generalization of the spectral mapping theorem like
\begin{align}
\label{eq:exponential_decay_matrix}
\forall x\ge0,\;\; e^{-x}=\frac{1}{\sqrt{2\pi}}\int_{\mathbb{R}}\hat{f}(k)e^{-ikx}\mathrm{d}k
\quad\Rightarrow
\quad
\forall t\ge0,\;\; e^{-(L+iH)t}=\frac{1}{\sqrt{2\pi}}\int_{\mathbb{R}}\hat{f}(k)e^{-i(kL+H)t}\mathrm{d}k,
\end{align}
under some additional conditions on $\hat{f}$.
Hence, one block-encodes the non-unitary time-evolution operator $\textsc{Be}[e^{-At}/\alpha_{\hat{f}}]$ to additive error $\epsilon$ with normalization $\alpha_{\hat{f}}\doteq\frac{1}{\sqrt{2\pi}}\int_{\mathbb{R}}|\hat{f}(k)|\mathrm{d}k$ by the Linear-Combination-of-Unitaries (LCU) lemma~\cite{Childs2012LCU} after truncating the integral~\cref{eq:exponential_decay_matrix} to finite $R$ and discretizing it to a finite sum.
The query complexity is determined by that of Hamiltonian simulation of the `most expensive' unitary $e^{-i(\pm RL+H)t}$.
As optimal Hamiltonian simulation in the large time limit requires $\alpha t+\mathcal{O}(\log\frac{1}{\epsilon})$ queries to the block-encoded Hamiltonian~\cite{Low2016HamSim,Low2016Qubitization,Berry2024HamSim}, this has query complexity $\mathcal{O}(\alpha_A Rt+\log\frac{1}{\epsilon})$. 
If preparing $\ket{u(t)}$ is desired, this is further multiplied by at least $\alpha_{\hat{f}}$ to account for amplitude amplification.
Previously, a LCHS based on $\hat{f}_{\text{ACL23}}(k)\propto\big((1-ik)e^{(1+ik)^{\varphi}}\big)^{-1}$ for any constant $\varphi\in(0,1)$ was shown to achieve a truncation threshold $R=\mathcal{O}(\log^{1/\varphi}\frac{1}{\epsilon})$~\cite{an2023lchsoptimal}.
Unfortunately, the optimal scaling of $\alpha_{\hat{f}}R=\mathcal{O}(\log\frac{1}{\epsilon})$ was unattainable as $\lim_{\varphi\rightarrow1}\alpha_{\hat{f}_{\text{ACL23}}}=\infty$, and a no-go theorem~\cite{an2023lchsoptimal} forbids this scaling for any LCHS that satisfies~\cref{eq:exponential_decay_matrix}.

\begin{figure}[t]
    \centering   
    \vspace{-0.5cm}
    \includegraphics[width=\linewidth]{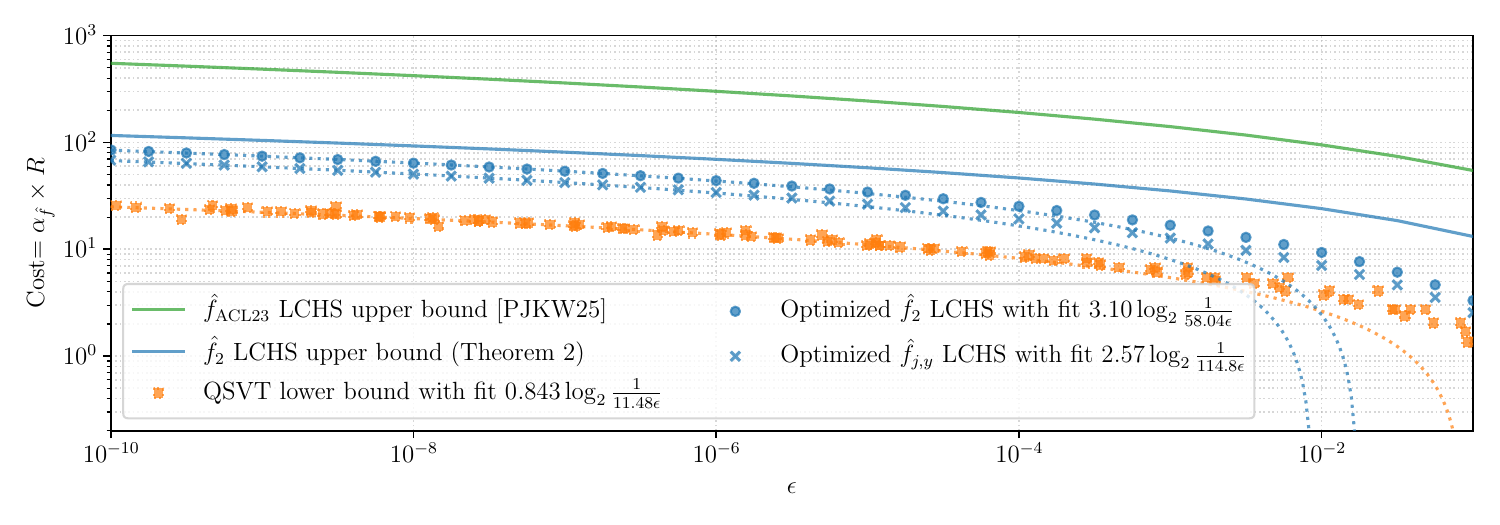}
    \vspace{-0.85cm}
    \caption{Query complexity prefactor ${\alpha_{\hat{f}}}\cdot R$ versus target error $\epsilon$ for the LCHS approximating $e^{-At}$ based on (green) the nearly-optimal kernel function $\hat{f}_{\text{ACL23}}$~\cite{an2023lchsoptimal,pocrnic2025LCHSconstants} and 
     (blue) our family of kernel functions $\hat{f}_{j,y}(k;\gamma,c)$. We plot upper bounds (solid lines) based on closed-form expressions that are valid for any $\epsilon$ as well as numerically bounded upper bounds for numerically optimized kernel functions with (circles) $j=2, y=1$ fixed and (crosses) the best possible $j,y$.
     The equivalent cost of a (yellow) numerically computed lower bound on an optimal QSVT approach for $e^{-At}$ in the case of Hermitian $A\succeq0$ establishes lower bounds on the performance of LCHS.
    (Dotted lines) The empirical fits indicate that our LCHS improves over prior art by $\gtrsim10\times$ for all $\epsilon$, and that further improvements by $\gtrsim3\times$ are infeasible for small $\epsilon$.
    }
    \label{fig:lchs_comparison}
\end{figure}

In this work, we resolve these limitations and present the first query-optimal algorithm for this problem.
Our key contributions are: 1) A generalized LCHS framework that achieves the optimal query complexity $\alpha_{\hat{f}}R=\mathcal{O}(\log\frac{1}{\epsilon})$ of block-encoding $e^{-At}$, circumventing a no-go result that applied to previous LCHS methods, and as shown in~\cref{fig:lchs_comparison}, with a significant benefit in practical applications; 
2) A novel, analytically simple, family of kernel functions that only approximate exponential decay. This is enabled by the generalized LCHS framework and is a conceptual departure from prior work that required an exact match;
3) A greatly simplified and more efficient implementation of the LCHS by proving that a simple, uniform discretization of the LCHS integral converges exponentially fast, leading to simpler quantum circuits than previous non-uniform quadrature methods;
4) New provable and empirical lower bounds on the query complexity, showing that our algorithm's performance is optimal up to a small constant factor (which we estimate to be approximately three), making further significant improvements in this framework infeasible.
These insights are formalized in our main results, which we now summarize.

Our key insight that sidesteps previous no-go limitations is that the more general condition of approximate exponential decay may still imply a LCHS that forms a controlled approximation to $e^{-At}$. 
Informally, 
\begin{align}\label{eq:exponential_decay_approximate}
\forall x\ge0,\;\; |e^{-x}-f(x)|\le\epsilon
\quad
\Rightarrow
\quad
\forall t\ge0,\;\;
e^{-(L+iH)t}\approx\frac{1}{\sqrt{2\pi}}\int_{\mathbb{R}}\hat{f}(k)e^{-i(kL+H)t}\mathrm{d}k.
\end{align}
This is formalized in our main mathematical result~\cref{thm:LCHS_general}, stated for general time-dependent $A$, which establishes rigorous error bounds on~\cref{eq:exponential_decay_approximate}.
A consequence of the exact exponential decay equality~\cref{eq:exponential_decay_matrix} is that previous approaches required $\hat{f}(z)$ to decay uniformly to zero in the lower complex plane~\cite{an2023lchs,an2023lchsoptimal}, which is also a key assumption of the previous no-go.
In contrast, we only require decay $\lim_{|z|\rightarrow\infty}\hat{f}(z)= 0$ in an infinite strip of finite width, which implies that $f(x)$ only approximates exponential decay.
This flexibility allows us to study a new $4$-parameter family of kernel functions
\begin{align}\label{eq:generalized_kernel_intro}
\forall j\ge1,\;y>0,\;\gamma>0,\;c\in\mathbb{R},\quad \hat{f}_{j,y}(k;\gamma,c)&\doteq\frac{(y+1)^{j-1}}{\sqrt{2\pi}}\frac{e^{c(1-i k)}e^{-\frac{k^2+1}{4\gamma^2}}}{(1-ik) (y+ik)^{j-1}},
\end{align}
that may grow to infinity outside both sides of this strip, but is able to achieve the optimal scaling.
\begin{theorem}[Approximate LCHS]\label{thm:LCHS_general}
For any $t\ge0$, let the cartesian decomposition $A=L+iH:[0,t]\rightarrow\mathbb{C}^{2^n\times 2^n}$ satisfy $L\succeq0$, $\|L\|_{L^1}<\infty$\footnote{For any time-dependent vector $\vec{x}(s)$ and matrix $M(s)$ defined for $s\in[0,t]$, let $|\vec{x}(s)|$ be the two-norm and $\|M(s)\|$ be the spectral norm. We pad with zeros outside $[0,t]$ so $\|\vec{x}\|_{L^{p}}\doteq(\int_{\mathbb{R}}|\vec{x}(s)|^{p}\mathrm{d}s)^{1/p}$ and $\|M\|_{L^{p}}\doteq(\int_{\mathbb{R}}\|M(s)\|^{p}\mathrm{d}s)^{1/p}$.}, $\|H\|_{L^1}<\infty$. 
Let the uniform strip $S_{[-y_0,0]}\doteq\{z:\mathrm{Im}[z]\in[-y_0,0]\}$ for some $y_0>1$. 
On this strip, let the kernel function $\hat{f}(z)$ decay uniformly such that $\lim_{|z|\rightarrow\infty}\hat{f}(z)=0$, let
$(z+i)\hat{f}(z)$ be analytic, and let the residue $\mathrm{Res}(\hat{f},-i)=\lim_{z\rightarrow-i}(z+i)\hat{f}(z)=\frac{i}{\sqrt{2 \pi }}$. Then the LCHS
\begin{align}\label{eq:LCHS_truncated}
 O_R(t)\doteq\frac{1}{\sqrt{2\pi}}\int_{-R}^R\hat{f}(k)\mathcal{U}(t;k)\mathrm{d}k,
\quad
\mathcal{U}(s;k)\doteq\mathcal{T}e^{-i\int_0^skL(s')+H(s')\mathrm{d}s'},
\end{align}
satisfies
\begin{align}\label{eq:thm_general_error_bound}
\|O_R(t)-U_0(t)\|&\le\frac{1}{\sqrt{2\pi}}\int_{\mathbb{R}\backslash[-R,R]}|\hat{f}(k)|\mathrm{d}k+\frac{1}{\sqrt{2\pi}}\int_{\mathbb{R}}|\hat{f}(k-iy_0)|\mathrm{d}k.
\end{align}
\end{theorem}

We prove the optimal cost scaling in our main algorithmic result~\cref{thm:LCHS_entire} by instantiating~\cref{thm:LCHS_general} with the explicit choice $\hat{f}_2(k;\gamma,c)\doteq \hat{f}_{2,1}(k;\gamma,c)$ illustrated in~\cref{fig:exponential_decay} and alternatively defined in terms of elementary functions as follows:
\begin{align}\label{eq:approximate_decay_scalar}
f_2(x;\gamma,c)\doteq e^{c}f_2(x+c;\gamma),
\;\;
f_2(x;\gamma)\doteq f_1(x;\gamma)+f_1(-x;\gamma),
\;\;
f_1(x;\gamma)\doteq \frac{e^{-x}}{2}\mathrm{erfc}\left(\frac{1}{2\gamma}-\gamma x\right),
\end{align}
where $\mathrm{erfc}$ is the complementary error function, $c$ is any positive constant,  $\gamma=\mathcal{O}(\frac{1}{c}\log^{1/2}\frac{1}{\epsilon})$, and $R=\mathcal{O}(c\gamma^2)$.
We design $f_2$ to be analytic in the entire complex plane, which is motivated by the super-polynomial to super-exponential decay of entire functions in the Fourier domain.
$f_2$ is also the convolution of $e^{-|x+c|+c}$ with a Gaussian kernel, and its simple Fourier transform and constant normalization
satisfy the conditions of~\cref{thm:LCHS_general} and allow us to derive simple and reasonably tight upper bounds on the performance of our LCHS.
Although infinitely many kernel functions satisfy~\cref{thm:LCHS_general}, not all lead to an optimal-scaling LCHS, even among those where $f$ is entire like $f_2$.
For instance, the approximations $f_{\text{HMS22}}=\frac{1}{2}e^{-x}\mathrm{erfc}(\gamma-x)$~\cite{holmes2022fluctuation} and $f_1(x;\gamma)$ both turn out to have unbounded norm $\alpha_{\hat{f}}\rightarrow\infty $ when $\gamma\rightarrow\infty $.
\begin{figure}
    \centering
    \vspace{-0.5cm}
    \includegraphics[width=0.75\linewidth]{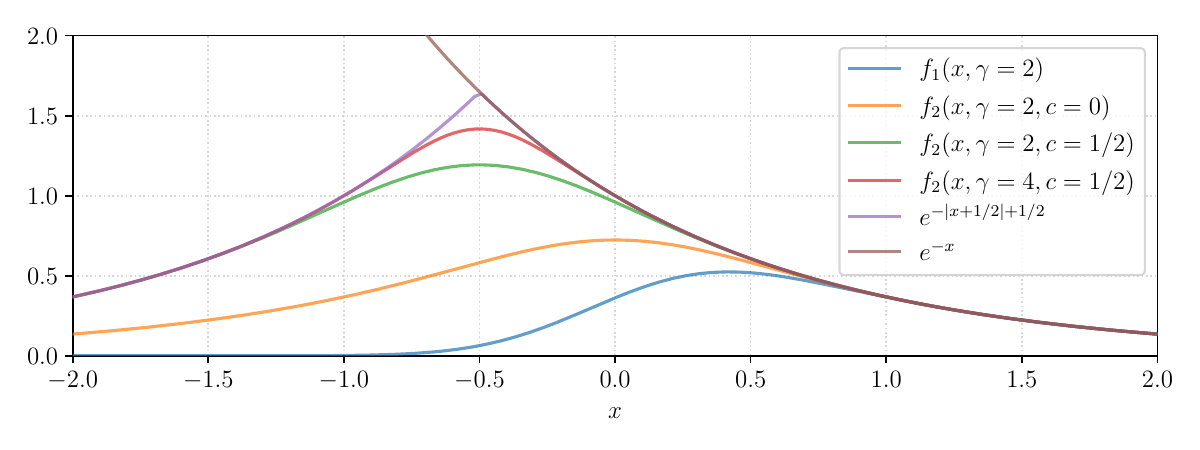}\vspace{-.5cm}
    \caption{Our entire approximation $f_2$ to exponential decay converges to $\lim_{\gamma\rightarrow\infty}f_2(k,\gamma,c)=e^{-|x+c|+c}$.}
    \label{fig:exponential_decay}
\end{figure}
\begin{theorem}[Optimal-scaling LCHS]\label{thm:LCHS_entire}
For any $t\ge0$, let $A=L+iH$ be as described in~\cref{thm:LCHS_general}.
For any $\epsilon_{\text{lchs}}\in(0,0.9027...]$ and any $c>0$,
the LCHS of time-evolution generated by  $kL+H$ and the kernel function $\hat{f}_2\doteq\hat{f}_{2,1}$ from ~\cref{eq:generalized_kernel_intro} with $\gamma=\frac{1}{c}\sqrt{c+\log\frac{1+1/({2\pi})}{\epsilon_{\text{lchs}}}}$ and $R=2c\gamma^2$ satisfies
\begin{align}\label{eq:thm_LCHS_truncation}
\left\|U_0(t)-O_{R,\gamma,c}(t)\right\|\le\epsilon_{\text{lchs}},
\quad
O_{R,\gamma,c}(t)\doteq \frac{1}{\sqrt{2\pi}}\int_{-R}^{R}\hat{f}_2(k;\gamma,c)\mathcal{U}(t;k)\mathrm{d}k,
\end{align}
with normalization factor $\alpha_{\hat{f}_2,R}\doteq\frac{1}{\sqrt{2\pi}}\int_{-R}^R|\hat{f}_2(k;\gamma,c)|\mathrm{d}k\le
\alpha_{\hat{f}_2}\doteq\alpha_{\hat{f}_2,\infty}=e^c\mathrm{erfc}\big(\frac{1}{2\gamma}\big)\le e^c$.
\end{theorem}

We also greatly improve the practical implementation of block-encoding the LCHS.
We prove in~\cref{thm:LCHS_quadrature} that the error $\epsilon_\text{quad}$ of a uniform quadrature is exponentially convergent in the stepsize $h$~\cite{Trefthen2014trapezoid}.
In contrast, previous analyses either showed polynomial error scaling for the uniform quadrature~\cite{an2023lchs}, or used a complicated multi-step non-uniform Gauss-Legendre quadrature~\cite{an2023lchsoptimal} with still an asymptotically worse stepsize, possibly a larger block-encoding normalization with no simple closed-form expression due to non-uniform weights, a more complicated quantum circuit implementation for non-uniform quadrature points, and a complicated analysis involving higher derivatives.

\begin{theorem}[Uniform quadrature on optimal LCHS]\label{thm:LCHS_quadrature}
Choose the LCHS parameters $\epsilon_\text{lchs},R,\gamma,c$ according to~\cref{thm:LCHS_entire}. For any $\epsilon_{\text{quad}}\in(0,4/15]$, and any stepsize ${h}\le{\pi}/\left(\frac{\|L\|_{L^1}}{2}+\log\left(\frac{64e^{3c/2}}
{15\epsilon_{\text{quad}}}\right)\right)$ such that $R/h$ is an integer, let $I_h$ be the uniform discretization of $O_{R,\gamma,c}(t)$. Then
\begin{align}\label{eq:thm_LCHS_quadrature}
\left\|U_0(t) -I_h\right\|\le\epsilon_{\text{lchs}}+\epsilon_{\text{quad}},
\quad
I_h\doteq \frac{h}{\sqrt{2\pi}}\sum_{j=-R/h}^{R/h}\hat{f}_2(hj;\gamma,c)\mathcal{U}(t;hj),
\end{align}
where the block-encoding normalization
$\alpha_{\hat{f}_2,R,h}\doteq\frac{h}{\sqrt{2\pi}}\sum_{j=-R/h}^{R/h}|\hat{f}_2(hj;\gamma,c)|$ satisfies 
\begin{align}
    |\alpha_{\hat{f}_2,R,h}-\alpha_{\hat{f}_2}|\le\frac{\epsilon_\text{lchs}}{1+2\pi}+\frac{\epsilon_\text{quad}}{e^{(\|L\|_{L^1}+c)/2}}.
\end{align}
\end{theorem}

We then block-encode non-unitary time-evolution in~\cref{thm:LCHS_algorithm} to error $\epsilon$ by an explicit and efficient quantum circuit implementation of the Linear-Combination-of-Unitaries (LCU) method~\cite{Childs2012LCU} applied to~\cref{thm:LCHS_entire,thm:LCHS_quadrature}, for instance, with $\epsilon_\text{quad}=\epsilon_\text{lchs}=\epsilon/3$.
Note that for any reasonably small error, the difference in block-encoding normalization $\alpha_{\hat{f}_2,R,h}$ from $\alpha_{\hat{f}_2}$ is extremely small and contributes a multiplicative constant cost that is $1+\mathcal{O}(\epsilon)$.
Our use of the uniform quadrature greatly simplifies the explicit quantum circuit implementation, which we provide.
As is known in prior art, the query complexity of $\textsc{Sel}_R$ for selecting time-evolution for different $k_j$ in~\cref{thm:LCHS_algorithm}, even with the additional control by $\ket{j}$, can be reduced to that of $U_0(\pm Rt)$ without any controls by any choice of Hamiltonian simulation algorithm.
In the time-independent case, Hamiltonian simulation by quantum signal processing and qubitization~\cite{Low2016HamSim,Low2016Qubitization} leads to an overall $\mathcal{O}(R\alpha t+\log\frac{1}{\epsilon})=\mathcal{O}(\alpha t\log\frac{1}{\epsilon})$ query complexity. 
\begin{theorem}[Block-encoding of optimal LCHS]\label{thm:LCHS_algorithm}
For any $\epsilon\in(0,4/5]$, and under the same condition as~\cref{thm:LCHS_entire,thm:LCHS_quadrature}, we may block-encode an operator $I_h'$ such that $\left\|U_0(t) -I_h'\right\|\le\epsilon$ by
\begin{align}
&\textsc{Be}\left[\frac{I_h'}{\alpha_{\hat{f}_2,R,h}}\right]=\textsc{Prep}^\dagger\cdot\textsc{Sel}_R\cdot\overline{\textsc{Prep}},
&\textsc{Sel}_R&=\sum_{j\in\mathcal{U}_R}\ket{j}\bra{j}\otimes \mathcal{U}(t;hj),\\\nonumber
&\overline{\textsc{Prep}}\ket{0}=\ket{\overline{\hat{f}}}\propto\sum_{j\in\mathcal{U}_R}e^{i\mathrm{Arg}[\hat{f}_2(hj;\gamma,c)]}|\hat{f}_2(hj;\gamma,c)|^{1/2}\ket{j},
&
\textsc{Prep}\ket{0}&=\ket{\hat{f}}\propto\sum_{j\in\mathcal{U}_R}|\hat{f}_2(hj;\gamma,c)|^{1/2}\ket{j},
\end{align}
where $\textsc{Prep}$ and $\overline{\textsc{Prep}}$ cost $\mathcal{O}\left(\log(\|L\|_{L^1}+\log\frac{1}{\epsilon})\log^{5/2}\frac{1}{\epsilon}\right)$ two-qubit gates, $\textsc{Sel}_R$ 
has the same query complexity $Q$ and gate complexity as simulating $\mathcal{T}e^{-i\int_0^t\pm RL(s)+H(s)\mathrm{d}s}$ plus $\mathcal{O}\big(Q\log(\|L\|_{L^1}+\log\frac{1}{\epsilon})\big)$ two-qubit gates.
\end{theorem}

We also optimize the constant factors in the query complexity of our optimal LCHS.
The query complexity proportionality constant ${\alpha_{\hat{f}}}\cdot R$ is loosely bounded by~\cref{thm:LCHS_entire}, such as by setting $c=1$ in~\cref{fig:lchs_comparison}.
Due to our use of the uniform quadrature, the LCU normalization factor $\alpha_{\hat{f},R}\approx \alpha_{\hat{f}}$ up to a vanishingly small error, 
unlike prior art~\cite{pocrnic2025LCHSconstants} where $\alpha_{\hat{f},R}$ can be a significant constant factor larger than $\alpha_{\hat{f}}$ due to non-uniform quadrature weights.
We obtain tighter upper bounds on approximation error by evaluating the integrals in~\cref{thm:LCHS_general} numerically, and then numerically minimizing the product ${\alpha_{\hat{f}_2,R}}\cdot R$ for each target $\epsilon$ with respect to $c,\gamma,R$.
As seen in~\cref{fig:lchs_comparison}, the LCHS with the kernel function $\hat{f}_2$ strictly outperforms prior art~\cite{an2023lchsoptimal,pocrnic2025LCHSconstants} in all parameter regimes, and further improvements of $\approx 25\%$ are obtained by also optimizing over $j,y$ in the more general family~\cref{eq:generalized_kernel_intro}.
Moreover, the proofs of implementation details in~\cref{thm:LCHS_quadrature,thm:LCHS_algorithm} for the special case $\hat{f}_2$ generalize readily to $\hat{f}_{j,y}$.

We also present lower bounds to justify that~\cref{thm:LCHS_entire} has optimal scaling.
Applying $\textsc{Be}[I_h'/\alpha_{\hat{f}}]$ to any initial state $\ket{\vec{u}_0}$ produces the normalized quantum state $\ket{\vec{u}'(t)}$ such that $|\ket{\vec{u}(t)}-\ket{\vec{u}'(t)}|\le\epsilon \;u_{\text{lchs}}$ with success probability $(u_{\text{lchs}}\alpha_{\hat{f}})^{-2}$.
By using $\mathcal{O}(u_{\text{lchs}}\alpha_{\hat{f}})$ rounds of amplitude amplification, we succeed with constant probability with the following overall cost.
\begin{align}\label{eq:queries_high_level_overview}
Q_A\doteq\text{Queries to}\;\textsc{Be}\left[\frac{A}{\alpha_A}\right]&=\mathcal{O}\left(u_{\text{lchs}}\cdot\alpha_{\hat{f}}\cdot (\text{Hamiltonian simulation for time}\;\alpha_A Rt)\right),
\\
Q_{\vec{u}_0}\doteq\text{Queries to}\;\textsc{Prep}_{\ket{\vec{u}}_0}&=\mathcal{O}\left(u_{\text{lchs}}\cdot\alpha_{\hat{f}}\right).
\end{align}
As $\alpha_{\hat{f}_2}=\Theta(1)$ is constant for any constant $c$, the query complexity of our LCHS to state preparation $\textsc{Prep}_{\ket{\vec{u}}_0}$, similar to prior LCHS approaches, is optimal due to an $\Omega(u_{\text{lchs}})$ lower bound~\cite[Corollary 8]{An2022Differential} based on quantum state discrimination.
We note that existing lower bounds on $Q_{A}$ are quite loose. 
When $A$ is anti-Hermitian, this reduces to Hamiltonian simulation with $u_{\text{lchs}}=1$ and $Q_A=\Omega(\alpha_At+\frac{\log1/\epsilon}{\log\log1/\epsilon})$~\cite{Berry2015Hamiltonian}.
For the general case of $A$ with $L\succeq0$, it is only known that $Q_A=\Omega(\min((\alpha_A t)^{\alpha},\log^{\alpha}\frac{1}{\epsilon}))$ for some unknown parameter $\alpha\le1$~\cite{An2022Differential}.
In~\cref{thm:lower_bound}, we present a significantly tighter quantum search lower bound $Q_A=\Omega(u_{\text{lchs}}\alpha_At)$, which proves optimality of our LCHS up to a logarithmic factor in $u_{\text{lchs}}$ when $\epsilon$ is constant.

The query complexity of block-encoding $\textsc{Be}[e^{-At}/\alpha]$ for some constant $\alpha>1$ can be bounded below by assuming that $A$ is Hermitian, and evaluating the approximate polynomial degree of an appropriate choice of $f$.
In this case, the spectral mapping theorem states that LCHS block-encodes $\textsc{Be}[f_{}(A/\alpha_A)]$, where $f$ is some function that satisfies
\begin{align}\label{eq:approximate_exponential}
\max_{x\in[0,1]}\left|\frac{e^{-\tau x}}{\alpha}- f(x)\right|\le \frac{\epsilon}{\alpha}=\epsilon',
\quad \max_{x\in[-1,0)}|f(x)|\le1,\quad\tau=\alpha_A t\ge0,\quad\alpha\ge1.
\end{align}
Let $p_n$ be the polynomial of degree $n=\min_f\widetilde{\mathrm{deg}}_{\epsilon'}(f)$, minimized over any $f$ that satisfies~\cref{eq:approximate_exponential} for some constant $\alpha$ -- for example, see~\cref{fig:optimal_polynomial}.
By Quantum Singular Value Transformations (QSVT)~\cite{Gilyen2018singular}, $f(A/\alpha_A)$ is block-encoded with a normalization factor at most $2(1+\epsilon')$ using exactly $\widetilde{\mathrm{deg}}_{\epsilon'}(f)$ queries to $\textsc{Be}[A/\alpha_A]$, which is optimal including constant factors by a matching $\widetilde{\mathrm{deg}}_{\epsilon'}(f)$~\cite{Montanaro2024MatrixQueryComplexity,Laneve2025QSPAdversary} lower bound.
Hence, our LCHS can be shown to be asymptotically optimal in all parameters if $\min_f\widetilde{\mathrm{deg}}_{\epsilon'}(f)=\Omega(\tau\log\frac{1}{\epsilon})$, where minimization is over all $f$ satisfying~\cref{eq:approximate_exponential}.
We compute $p_n$ by writing the optimization for $\epsilon$ over $p_n$ satisfying~\cref{eq:approximate_exponential} as a semi-infinite linear program~\cite{Hettich1993semiinfiniteLP} and in~\cref{fig:optimal_polynomial} observe an empirical scaling of $\min_f\widetilde{\mathrm{deg}}_{\epsilon'}(f)\approx0.31\tau\log\frac{e}{31\epsilon}$ when $\alpha=e$.
To our knowledge, this is the first computation of the optimal bounded polynomial for $e^{-\tau x}$ for $\tau x\ge0$ and is of independent interest.
Based on these empirical fits, our optimized LCHS $\hat{f}_{2}$ and $\hat{f}_{j,y}$ cannot be improved by more than a factor of $4$ and $3$ respectively.
\begin{figure}
    \centering
    \vspace{-0.5cm}
    \includegraphics[width=0.48\linewidth]{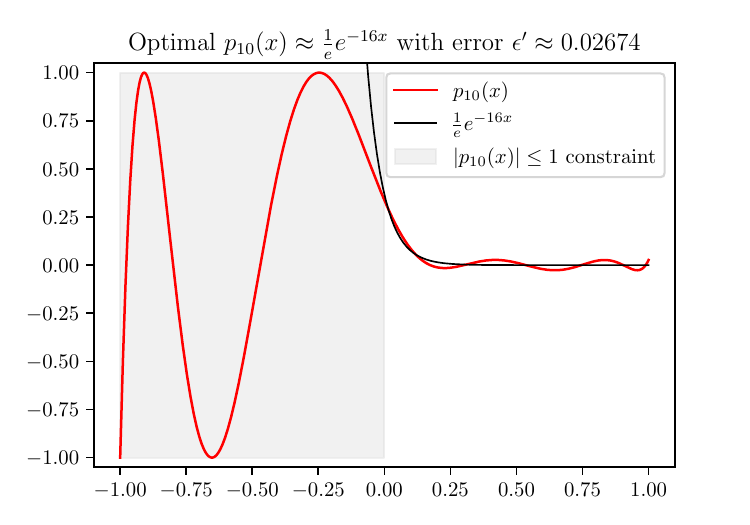}
    \includegraphics[width=0.48\linewidth]{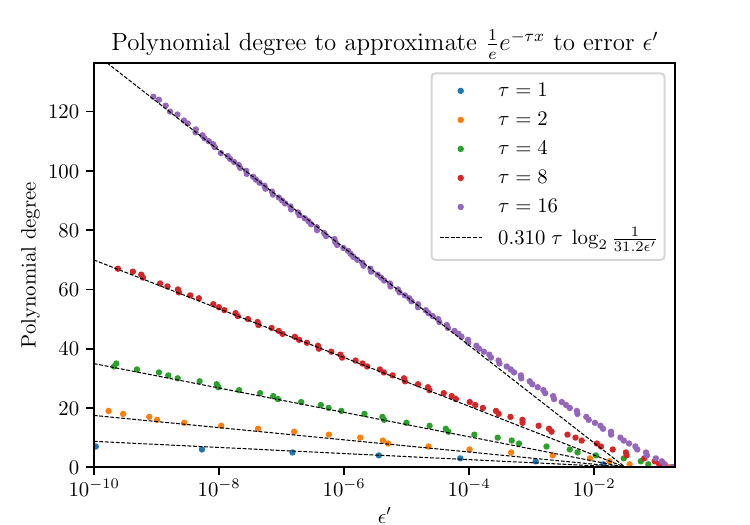}
    \caption{(Left) Plot of degree $n=10$ polynomial $p_n$ that is bounded on $[-1,1]$ and minimizes $\epsilon=\min_{p_n}\max_{x\in[0,1]}|e^{-\tau x-1}-p_n(x)|$. (Right) Plot of degree of the numerically computed $p_n$ with respect to error and $\tau$ together with a fit consistent with $\Theta(\tau\log\frac{1}{\epsilon})$ scaling.}
    \label{fig:optimal_polynomial}
\end{figure}

The rest of the manuscript is organized as follows.
In~\cref{sec:discussion}, we discuss some possible extensions of our results and conclude with some open problems.
In~\cref{sec:LCHS_main_theorem_proof}, we prove our main mathematical result~\cref{thm:LCHS_general} that a LCHS can be constructed from a kernel function that only approximates exponential decay.
In~\cref{sec:entire_approximation_to_decay}, we prove our main algorithmic result~\cref{thm:LCHS_entire} that a LCHS based on the kernel function $\hat{f}_2$ block-encodes $U_0(t)$ with the desired scaling in normalization and truncation radius of $\alpha_{\hat{f}_2}R=\mathcal{O}(\log\frac{1}{\epsilon})$. 
In~\cref{sec:uniform_quadrature}, we prove~\cref{thm:LCHS_quadrature} which shows that the integral of the kernel function may be very efficiently approximated by a uniform Riemann sum with inverse step size additive in $\|L\|_{L^1}$ and $\log\frac{1}{\epsilon}$.
In~\cref{sec:circuit_implementation}, we present an efficient quantum circuit compilation of the LCHS subroutines that proves~\cref{thm:LCHS_algorithm}.
In~\cref{sec:constant_factors}, we find best-possible LCHS formulas based on the family of kernel functions~\cref{eq:generalized_kernel_intro} by numerically optimizing~\cref{thm:LCHS_general}.
In~\cref{sec:lower_bounds}, we prove our new lower bounds and numerically compute best-possible lower bounds based on the QSVT approach.
\subsection{Discussion and conclusion}\label{sec:discussion}
\label{sec:high_level_overview}
Our work admits a number of straightforward extensions, and also highlights several interesting open problems.
\subsubsection{Schr\"odingerization}
Schr\"odingerization~\cite{Jin2024Schrodingerization} solves~\cref{eq:inhomogeneous_differential_equation} under the same conditions as LCHS by a dilation to a higher-dimensional partial differential equation.
This perspective is almost equivalent to LCHS, which we elucidate in~\cref{sec:schrodingerization}.
As both in prior art rely on the exact decay equality~\cref{eq:exponential_decay_matrix}, both benefit in the same way from the optimal kernel functions of this work.
The difference lies in details of how the block-encoding of $e^{-At}$ is implemented, and we demonstrate that the implementation of Schr\"odingerization in prior art reduces to an asymmetric block-encoding where the magnitude of kernel function coefficients by $\textsc{Prep}$ and $\overline{\textsc{Prep}}$ in~\cref{thm:LCHS_algorithm} are not equal.
In~\cref{sec:schrodingerization}, we prove that this asymmetric choice leads to a strictly worse block-encoding normalization compared to the symmetric choice of the LCHS approach.
Other kernel functions have also been proposed in the context of Schr\"odingerization.
These are based on smooth extensions to $x<0$ from exact exponential decay for $x\ge0$ by polynomial interpolation~\cite{Jin2025inhomogeneous} or infinitely-differentiable mollifiers like $\eta(x)=e^{1/(x^2-1)}$~\cite{Jin2025Schrodingerization} with sub-exponential decay $\hat{\eta}(k)=\Theta( e^{-\sqrt{|k|}})$~\cite{Johnson2015bump}.
However, these are all subject to no-go results~\cite{an2023lchsoptimal,Huang2025FourierLCHS} and cannot decay exponentially fast.

\subsubsection{Inhomogeneous case}
The inhomogeneous contribution $\int_0^tU_s(t)\vec{b}(s)\mathrm{d}s$ is obtained by a straightforward modification of the homogeneous case~\cref{thm:LCHS_algorithm} in more-or-less the same way as described in prior art~\cite{an2023lchs,an2023lchsoptimal}, which we describe for completeness.
The integral is discretized, such as by a Riemann sum over $M$ points with stepsize $\delta =t/M$ like
\begin{align}
\int_0^tU_s(t)\vec{b}(s)\mathrm{d}s\approx \delta\sum_{j=0}^{M-1}|\vec{b}(j\delta)|U_{j\delta}(t)\ket{\vec{b}(j\delta)}.
\end{align}
Then assuming a unitary state preparation oracle for the coefficients $\frac{1}{\|\vec{b}\|_{L^1}}\sum_{j=0}^{M-1}\sqrt{|\vec{b}(j\delta)|}\ket{j}$, and an oracle for the state $\ket{j}\rightarrow\ket{\vec{b}(j\delta)}$, and a controlled-block-encoding of $U_{j\delta}(t)$ by LCHS, one may apply linear-combination-of-unitaries to block-encode an approximation to $\int_0^tU_s(t)\vec{b}(s)\mathrm{d}s$.
The normalized inhomogeneous solution is then obtained by taking another linear combination of the homogeneous $U_0(t)\ket{\vec{u_0}}$ and inhomogeneous contributions $\int_0^tU_s(t)\vec{b}(s)\mathrm{d}s$ with success probability $\big(\frac{|\vec{u}(t)|}{|\vec{u}_0|+\|\vec{b}\|_{L^1}}\big)^2$.

\subsubsection{Time-dependent simulation}
Our main LCHS results~\cref{thm:LCHS_general,thm:LCHS_entire,thm:LCHS_quadrature,thm:LCHS_algorithm} are stated in full generality for time-dependent  $A(t)$.
However, the results used to justify optimal query complexity scaling required time-independent $A$.
Generally speaking, the query complexity of LCHS of a time-dependent $A(t)$ will be worse by a logarithmic factor in all parameters.
This is inherited from the suboptimal query complexity of general time-dependent Hamiltonian simulation by Hermitian $H(t)$, costing $\mathcal{O}(\alpha_H t\frac{\log(\alpha_H t/\epsilon)}{\log\log(\alpha_H t/\epsilon)})$~\cite{Low2018IntPicSim} queries to the block-encoding $\textsc{Be}[H/\alpha_H]$, compared to the time-independent complexity of $\mathcal{O}(\alpha_H t+\log\frac{1}{\epsilon})$.
In some special cases though~\cite{Mizuta2023OptimalTimeDependent}, the query complexity of time-dependent Hamiltonian simulation can be reduced to the time-independent case.

\subsubsection{$L$ with negative eigenvalues}\label{sec:negative_eigenvalues}
As noted in~\cite{an2023lchs}, any $L$ that is not positive semi-definite can be shifted by an identity term to satisfy the $L\succeq0$ requirement of LCHS.
In the general time-dependent setting, suppose that $L\succeq -l(t)$, for some {\textit{a priori}} known time-dependent function $l(t)$.
Then let $L_\text{shifted}=L+l(t)\mathcal{I}\succeq 0$, and the rescaled homogeneous solution $\vec{w}(t)=e^{-\int_0^tl(t)\mathrm{d}s}\vec{u}(t)$ satisfies
$\frac{\mathrm{d}}{\mathrm{d}t}\vec{w}(t)=-(L_\text{shifted}+H(t))\vec{w}(t)$.
Although this affects the success probability of obtaining $\ket{u(t)}$, it does not affect the query complexity of Hamiltonian simulation in the LCHS~\cref{thm:LCHS_entire} as $\mathcal{U}_\text{shifted}(s;k)=\mathcal{T}e^{-i\int_0^sk(L(s')+l(s')\mathcal{I}+H(s')\mathrm{d}s'}=e^{-ik\int_0^sl(s')\mathrm{d}s'}\mathcal{U}(s;k)$, and computation of the $k$- and $s$-dependent phase factor only affects gate complexity.
By this change of variables, we may also simulate with the same query complexity $L(t)$ with negative eigenvalues for short times if the negative semi-definite component $L_-(t)$ satisfies $\|L_-\|_{L^1}=\mathcal{O}(1)$, which can be possible even if $L(t)$ instantaneously has negative infinite eigenvalues -- a similar observation is made in concurrent work~\cite{Huang2025FourierLCHS} -- see~\cref{lem:LCHS_more_general} for a more general discussion.

\subsubsection{Optimizing constant factors for different applications} 
In a common scenario, one assumes access to the block-encodings $\textsc{Be}[L/\alpha_L]$ and $\textsc{Be}[H/\alpha_H]$.
In the time-independent case, the overall query complexity of LCHS to these block-encodings followed by amplitude amplification is  $\mathcal{O}(\alpha_{\hat{f}}((\alpha_LR+\alpha_H)t+\log\frac{1}{\epsilon}))$.
Hence our focus in this work on minimizing the query complexity prefactor $\alpha_{\hat{f}}\cdot R$ is optimal in the limit of large $R=\mathcal{O}(\log\frac{1}{\epsilon})$.
However, there will often be cases where $\epsilon$ is large and the decay term $L$ is small compared to $H$ so $\alpha_L \ll\alpha_H$. 
In this situation, one should instead numerically optimize $\alpha_{\hat{f}}(\alpha_LR+\alpha_H)$.
Alternatively, it was noted in~\cite{pocrnic2025LCHSconstants} that by writing $kL+H=\frac{\left(k-i\right)A+(\left(k-i\right)A)^\dagger}{2}$, one may block-encode $\textsc{Be}[\frac{kL+H}{\sqrt{R^2+1}\alpha_A}]$ using one query $\textsc{Be}[A/\alpha_A]$.
The most expensive term in the LCHS, time-evolution by $e^{-i(RL+H)t}$, is then simulated using $\alpha_A\sqrt{R^2+1}$ queries, which is slightly better than the naive $R+1$ queries obtained by first block-encoding $\textsc{Be}[L/\alpha_A]$ and $\textsc{Be}[H/\alpha_A]$ from $\textsc{Be}[A/\alpha_A]$. Moreover, if the initial state preparation unitary is considerably more expensive than the matrix block-encodings, one may opt to minimize $\alpha_{\hat{f}}$ alone, and in hybrid near-term quantum-classical schemes with limited depth, one should minimize $R$ alone.

\subsubsection{Fast-forwarding} In this work, we assume access to the block-encoding of the $L\succeq0$ term. 
However, previous work for Hermitian $A=L$ has also considered the stronger assumption of access to the block-encoding of $\sqrt{L}$~\cite{Chowdhury2016quantum, Gilyen2018singular, An2022Differential, Shang2024Lindblad} . 
Under this assumption, the query complexity for block-encoding $e^{-Lt}$ is known to scale like $\mathcal{O}\left(\sqrt{\alpha_L t+\log\frac{1}{\epsilon}}\sqrt{\log\frac{1}{\epsilon}}\right)$~\cite{Low2017USA}.
Unlike the Hamiltonian simulation problem for $e^{-iHt}$, this square-root dependence on $t$  for exponential decay in the square-root access model does not violate no-fast-forwarding.
A natural question is whether the LCHS approach can exploit this stronger query model to block-encode the general case $e^{-(L+iH)t}$ with the same square-root speedup with respect to $L$, for instance using $\mathcal{O}(\sqrt{\alpha_Lt\log\frac{1}{\epsilon}}+\alpha_Ht+\log\frac{1}{\epsilon})$ queries.

\subsubsection{Optimal scaling for general differential equations} Although our LCHS is optimal in the task of block-encoding $e^{-At}$, and in the number of state preparation queries, it is not known to be optimal when applied to solving differential equations. In~\cref{tab:time_independent_comparison}, we obtain the normalized solution state $\ket{\vec{u}(t)}$ to error $\epsilon'$ using $\mathcal{O}(u_\text{lchs}\alpha_At\log\frac{u_\text{lchs}}{\epsilon'})$ queries to the block-encoding of $A/\alpha_A$, and it is not known whether the $u_\text{lchs}\log u_\text{lchs}$ factor is necessary. 
We note that this same $\frac{1}{\sqrt{p}}\log \frac{1}{p}$ dependence on final state success probability $p$ also appears in different contexts, such as the current best quantum algorithms for ground-state preparation~\cite{Berry2024Rapid,King2025SOSSA}, and existing lower bounds are notably loose~\cite{Lin2020Ground,An2022Differential}.
Another natural question is whether a query complexity similar to LCHS can be achieved under more general stability conditions.

\subsubsection{Even better kernel functions} The comparison to the QSVT lower bound indicates that a further constant factor speedup of our LCHS by up to $\approx 3$ may still be possible. 
Infinitely many other kernel functions beyond those studied in this work satisfy the approximate LCHS conditions~\cref{thm:LCHS_general}, and as shown in the improvement between our examples of $\hat{f}_2$ and $\hat{f}_{j,y}$, exploration over the space of kernel functions could narrow this gap.
Our LCHS also automatically implies complexity improvements when applied to the more general problem of eigenvalue transformations~\cite{An2024EigenvalueTransformations} for functions beyond $e^{-At}$, and one could hope for further asymptotic improvements by tailoring the kernel functions to each desired function.

\section{Approximate LCHS}\label{sec:LCHS_main_theorem_proof}
In this section, we prove our main mathematical result~\cref{thm:LCHS_general} that kernel functions that satisfy certain generic analytic conditions lead to a LCHS 
\begin{align}\label{eq:LCHS_continuous}
t\ge0,\quad O(t)&\doteq\frac{1}{\sqrt{2\pi}}\int_{-R}^{R}\hat{f}(k)\mathcal{U}(t;k)\mathrm{d}k\approx U_0(t)\doteq\mathcal{T}e^{-\int_0^tA(s)\mathrm{d}s},
\end{align}
that approximates the time-ordered evolution operator to a bounded error $\epsilon$ in spectral norm that can be computed from some simple integrals over $\hat{f}(k)$.

\begin{proof}[Proof of~\cref{thm:LCHS_general}]
By~\cref{lem:time_ordered_solution_is_analytic}, $\mathcal{U}(t;z)$ is an entire function with respect to $z$. Combined with the definition of $\hat{f}(z)$, the only pole in $\hat{f}(z)\mathcal{U}(t;z)$ on the strip $S_{[-y_0,0]}$ is at $z=-i$. 
Let us choose $\Gamma$ to be a positively-oriented rectangular contour
\begin{align}
\Gamma=\lim_{x\rightarrow\infty}
\underbrace{[-x-iy_0,x-iy_0]}_{\Gamma_\text{bottom}}\cup
\underbrace{[x-iy_0,x]}_{\Gamma_\text{right}}\cup
\underbrace{[x,-x]}_{\Gamma_\text{top}}\cup
\underbrace{[-x,-x-iy_0]}_{\Gamma_\text{left}}.
\end{align}
By assumption, the residue $\mathrm{Res}_{z=-i}\hat{f}(z)=\frac{i}{\sqrt{2\pi}}$. Combined with the fact that $\mathcal{U}(t;z)$ is entire, Cauchy's residue theorem states that 
\begin{align}
\oint_{\Gamma}\frac{\hat{f}(z)}{\sqrt{2\pi}}\mathcal{U}(t;z)\mathrm{d}z = \frac{2\pi i\;\mathrm{Res}_{z=-i}\hat{f}(z)\mathcal{U}(t;z)}{\sqrt{2\pi}}=-\mathcal{U}(t;-i)
=
-\mathcal{T}e^{-\int_0^tL(s')+iH(s')\mathrm{d}s'}=-U_0(t).
\end{align}
Note that the equality $\mathcal{U}(t;-i)=U_0(t)$ is rigorous, such that through its power series definition, which is everywhere convergent by~\cref{lem:time_ordered_solution_is_analytic}.

By definition, $\lim_{x\rightarrow\pm\infty}\hat{f}(x-iy)=0$ for $y\in[0,y_0]$. By~\cref{lem:time_dependent_trotter_upper_bound}, $\|\mathcal{U}(t;z)\|\le1$ on these segments. Hence, the integrals over $\Gamma_\text{right}$ and $\Gamma_\text{left}$ vanish as follows.
\begin{align}\nonumber
\left\|\lim_{x\rightarrow\pm\infty}\int_{0}^{y_0}\hat{f}(x-iy)\mathcal{U}(t;x-iy)\mathrm{d}y\right\|&\le\lim_{x\rightarrow\pm\infty}\int_{0}^{y_0}|\hat{f}( x-iy)|\|\mathcal{U}(t;x-iy)\|\mathrm{d}y
\\\nonumber
&\le \lim_{x\rightarrow\pm\infty}\int_{0}^{y_0}|\hat{f}( x-iy)|\mathrm{d}y
\\
&
\le\lim_{x\rightarrow\pm\infty}y_0\max_{y\in[0,y_0]}|\hat{f}( x-iy)|=0.
\end{align}
The remaining non-zero integrals over the contour $\Gamma$ are over the segments $\Gamma_{\text{bottom}}$ and $\Gamma_{\text{top}}$.
\begin{align}
\int_{-R}^R\hat{f}(k)\mathcal{U}(t;k)\mathrm{d}k-{\sqrt{2\pi}}U_0(t)
&=-\int_{\mathbb{R}\backslash[-R,R]}\hat{f}(k)\mathcal{U}(t;k)\mathrm{d}k+\int_{\mathbb{R}}\hat{f}(k-iy_0)\mathcal{U}(t;k-iy_0)\mathrm{d}k.
\end{align}
Hence, using a triangle inequality and the submultiplicative property, the approximation error of the LCHS is
\begin{align}\label{eq:thm_general_error_bound_proof}
\left\|\int_{-R}^R\hat{f}(k)\mathcal{U}(t;k)\mathrm{d}k-{\sqrt{2\pi}}U_0(t)\right\|\le\int_{\mathbb{R}\backslash[-R,R]}|\hat{f}(k)|\|\mathcal{U}(t;k)\|\mathrm{d}k+\int_{\mathbb{R}}|\hat{f}(k-iy_0)|\|\mathcal{U}(t;k-iy_0)\|\mathrm{d}k.
\end{align}
As $\mathcal{U}(t;k)$ is unitary for real $k$, the spectral norm $\|\mathcal{U}(t;k)\|=1$ in the first integral. As $L\succeq0$,~\cref{lem:time_dependent_trotter_upper_bound} implies that the spectral norm $\|\mathcal{U}(t;k-iy_0)\|\le 1$ for $y_0\ge0$ in the second integral.
\end{proof}
Note that the error bound~\cref{thm:LCHS_general} can be looser than expected and can accommodate unstable dynamics. 
For instance, the following~\cref{lem:LCHS_more_general} characterizes error in terms of the smallest eigenvalue $l_0\preceq L$, which we now allow to be negative.
A lack of prior knowledge $l_0(t)$ is relevant, as we would otherwise first define a shifted $L_\text{shifted}=L-l_0\succeq0$ following~\ref{sec:negative_eigenvalues} and instead simulate $L_\text{shifted}+iH$, which gives the same solution but rescaled to have the best-possible final state norm $\vec{u}_{\text{shifted}}(t)=e^{\int_0^tl(s)\mathrm{d}s}\vec{u}(t)$. 
\begin{corollary}[Approximate LCHS with smallest $L$ eigenvalue dependence.]\label{lem:LCHS_more_general} Under the same conditions as~\cref{thm:LCHS_general}, except that $L\succeq l_0$ for any $l_0:[0,t]\rightarrow\mathbb{R}$, the LCHS $O_R(t)$ of~\cref{eq:LCHS_truncated}
satisfies
\begin{align}\label{eq:thm_general_error_bound}
\|O_R(t)-U_0(t)\|&\le\frac{1}{\sqrt{2\pi}}\int_{\mathbb{R}\backslash[-R,R]}|\hat{f}(k)|\mathrm{d}k+\frac{e^{-y_0\int_0^tl_0(s)\mathrm{d}s}}{\sqrt{2\pi}}\int_{\mathbb{R}}|\hat{f}(k-iy_0)|\mathrm{d}k.
\end{align}
\end{corollary}
\begin{proof}
The proof identical to that of~\cref{thm:LCHS_general} up to~\cref{eq:thm_general_error_bound_proof}.
As $\mathcal{U}(t;k)$ is unitary for real $k$, the spectral norm $\|\mathcal{U}(t;k)\|=1$ and the first integral of~\cref{eq:thm_general_error_bound_proof} for error of truncation to finite $R$ is unchanged. 
Let $L_\text{shifted}(t)=L(t)-l_0(t)\succeq0$.
Observe that
\begin{align}
\mathcal{U}(t;z)=\mathcal{T}e^{-i\int_0^t z(L(s)+H(s)\mathrm{d}s}
=
e^{-iz\int_0^tl_0(s)\mathrm{d}s}\mathcal{T}e^{-i\int_0^t z(L_\text{shifted}(s)+H(s)\mathrm{d}s},
\end{align}
By~\cref{lem:time_dependent_trotter_upper_bound}, $\|\mathcal{T}e^{-i\int_0^t (k-iy_0)(L_\text{shifted}(s)+H(s)\mathrm{d}s}\|\le1$ for any $y_0\ge0$.
Hence we substitute the following bound into the second integral of~\cref{eq:thm_general_error_bound_proof} 
\begin{align}
\|\mathcal{U}(t;x-iy_0)\|\le|e^{-i(x-iy_0)\int_0^tl_0(s)\mathrm{d}s}|=e^{-y_0\int_0^tl_0(s)\mathrm{d}s}.
\end{align}
\end{proof}
\cref{lem:LCHS_more_general} indicates that the error contribution of unstable dynamics where $l_0(s)<0$ is canceled by that of stable dynamics where $l_0(s)\ge0$.
However, this does not provide an efficient algorithm for simulating arbitrary unstable dynamics, such as by padding $A(s)$ to be identity $\mathcal{I}$ for $s\in(t,T]$ so that $\int_0^Tl(s)\mathrm{d}s=\int_0^tl(s)\mathrm{d}s+(T-t)=0$, as any possible padding with strictly stable dynamics reduces the final state norm by an exponential factor $e^{-(T-t)}=e^{\int_0^tl(s)\mathrm{d}s}$.

\begin{lemma}\label{lem:time_ordered_solution_is_analytic}
For all $s\in[0,t]$, let $H(s)$ and $L(s)$ be Hermitian matrices such that $\int_{0}^t\|L(s)\|\mathrm{d}s<\infty$ and $\int_{0}^t\|H(s)\|\mathrm{d}s<\infty$. Then
the time-ordered exponential $\mathcal{U}(t;z)\doteq\mathcal{T}e^{-i\int_0^t zL(s)+H(s)\mathrm{d}s}$ is an entire function with respect to $z$.
\end{lemma}
\begin{proof}
We prove $\mathcal{U}(t;z)$ is an entire function in $z$ by showing that its power series expansion in $z$ has an infinite radius of convergence. By definition, the unitary $\mathcal{U}(t;z)$ is the solution to the differential equation 
\begin{align}\label{eq:time_dependent_solution}
\frac{\mathrm{d}}{\mathrm{d}t}\mathcal{U}(t;z)=-i(z L(t)+H(t))\mathcal{U}(t;z),
\quad
\mathcal{U}(z;0)=\mathcal{I}.
\end{align}
Let the unitary $W(t)=\mathcal{T}e^{-i\int_0^tH(s)\mathrm{d}s}$ solve the differential equation 
\begin{align}
\frac{\mathrm{d}W(t)}{\mathrm{d}t}=-iH(t)W(t),
\quad
W(0)=\mathcal{I}.
\end{align}
This unitary exists as $\int_{0}^t\|H(s)\|\mathrm{d}s<\infty$.
Define the change of variables $\mathcal{U}(t;z)=W(t)V(t;z)$ for some $V(t;z)$. 
\begin{align}\nonumber
\frac{\mathrm{d}}{\mathrm{d}t}\mathcal{U}(t;z)&=\frac{\mathrm{d}}{\mathrm{d}t}(W(t)V(t;z))=\left(\frac{\mathrm{d}}{\mathrm{d}t}W(t)\right)V(t;z)+W(t)\frac{\mathrm{d}}{\mathrm{d}t}V(t;z)
\\
&=-iH(t)W(t)V(t;z)+W(t)\frac{\mathrm{d}}{\mathrm{d}t}V(t;z)
=-i(z L(t)+H(t))W(t)V(t;z).
\end{align}
By multiplying both sides by $W^{-1}(t)$, which exists as $W(t)$ is unitary, we see that $V(t;z)$ is the solution to the linear differential equation
\begin{align}
\frac{\mathrm{d}}{\mathrm{d}t}V(t;z)&=-izK(t)V(t;z),\quad K(t)=W^{-1}
(t)L(t)W(t).
\end{align}

We now develop a power series expansion of $\mathcal{U}(t;z)$ using the Dyson series representation of $V(t;z)=\mathcal{T}e^{-iz\int_0^tK(s)\mathrm{d}s}$.
\begin{align}
\mathcal{U}(t;z)=\sum_{n=0}^\infty {z^n}a_n,
\quad
a_n=\frac{W(t)V_n}{n!},
\quad
V_n=(-i)^n\int_{0}^t\mathrm{d}t_1\cdots \int_{0}^t\mathrm{d}t_n\mathcal{T}K(t_1)\cdots K(t_n).
\end{align}
As the spectral norm is unitarily invariant like $\|K(t_j)\|=\|L(t_j)\|$, each term $V_n$ has a finite norm
\begin{align}
\|V_n\|&\le\int_{0}^t\mathrm{d}t_1\cdots \int_{0}^t\mathrm{d}t_n\|L(t_1)\|\cdots\|L(t_n)\|=\left(\int_{0}^t\|L(s_1)\|\mathrm{d}s\right)^n.
\end{align}
Hence each coefficient of $z^n$ has a norm bounded by
\begin{align}
\|a_n\|=\frac{\|W(t)V_n\|}{n!}\le\frac{\|W(t)\|\|V_n\|}{n!}=\frac{\|V_n\|}{n!}\le\frac{1}{n!}\left(\int_{0}^t\|L(s_1)\|\mathrm{d}s\right)^n.
\end{align}
Using Stirling's approximation $n!=\sqrt{2\pi n}(n/e)^n(1+\Theta(1/n))$, the radius of convergence $R$ by Hadamard's formula is infinite following
\begin{align}
\frac{1}{R}=\lim_{n\rightarrow\infty}\|a_n\|^{1/n}\le\lim_{n\rightarrow\infty}\frac{\int_{0}^t\|L(s_1)\|\mathrm{d}s}{(n!)^{1/n}}=
\lim_{n\rightarrow\infty}\left(\frac{1+\Theta(1/n)}{\sqrt{2\pi n}}\right)^{1/n}\frac{e}{n}=\lim_{n\rightarrow\infty}\frac{e}{n}=0.
\end{align}
\end{proof}

\begin{lemma}\label{lem:time_dependent_trotter_upper_bound}
For all $s\in[0,t]$, let $H(s)$ and $L(s)=L_{+}(s)-L_{-}(s)$, be Hermitian matrices, where $L_{+}$ and $L_{-}$ are the positive semi-definite and negative parts of $L$ respectively. 
Then for any $t\ge0$, $z\in\mathbb{C}$, 
\begin{align}
\|\mathcal{T}e^{-i\int_0^t zL(s)+H(s)\mathrm{d}s}\|\le
e^{|\mathrm{Im}[z]|\int_{0}^{t}\|L_{\mathrm{sign}(\mathrm{Im}[z])}(s)\|\mathrm{d}s},
\quad
\mathrm{sign}(y)=\begin{cases}
+,&y\ge0,\\
-,&y<0.
\end{cases}
\end{align}
\end{lemma}
\begin{proof}
For any matrix $B:[0,t]\rightarrow\mathbb{C}^{2^n\times 2^n}$, let the time-ordered evolution operator $\mathcal{B}(t)=\mathcal{T}e^{\int_0^t B(s)\mathrm{d}s}$ be the solution to the first order differential equation
\begin{align}\label{eq:time_dependent_solution_B}
\frac{\mathrm{d}}{\mathrm{d}t}\mathcal{B}(t)=B(t)\mathcal{B}(t),
\quad
\mathcal{B}(0)=\mathcal{I}.
\end{align}
The spectral norm of $\mathcal{B}(t)$ is bounded by the exponential of the logarithmic norm~\cite{Soderlind2006LogarithmicNorm}.
\begin{align}
\left\|\mathcal{T}e^{\int_0^t B(s)\mathrm{d}s}\right\|
\le \exp\left(\int_0^t\lambda_{\text{max}}\left(\frac{B(s)+B^\dagger(s)}{2}\right)\mathrm{d}s\right).
\end{align}
Let $B(t)=-i(zL(t)+H(t))=-i(xL(t)+H(t))+yL(t)$, where $z=x+iy$. Then 
\begin{align}
\left\|\mathcal{T}e^{-i\int_0^t zL(s)+H(s)\mathrm{d}s}\right\|\le\exp\left(\int_0^t\lambda_{\text{max}}\left(yL(s)\right)\mathrm{d}s\right)=\exp\left(\int_0^t\lambda_{\text{max}}\left(yL_+(s)-yL_-(s)\right)\mathrm{d}s\right).
\end{align}
As $L_+$ and $L_-$ commute, we may readily evaluate $\lambda_\text{max}$.
\begin{align}
\lambda_{\text{max}}\left(yL_+(s)-yL_-(s)\right)\le
\begin{cases}
y\|L_+(s)\|, &y\ge0,\\
-y\|L_-(s)\|, & y< 0,
\end{cases}
\quad
=\quad |y|\|L_{\text{sign}(y)}(s)\|.
\end{align}
Hence,
\begin{align}
\left\|\mathcal{T}e^{-i\int_0^t zL(s)+H(s)\mathrm{d}s}\right\|\le\exp\left(\int_0^t|y|\|L_{\text{sign}(y)}(s)\|\mathrm{d}s\right)
=e^{|\mathrm{Im}[z]|\int_{0}^{t}\|L_{\mathrm{sign}(|\mathrm{Im}[z]|)}(s)\|\mathrm{d}s}.
\end{align}
\end{proof}

\section{Optimal scaling LCHS}\label{sec:entire_approximation_to_decay}
In this section, we prove our main algorithmic result~\cref{thm:LCHS_entire} that $\hat{f}_2(k;\gamma,c)$ leads to a LCHS that approximates $U_0(t)$ to error $\epsilon$ in spectral norm with the desired scaling $\gamma=\mathcal{O}(\frac{1}{c}\log^{1/2}\frac{1}{\epsilon})$ and $R=\mathcal{O}(\frac{1}{c}\log\frac{1}{\epsilon})$ for any $c>0$.
We start by evaluating some useful properties of the entire functions $f_{2,1}$~\cref{eq:approximate_decay_scalar}, summarized in~\cref{tab:kernels}. 
These functions satisfy
\begin{align}
\forall \epsilon>0,\quad\exists\gamma>0\quad\text{such that}\quad\max_{x\ge c}|f(x;\gamma)-e^{-x}|\le\epsilon,
\end{align}
for some constant $c$, where the variable $\gamma>0$ goes to infinity as $\epsilon$ approaches zero.
There are infinitely many other examples of functions $f(x;\gamma)$ that are entire, but not all have desirable properties in the context of LCHS. 
For instance, a previous approximation~\cite{holmes2022fluctuation} based on the complementary error function has unbounded maximum norm in the limit of infinitesimal approximation error.
\begin{table}[h]
    \centering
    \begin{tabularx}{\textwidth}{c|c|Y|c|c|c}
    \hline\hline
         Result& $f(x)$ & $\hat{f}(k)$ &$\|f\|_{L^\infty}$ & $\alpha_{\hat{f}}$ & $R$ 
         \\
         \hline
         \cite{an2023lchs}& $e^{-|x|}$&$\sqrt{\frac{2}{\pi }}\frac{1}{k^2+1}$ & $1$ & $1$ & $\cot \left(\sqrt{\frac{\pi }{8}} \epsilon \right)$ 
         \\
         \cite{an2023lchsoptimal}&\text{No closed form}&$\frac{e^{2^{\varphi}}}{\sqrt{2\pi} }\frac{1}{1-ik}e^{-(1+ik)^{\varphi}}$ & $\mathcal{O}(1)$ & $\mathcal{O}(1)$ & $\mathcal{O}(\log^{1/\varphi}\frac{1}{\epsilon})$  
         \\
         \cite{holmes2022fluctuation}&$\frac{1}{2}e^{-x}\mathrm{erfc}(\gamma-x)$&$-\frac{e^{\gamma}}{\sqrt{2\pi}}\frac{e^{-i k( \gamma  +1/2)}}{1-ik}e^{-\frac{k^2+1}{4}}$& $e^{\mathcal{O}(\gamma)}$ & $e^{\mathcal{O}(\gamma)}$ & $\mathcal{O}(\log\frac{1}{\epsilon})$ 
         \\
         \cref{lem:fourier_one_sided_decay}&$f_1(x;\gamma)$&$\frac{1}{\sqrt{2\pi}}\frac{1}{1-ik}e^{-\frac{k^{2}+1}{4\gamma^{2}}}$& $\mathcal{O}(1)$ & $\mathcal{O}(\log\gamma)$ & $\mathcal{O}(\log\frac{1}{\epsilon})$ 
         \\
        \cref{lem:fouruer_two_sided_decay}&$f_2(x;\gamma)$&$\sqrt{\frac{2}{\pi}}\frac{1}{1+k^{2}}e^{-\frac{k^{2}+1}{4\gamma^{2}}}$& $\text{erfc}\left(\frac{1}{2\gamma}\right)$ & $\text{erfc}\left(\frac{1}{2\gamma}\right)$ & $\mathcal{O}(\log\frac{1}{\epsilon})$
        \\
        \cref{cor:f3_kernel} &$f_2(x;\gamma,c)$ & $\sqrt{\frac{2}{\pi}}\frac{e^{c(1-ik)}}{1+k^{2}}e^{-\frac{k^{2}+1}{4\gamma^{2}}}$& $e^c\text{erfc}\left(\frac{1}{2\gamma}\right)$ & $e^c\text{erfc}\left(\frac{1}{2\gamma}\right)$& $\mathcal{O}(\frac{1}{c}\log\frac{1}{\epsilon})$
        \\
     \hline\hline
    \end{tabularx}
    \caption{Examples of kernel functions whose inverse Fourier transform satisfies: (First two rows) Exact exponential decay~\cref{eq:exponential_decay_matrix}; (Last four rows) Approximate exponential decay~\cref{eq:exponential_decay_approximate}. 
    $R>0$ is the truncation radius that leads to at most an error $\epsilon$ in the truncated integral $\frac{1}{\sqrt{2\pi}}\int_{\mathbb{R}\backslash[-R,R]}|\hat{f}(k)|\mathrm{d}k\le\epsilon$. 
    The approximate cases are parameterized by  $R=2c\gamma^2=\Theta(\frac{1}{c}\log\frac{1}{\epsilon})$.}
    \label{tab:kernels}
\end{table}

The first step to constructing a good approximation is the following choice that multiplies $e^{-x}$ by the complementary error function.
The complementary error function is an entire function that approximates a step function.
Hence the product is also an entire function.
\begin{lemma}[Entire approximation to truncated exponential
decay.\label{lem:fourier_one_sided_decay}]
The function
\begin{align}
f_1(x;\gamma)&\doteq\frac{1}{2}e^{- x}\text{erfc}\left(\frac{1}{2\gamma}-\gamma x\right),
\end{align}
has the Fourier transformation
\begin{align}
\hat{f}_1(k;\gamma)\doteq\mathcal{F}[f_1(x;\gamma)](k) & =\frac{1}{\sqrt{2\pi}(1-ik)}\exp\left(-\frac{k^{2}+1}{4\gamma^{2}}\right),
\;\;\text{and}\;\;
\int_\mathbb{R}|\hat{f}_1(k;\gamma)|\mathrm{d}k=\mathcal{O}(\log\gamma).
\end{align}
For any $c>0$ and any $\epsilon>0$, there exists $\gamma=\Theta(\frac{1}{c}\log^{1/2}\frac{1}{\epsilon})$ such that
\begin{align}\label{eq:f_bound3}
\forall x\ge c,\quad&0\le e^{- x}-f_1(x;\gamma)=\mathcal{O}\left(\epsilon\; e^{-\gamma^2\left(x^2-{c^2}\right)}\right),
\\
\label{eq:f_bound1}
\forall x\le-c,\quad&0\le f_1(x;\gamma)=\mathcal{O}\left(\epsilon \; e^{-\gamma^2\left(x^2-c^2\right)}\right).
\end{align}
\end{lemma}
\begin{proof}
We evaluate the Fourier transform using the definition of the complementary error function $\mathrm{erfc}(x)=\frac{2}{\sqrt{\pi}}\int_{x}^{\infty}e^{-y^{2}}\mathrm{d}y$ and integration by parts.
\begin{align}\nonumber
\hat{f}_1(k;\gamma)= 
 & =\frac{1}{\sqrt{2}\pi}\left[\left(\int_{\frac{1}{2\gamma}-\gamma x}^{\infty}e^{-y^{2}}\mathrm{d}y\right)\frac{e^{-(1-ik)x}}{-(1-ik)}\right]_{-\infty}^{\infty}+\frac{\gamma}{\sqrt{2}\pi}\int_{-\infty}^{\infty}e^{-(\frac{1}{2\gamma}-\gamma x)^{2}}\frac{e^{-(1-ik)x}}{1-ik}\mathrm{d}x\\
  & =\frac{1}{\sqrt{2\pi}}\frac{1}{1-ik}\exp\left(-\frac{k^{2}+1}{4\gamma^{2}}\right).
\end{align}
The $L^1$ norm of $\hat{f}_1$ is a standard integral
$\|\hat{f}_1(;\gamma)\|_{L^1}=\frac{1}{\sqrt{2\pi}}{e^{-\frac{1}{8 \gamma ^2}} K_0\left(\frac{1}{8 \gamma ^2}\right)}=\mathcal{O}(\log\gamma)$, where $K_0$ is the modified Bessel function of the second kind. 

For any $x\ge c$, let $x'=\gamma x-\frac{1}{2\gamma}=\Omega(\gamma c)=\Omega(\log^{1/2}\frac{1}{\epsilon})=\Omega(1)$. Then
\begin{align}\nonumber
0\le e^{- x}-f_1(x;\gamma)
&=
\frac{e^{- x}}{\sqrt{\pi}}\int_{x'}^{\infty}e^{-y^{2}}\mathrm{d}y
\le\frac{e^{- x}}{\sqrt{\pi}}\int_{x'}^{\infty}\frac{y}{x'}e^{-y^{2}}\mathrm{d}y=\frac{1}{2\sqrt{\pi}x'}e^{-\gamma^2x^2-\frac{1}{4\gamma^2}}
=\mathcal{O}\left(e^{-\gamma^2x^2}\right)
\\\label{eq:f1bound3proof}
&=\mathcal{O}\left(e^{-c^2\gamma^2-\gamma^2(x^2-c^2)}\right)\le \mathcal{O}\left(\epsilon \;e^{-\gamma^2(x^2-c^2)}\right).
\end{align}
For any $x>c$, let $x'=\frac{1}{2\gamma}+\gamma x=\Omega(\gamma c)=\Omega(\log^{1/2}\frac{1}{\epsilon})=\Omega(1)$. Then
\begin{align}\nonumber
0\le f_1(-x;\gamma)
&=
\frac{e^{x}}{\sqrt{\pi}}\int_{x'}^{\infty}e^{-y^{2}}\mathrm{d}y
\le\frac{e^{x}}{\sqrt{\pi}}\int_{x'}^{\infty}\frac{y}{x'}e^{-y^{2}}\mathrm{d}y=
\mathcal{O}\left(\frac{1}{x'}e^{-\gamma^2x^2-x-\frac{1}{4\gamma^2}}
\right)
\\
&=\mathcal{O}\left(e^{-\gamma^2x^2}\right),
\end{align}
and the proof is completed similar to~\cref{eq:f1bound3proof}.

\end{proof}

As the Fourier transform $\hat{f}_1(k;\gamma)$ has Gaussian decay with respect to $k$, a truncation error of at most $\epsilon$ is obtained by choosing $R=\mathcal{O}(\gamma\log^{1/2}\frac{1}{\epsilon})=\mathcal{O}(\log(1/\epsilon))$, however it has a one-norm that scales like $\mathcal{O}(\log(1/\epsilon))$, which is undesirable for LCHS. This is a manifestation of the Gibbs phenomenon at the diverging first derivative of $\hat{f}_1(0;\gamma)$ with respect to large $\gamma$.
To resolve this, we now consider an approximation with a continuous first derivative even in the limit $\gamma\rightarrow \infty$.
\begin{lemma}[Entire approximation to two-sided exponential
decay\label{lem:fouruer_two_sided_decay}]
The function
\begin{align}
f_{2}(x;\gamma)\doteq f_1(x;\gamma)+f_1(-x;\gamma),
\end{align}
has the Fourier transform 
\begin{align}
\hat{f}_{2}(k;\gamma)&\doteq\mathcal{F}[f_{2}(x;\gamma)](k) =\sqrt{\frac{2}{\pi}}\frac{1}{1+k^{2}}e^{-\frac{k^{2}+1}{4\gamma^{2}}},
\\
\int_\mathbb{R}|\hat{f}_2(k;\gamma)|\mathrm{d}k&=\sqrt{2\pi}\text{erfc}\left(\frac{1}{2\gamma}\right)\le\sqrt{2\pi}.
\end{align}
For any $c>0$ and any $\epsilon>0$, let $\gamma=\Theta(\frac{1}{c}\log^{1/2}\frac{1}{\epsilon})$.
 Then
\begin{align}\label{eq:f2bound1}
\forall |x|\ge c,\quad &0\le e^{-|x|}-f_2(x;\gamma)=\mathcal{O}\left(\epsilon e^{-\gamma^2\left(x^2-c^2\right)}\right),
\\\label{eq:f2bound2}
\forall x\in\mathbb{R}\quad &0\le f_2(x;\gamma)\le \mathrm{erfc}\left(\frac{1}{2\gamma}\right)\le1.
\end{align}
\end{lemma}
\begin{proof}
The expression for $\hat{f}_{2}(k;\gamma)$ is proven using
the following identity for the Fourier transform of a reflected function
\begin{align}
\mathcal{F}[f(x)](k)=\frac{1}{\sqrt{2\pi}}\int_{-\infty}^{\infty}f(x)e^{ikx}\mathrm{d}x & =-\frac{1}{\sqrt{2\pi}}\int_{\infty}^{-\infty}f(-x)e^{-ikx}\mathrm{d}x=\mathcal{F}[f(-x)](-k),
\end{align}
and~\cref{lem:fourier_one_sided_decay} for the Fourier transform of
$\hat{f}(k;\gamma)$. The integral of the absolute value 
\begin{align}\label{eq:f2_one_norm}
\int_{-\infty}^{\infty}|\hat{f}_{2}(k;\gamma)|\mathrm{d}k & =\int_{-\infty}^{\infty}\hat{f}_{2}(k;\gamma)\mathrm{d}k=\sqrt{2\pi}f_{2}(0;\gamma)=\sqrt{2\pi}\text{erfc}\left(\frac{1}{2\gamma}\right).
\end{align}
\cref{eq:f2bound1} is proven by summing the upper bounds on $f_1$~\cref{eq:f_bound3,eq:f_bound1} according to the definition of $f_2$.
By the convolution theorem, $f_2$ is proportional to the convolution of $e^{-|x|}$ and a Gaussian $e^{-x^2\gamma^2}$. As both functions have maximum at $x=0$, are symmetric, and have monotonic decay, $f_2$ is maximized at $x=0$. Together with~\cref{eq:f2_one_norm}, this proves~\cref{eq:f2bound2}.
\end{proof}

The approximations~\cref{lem:fourier_one_sided_decay,lem:fouruer_two_sided_decay} approximate $e^{-x}$ for $x\ge c$. An approximation for $x\ge0$ is thus obtained by a translation and scale by $c$ and $e^c$ respectively.
\begin{corollary}[Shifted and scaled entire approximation to two-sided exponential
decay]\label{cor:f3_kernel}
The function $f_2(x;\gamma,c)\doteq e^cf_2\left(x+c;\gamma\right)$ has Fourier transform
\begin{align}
\hat{f}_2(k;\gamma,c)&\doteq\mathcal{F}[f_2(x;\gamma,c)](k)=\sqrt{\frac{2}{\pi}}\frac{ e^{c-ikc}}{1+k^{2}}e^{-\frac{k^{2}+1}{4\gamma^{2}}},
\\
\int_\mathbb{R}|\hat{f}_2(k;\gamma,c)|\mathrm{d}k&=\sqrt{2\pi}e^c\mathrm{erfc}\left(\frac{1}{2\gamma}\right).
\end{align}
For any $c>0$, and any $\epsilon>0$, let $\gamma=\Theta(\frac{1}{c}\log^{1/2}\frac{1}{\epsilon})$. Then
\begin{align}\nonumber
\forall x\ge 0,\quad 0\le e^{-x}-f_2(x;\gamma)
&=\mathcal{O}\left(
\epsilon\;e^{-\gamma^2\left(\left(x+c\right)^2-c^2\right)}\right)
=\mathcal{O}\left(
\epsilon \;e^{-\gamma^2\left(x^2+2x\right)}\right)
=\mathcal{O}\left(
\epsilon\; e^{-\gamma^2 x^2}\right)
\\
&=\mathcal{O}\left(
\epsilon\right).
\end{align}
\end{corollary}

\begin{proof}[Proof of~\cref{thm:LCHS_entire}]
According to~\cref{thm:LCHS_general}, the error of the LCHS 
\begin{align}
O_{R,\gamma,c}(t)&\doteq\frac{1}{\sqrt{2\pi}}\int_{-R}^{R}\hat{f}_2(k;\gamma,c)\mathcal{U}(t;k)\mathrm{d}k,
\end{align}
approximates $U_0(t)\doteq\mathcal{T}e^{-\int_0^tA(s)\mathrm{d}s}$ to error
\begin{align}\label{eq:thm1_triangle_inequality}
\|O_{R ,\gamma,c}(t)-U_0(t)\|\le\frac{1}{\sqrt{2\pi}}\int_{\mathbb{R}\backslash[-R,R]}|\hat{f}_2(k;\gamma,c)|\mathrm{d}k+\frac{1}{\sqrt{2\pi}}\int_{\mathbb{R}}|\hat{f}(k-i2c\gamma^2;\gamma,c)|\mathrm{d}k.
\end{align}
In~\cref{lem:error_term_no_turncation}, we bound the integral
\begin{align}
\frac{1}{\sqrt{2\pi}}\int_{\mathbb{R}}|\hat{f}(k-i2c\gamma^2;\gamma,c)|\mathrm{d}k\le\epsilon_1    
\end{align}
for any $c>0$ and any $\gamma=\frac{1}{c}\log^{1/2}\frac{e^c}{\epsilon_1}\ge\frac{1}{\sqrt{c}}$. 
The constraint $\gamma\ge\frac{1}{\sqrt{c}}$ may be dropped as it is always satisfied following
\begin{align}
\quad \gamma=\frac{1}{c}\log^{1/2}\frac{e^c}{\epsilon_1}
\ge
\frac{1}{\sqrt{c}}
\quad
\Rightarrow
\quad
\frac{e^c}{\epsilon_1}
\ge
e^{c}
\quad
\Rightarrow
\quad{\epsilon_1}\le 1.
\end{align}
In~\cref{lem:error_term_turncated}, we bound the integral
\begin{align}
\frac{1}{\sqrt{2\pi}}\int_{\mathbb{R}\backslash[-R,R]}|\hat{f}_2(k;\gamma,c)|\mathrm{d}k\le \epsilon_2,
\end{align}
for any $c>0$ if $R=2c\gamma^2$, and $\gamma\ge\frac{1}{c}\log^{1/2}\frac{e^c}{2\pi\epsilon_2}\ge c^{-3/4}$. The $\gamma\ge c^{-3/4}$ constraint is satisfied for any $c\ge0$ and $\epsilon_2\le\frac{e^{-1/4}}{2\pi}=0.1239...$.
 
Then we may upper bound~\cref{eq:thm1_triangle_inequality} by $\epsilon_1+\epsilon_2=\epsilon$ with the choice $r=\frac{1}{2\pi}$, $\epsilon_1=\frac{1}{1+r}\epsilon$ and $\epsilon_2=\frac{r}{1+r}\epsilon$.
Then $\forall c>0$, we choose
\begin{align}
\gamma=\max\left\{\frac{1}{c}\log^{1/2}\frac{e^c}{\epsilon_1},\frac{1}{c}\log^{1/2}\frac{e^c}{2\pi\epsilon_2}\right\}
=
\frac{1}{c}\log^{1/2}\left(\left(1+\frac{1}{2\pi}\right)\frac{e^c}{\epsilon}\right),
\end{align}
for any $\epsilon=\frac{1+r}{r}\epsilon_2\le \frac{1+r}{r}\frac{e^{-1/4}}{2\pi}=(1+r)e^{-1/4}=0.9027...$.

The bound $\alpha_{\hat{f},R}\le\alpha_{\hat{f},\infty}\doteq\alpha_{\hat{f}}$ is true for any $\hat{f}$ as the integrand $|\hat{f}(k)|$ is positive, and in~\cref{cor:f3_kernel}, we found that $\alpha_{\hat{f}_2}=\frac{1}{\sqrt{2\pi}}\int_\mathbb{R}|\hat{f}_2(k;\gamma,c)|\mathrm{d}k=e^c\mathrm{erfc}\left(\frac{1}{2\gamma}\right)$, and for any $x\ge0$, $\mathrm{erfc}(x)\le1$.
\end{proof}
\begin{lemma}\label{lem:error_term_no_turncation}
For any $c>0,\gamma>\frac{1}{\sqrt{c}}$,
\begin{align}
\frac{1}{\sqrt{2\pi}}\int_{\mathbb{R}}|\hat{f}_2(k-i2c\gamma^2;\gamma,c)|\mathrm{d}k\le 
e^{c-c^2\gamma^2}\le\epsilon,
\end{align}
where the last inequality follows if $\gamma=\frac{1}{c}\log^{1/2}\frac{e^c}{\epsilon}$.
\end{lemma}
\begin{proof}
By the definition of $\hat{f}_2$, the absolute value
\begin{align}
|\hat{f}_2(z;\gamma,c)|=\sqrt{\frac{2}{\pi}}\frac{e^c}{|1+z^2|}\left|e^{-\frac{z^2+1}{4\gamma^2}-icz}\right|.
\end{align}
We apply the following inequalities to each term where $z=x-iy_0$.
\begin{align}
\frac{1}{|1+z^2|}
&=\frac{1}{\sqrt{\left(x^2+\left(y_0^2+1\right)+2 y_0\right) \left(x^2+\left(y_0^2+1\right)-2 y_0\right)}}
\le\frac{1}{x^2+\left(1-y_0\right)^2},
\\
\left|e^{-\frac{z^{2}+1}{4\gamma^{2}}-icz}\right|&
=
e^{-\frac{x^2-y_0^2+1}{4 \gamma ^2}- cy_0}.
\end{align}
We choose $y_0=2c\gamma^2$, as it maximizes the exponential decay rate. With this choice, $\left|e^{-\frac{z^{2}+1}{4\gamma^{2}}-icz}\right|=e^{-c^2 \gamma ^2-\frac{x^2+1}{4 \gamma ^2}}$. Using the assumption that $c\gamma^2\ge1$, the integral is bounded by
\begin{align}\nonumber
\frac{1}{\sqrt{2\pi}}\int_{\mathbb{R}}|\hat{f}_2(k-i2c\gamma^2;\gamma,c)|\mathrm{d}k
&\le\frac{e^{c-c^2 \gamma ^2}}{\pi}\int_{\mathbb{R}}\frac{1}{x^2+(1-2c\gamma^2)^2}e^{-\frac{x^2+1}{4 \gamma ^2}}\mathrm{d}x
\\\nonumber
&\le\frac{e^{c-c^2 \gamma ^2}}{\pi}\int_{\mathbb{R}}\frac{1}{x^2+1}e^{-\frac{x^2+1}{4 \gamma ^2}}\mathrm{d}x
\\
&=e^{c-c^2 \gamma ^2} \mathrm{erfc}\left(\frac{1}{2 \gamma }\right)=\alpha_{\hat{f}_2}e^{-c^2 \gamma ^2}\le e^{c-c^2 \gamma ^2}=\epsilon.
\end{align}
where we evaluate the integral with the identity $\mathrm{erfc}(a)=\frac{1}{\pi}\int_\mathbb{R}\frac{e^{-a^2(1+x^2)}}{1+x^2}\mathrm{d}x$ for $a\in\mathbb{R}$~\cite[Equation.7.7.1]{NISTDLMF}
and we solve for the $\gamma$ that achieves the desired $\epsilon$.
\end{proof}

\begin{lemma}\label{lem:error_term_turncated}
For any $c>0$, $\gamma>0, R>0$,
\begin{align}
\frac{1}{\sqrt{2\pi}}\int_{\mathbb{R}\backslash[-R,R]}|\hat{f}_2(k;\gamma,c)|\mathrm{d}k
<
\frac{4e^c}{\pi}\frac{\gamma^2}{R^3}e^{-\frac{R^2}{4 \gamma ^2}}\le\frac{e^c}{2\pi}e^{-c^2\gamma^2}=\epsilon,
\end{align}
where the last inequality holds if $R=2c\gamma^2$ and $\gamma\ge c^{-3/4}$, and the last equality holds if $\gamma=\frac{1}{c}\log^{1/2}\frac{e^c}{2\pi\epsilon}$ and $\epsilon\le\frac{e^{-1/4}}{2\pi}=0.1239...$.
\end{lemma}
\begin{proof}
By the definition of $\hat{f}_2$,
\begin{align}\label{eq:lem_error_term_turncated_tight}
\frac{1}{\sqrt{2\pi}}\int_{\mathbb{R}\backslash[-R,R]}|\hat{f}_2(k;\gamma,c)|\mathrm{d}k
=
\frac{2e^c}{\pi}\int_{R}^{\infty}\frac{e^{-\frac{k^2+1}{4 \gamma ^2}}}{k^2+1}\mathrm{d}k.
\end{align}
Using the following inequalities,
\begin{align}
\int_{R}^{\infty}\frac{e^{-\frac{k^2}{4\gamma^2}}}{1+k^2}\mathrm{d}k<\frac{1}{R^2}\int_{R}^{\infty}e^{-\frac{k^2}{4\gamma^2}}\mathrm{d}k
<\frac{1}{R^2}\int_{R}^{\infty}\frac{k}{R}e^{-\frac{k^2}{4\gamma^2}}\mathrm{d}k
=
2\frac{\gamma^2}{R^3}e^{-\frac{R^2}{4 \gamma ^2}},
\end{align}
we obtain
\begin{align}
\|E(t)\|&
<
\frac{4e^c}{\pi}\frac{\gamma^2}{R^3}e^{-\frac{R^2}{4 \gamma ^2}}
=\frac{e^c}{2\pi c^3\gamma^4}e^{-c^2\gamma^2}
\le \frac{e^c}{2\pi}e^{-c^2\gamma^2}=\epsilon,
\end{align}
where the first equality follows if $R=2c\gamma^2$, and the second inequality follows if $c^3\gamma^4\ge 1$.
We then solve for the $\gamma$ that achieves the desired $\epsilon$.
This choice of $\gamma$ is valid for any $\epsilon$ that satisfies the following.
\begin{align}
\gamma\ge \frac{1}{c^{3/4}} \Rightarrow\frac{e^c}{2\pi\epsilon}\ge e^{c^{1/2}}
\quad
\Rightarrow
\quad
\epsilon\le\frac{e^{c-c^{1/2}}}{2\pi}\le\frac{e^{-1/4}}{2\pi},
\end{align}
where we have used the fact that $c-c^{1/2}$ is minimized at $c=1/4$.
\end{proof}

\section{Exponentially convergent uniform quadrature}\label{sec:uniform_quadrature}
A gate-efficient implementation of LCHS requires discretizing the integral of $O_{R,\gamma,c}(t)$ to a finite number of quadrature points in a finite domain $[-R,R]$. 
In this section, we prove~\cref{thm:LCHS_quadrature} that a uniform stepsize $h>0$, such that $R/h$ is an integer, with uniform weights 
\begin{align}\label{eq:O_discretization}
I_{R,\gamma,c,h}(t)&\doteq \frac{h}{\sqrt{2\pi}}\sum_{j=-R/h}^{R/h}\hat{f}_2(hj;\gamma,c)\mathcal{U}(t;hj),
\end{align}
leads to exponential error convergence like
\begin{align}
\left\|I_{R,\gamma,c,h}(t)-U_0(t)\right\|\le\epsilon_\text{lchs}+\mathcal{O}\left(e^{\frac{1}{2}\|L\|_{L^1}-\pi/h}\right)\le\epsilon_\text{lchs}+\epsilon_\text{quad},
\end{align}
where $\epsilon_\text{lchs}$ is the additive contribution from~\cref{thm:LCHS_entire}, and $h$ is chosen to achieve the desired error $\epsilon_\text{quad}$.
We prove this using the well-known result~\cite{Trefthen2014trapezoid} of an exponentially convergent trapezoid rule that we restate in~\cref{lem:trefthen_quadrature}, which applies to summands that are analytic on a finite strip in the complex plane.
Note that~\cref{lem:trefthen_quadrature} was originally proven for scalar $f$. However, the underlying techniques (e.g. the Poisson summation formula and contour integration), generalize to functions valued in a Banach space and so the absolute value $|f|$ generalizes to a matrix norm of $\|f\|$.
\begin{lemma}[Exponentially convergent trapezoidal rule: Theorem 5.1 of~\cite{Trefthen2014trapezoid}, restated for matrix-valued functions]\label{lem:trefthen_quadrature}
Let the uniform strip $S_a=\{z:|\mathrm{Im}(z)|<a\}$ for some $a>0$. On this strip, let $f(z)$ be analytic and decay uniformly like $\lim_{|z|\rightarrow\infty}f(z)=0$.
Then for all $h>0$ and $a'\in(0,a)$,
\begin{align}
\int_{\mathbb{R}}f(x)\mathrm{d}x-h\sum_{j\in\mathbb{Z}}f(hj)=\int_{-\infty-ia'}^{\infty-ia'}\frac{f(x)}{e^{2\pi x/h}-1}\mathrm dx+\int_{-\infty+ia'}^{\infty+ia'}\frac{f(x)}{e^{2\pi x/h}-1}\mathrm dx.
\end{align}
Hence, the discretization error is upper bounded by
\begin{align}
\left\|\int_{\mathbb{R}}f(x)\mathrm{d}x-h\sum_{j\in\mathbb{Z}}f(hj)\right\|\le\frac{2M}{e^{2\pi a/h}-1},
\end{align}
for any $M$ such that
\begin{align}
\forall b\in(-a,a),\quad\int_{\mathbb{R}}\|f(x+ib)\|\mathrm{d}x\le M.
\end{align}
\end{lemma}
We may apply~\cref{lem:trefthen_quadrature} to~\cref{eq:O_discretization} as we have previously shown in~\cref{lem:time_ordered_solution_is_analytic} that $\mathcal{U}(t;z)$ is an entire function with respect to $z$, and $\hat{f}_2(z;\gamma,c)$ has poles only at $k=\pm i$.

\begin{proof}[Proof of~\cref{thm:LCHS_quadrature}]
By a triangle inequality,
\begin{align}\nonumber
    &\|I_{R,\gamma,c,h}(t)-U_0(t)\|=\|(I_{R,\gamma,c,h}(t)-I_{\infty,\gamma,c,h}(t))+(I_{\infty,\gamma,c,h}(t)-O_{\infty,\gamma,c}(t))+(O_{\infty,\gamma,c}(t)-U_0(t))\|\\\nonumber
    &\qquad\qquad\le\|I_{R,\gamma,c,h}(t)-I_{\infty,\gamma,c,h}(t)\|+\|I_{\infty,\gamma,c,h}(t)-O_{\infty,\gamma,c}(t)\|+\|O_{\infty,\gamma,c}(t)-U_0(t)\|
    \\\label{eq:thm2_triangle_inequality}
    &\qquad\qquad\le\|I_{\infty,\gamma,c,h}(t)-O_{\infty,\gamma,c}(t)\|+\frac{1}{\sqrt{2\pi}}\left(\int_{I_R}|\hat{f}_2(k;\gamma,c)|\mathrm{d}k+\int_{\mathbb{R}}|\hat{f}_2(k-iy_0;\gamma,c)|\mathrm{d}k\right),
\end{align}
where the last line follows from~\cref{lem:lchs_discretization_error_2}, which proves that
\begin{align}
\|I_{R ,\gamma,c,h}(t)-I_{\infty ,\gamma,c,h}(t)\|\le\frac{1}{\sqrt{2\pi}}\int_{I_R}|\hat{f}_2(k;\gamma,c)|\mathrm{d}k,
\end{align}
for any $\gamma\ge0$ and integer $R/h >0$, and~\cref{thm:LCHS_general}, which proves the error bound on 
\begin{align}
\|O_{\infty,\gamma,c}(t)-U_0(t)\|\le\frac{1}{\sqrt{2\pi}}\int_{\mathbb{R}}|\hat{f}_2(k-iy_0;\gamma,c)|\mathrm{d}k,
\end{align}
for any $y_0>1$.
In the proof of~\cref{thm:LCHS_entire}, we bounded the integrals in~\cref{lem:error_term_no_turncation} and~\cref{lem:error_term_turncated} to show that the last two terms of~\cref{eq:thm2_triangle_inequality} $\frac{1}{\sqrt{2\pi}}\left(\int_{I_R}|\hat{f}_2(k;\gamma,c)|\mathrm{d}k+\int_{\mathbb{R}}|\hat{f}_2(k-iy_0;\gamma,c)|\mathrm{d}k\right)\le\epsilon_{\text{lchs}}$ for any $\epsilon_{\text{lchs}}\in(0,1]$ and any $c>0$ if $R=2c\gamma^2$ and $\gamma=\frac{1}{c}\sqrt{c+\log\frac{1+1/({2\pi})}{\epsilon_{\text{lchs}}}}$.
In~\cref{lem:lchs_discretization_error}, we prove that the uniform discretization~\cref{eq:O_discretization} contributes an error
\begin{align}
\|I_{\infty ,\gamma,c,h}(t)-O_{\infty,\gamma,c}(t)\|\le\frac{4e^{\frac{1}{2}(c+\int_0^t\|L(s)\|\mathrm{d}s)}}{e^{\pi /h}-1}\le\epsilon_\text{quad},
\end{align}
which is bounded by any $\epsilon_\text{quad}\in\left(0,\frac{4}{15}\right]$ for any $c>0$ with the any choice $h\le\frac{\pi}{\|L\|_{L^1}/2+\log\left(\frac{64e^{3c/2}}
{15\epsilon_\text{quad}}\right)}$.
In~\cref{lem:block_encoding_normalization}, we prove that under the same parameter choices of $\gamma,R,h$, the block-encoding normalization
\begin{align}
|\alpha_{\hat{f}}-\alpha_{\hat{f}_2,R,h}| \le \frac{1}{1+2\pi}\epsilon_{\text{lchs}}+\frac{1}{e^{(\|L\|_{L^1}+c)/2}}\epsilon_\text{quad},
\end{align}
is essentially the same as $\alpha_{\hat{f}}$ up to a small error.
\end{proof}

We now evaluate the error from truncating the sum to a finite number of points.
\begin{lemma}\label{lem:lchs_discretization_error_2}
For all $s\in[0,t]$, let $H(s)$ and $L(s)$ be Hermitian matrices such that $\int_{0}^t\|L(s)\|\mathrm{d}s<\infty$ and $\int_{0}^t\|H(s)\|\mathrm{d}s<\infty$. For any $\gamma\ge0$ and any integer $R/h>0$,
\begin{align}
\|I_{\infty,\gamma,c,h}(t)-I_{R,\gamma,c,h}(t)\|\le\frac{1}{\sqrt{2\pi}}\int_{I_R}|\hat{f}_2(k;\gamma,c)|\mathrm{d}k,
\end{align}
where the last quantity is bounded in~\cref{lem:error_term_turncated}.
\end{lemma}
\begin{proof}
Let $I_{R,\gamma,c,h}(t)=h\sum_{j=-R/h}^{R/h} o(hj)$, where
\begin{align}
\forall k\in\mathbb{R},\quad o(k)\doteq\frac{1}{\sqrt{2\pi}}\hat{f}_2(k;\gamma,c)\mathcal{U}(t;k),
\quad
\|o(k)\|=\frac{1}{\sqrt{2\pi}}|\hat{f}_2(k;\gamma,c)|,
\end{align}
where $\|\mathcal{U}(t;k)\|=1$ for real $k$ and $t$.
Let the integer $J=R/h$. Then truncating the sum to a finite number of points $2J+1$ points introduces an addition error
\begin{align}
\left\|h\sum_{j\in\mathbb{Z}}o(hj)-h\sum_{j=-J}^{J}o(hj)\right\|
&\le
h\sum_{j\in\mathbb{Z}\backslash\{-J,-J+1\cdots,J\}}\|o(hj)\|
=
\frac{h}{\sqrt{2\pi}}\sum_{j\in\mathbb{Z}\backslash\{-J,-J+1\cdots,J\}}|\hat{f}_2(hj;\gamma,c)|.
\end{align}
Observe that $|\hat{f}_2(k;\gamma,c)|\propto \frac{1}{1+k^2}e^{-\frac{k^2}{4\gamma^2}}$ is monotonically decreasing with respect to $|k|$ for real $k$, as its the sign of its derivative
\begin{align}
\frac{\mathrm d}{\mathrm{d}k}|\hat{f}_2(k;\gamma,c)|\propto-k\cdot\left(\frac{4 \gamma ^2+k^2+1}{2 \gamma ^2 \left(k^2+1\right)^2}e^{-\frac{k^2}{4 \gamma ^2}}\right)
\end{align}
depends only on the sign of $k$.
Hence
\begin{align}
h\sum_{j\in\mathbb{Z}\backslash\{-J,-J+1\cdots,J\}}|\hat{f}_2(hj;\gamma,c)|
&\le\int_{\mathbb{R}\backslash[-hJ,hJ]}|\hat{f}_2(k;\gamma,c)|\mathrm{d}k
=\int_{\mathbb{R}\backslash[-R,R]}|\hat{f}_2(k;\gamma,c)|\mathrm{d}k.
\end{align}
\end{proof}

\begin{lemma}\label{lem:lchs_discretization_error}
For all $s\in[0,t]$, let $H(s)$ and $L(s)$ be Hermitian matrices such that $\int_{0}^t\|L(s)\|\mathrm{d}s<\infty$ and $\int_{0}^t\|H(s)\|\mathrm{d}s<\infty$. For any $c>0$, any $\gamma>0$, and any $\epsilon_\text{quad}\in(0,1]$,
\begin{align}
\|O_{\infty,\gamma,c}(t)-I_{\infty,\gamma,c,h}(t)\|\le\frac{4e^{\|L\|_{L^1}/2+3c/2}}{e^{\pi /h}-1}\le\frac{64e^{\|L\|_{L^1}/2+3c/2}}{15e^{\pi /h}}\le\epsilon_\text{quad},
\end{align}
where the second last inequality holds if $h\le\frac{\pi}{\log 16}$ and the third inequality holds if $
h=\frac{\pi}{\|L\|_{L^1}/2+\log\left(\frac{64e^{3c/2}}
{15\epsilon_\text{quad}}\right)}$ and $\epsilon_\text{quad}\le\frac{4}{15}=0.266...$.
\end{lemma}
\begin{proof}
Let $I_{\infty,\gamma,c,h}(t)=h\sum_j o(hj)$, where
\begin{align}
o(z)\doteq\frac{1}{\sqrt{2\pi}}\hat{f}_2(z;\gamma,c)\mathcal{U}(t;z)=\frac{1}{\pi}\cdot\frac{e^{c}}{1+z^2}\cdot e^{-\frac{z^2+1}{4\gamma^2}-icz}\cdot\mathcal{U}(t;z).
\end{align}
In~\cref{lem:time_ordered_solution_is_analytic}, we proved that $\mathcal{U}(t;z)$ is an entire function with respect to $z$ if $\int_{0}^t\|L(s)\|\mathrm{d}s<\infty$ and $\int_{0}^t\|H(s)\|\mathrm{d}s<\infty$. Hence, $o(z)$ has poles only at $\pm i$. In other words, $o(z)$
is analytic on the strip $S_1$.
We apply the following inequalities to each term in $o(z)$, which hold for any ${z=x+iy\in S_{1}}$.
\begin{align}
\left|\frac{e^{c}}{1+z^2}\right|&=\frac{e^c}{|1+z^2|}
=\frac{e^c}{\sqrt{\left(x^2+\left(y^2+1\right)+2 y\right) \left(x^2+\left(y^2+1\right)-2 y\right)}}
\le\frac{e^c}{x^2+\left(1-y\right)^2},
\\
\left|e^{-\frac{z^{2}+1}{4\gamma^{2}}-icz}\right|&
=
e^{-\frac{x^2-y^2+1}{4 \gamma ^2}- cy},
\\
\|\mathcal{U}(t;z)\|&\le e^{|y|\|L\|_{L^1}}, \quad\text{\cref{lem:time_dependent_trotter_upper_bound}}.
\end{align}
These imply uniform convergence of $\lim_{x\rightarrow\pm\infty }o(x+iy)=0$ for any finite $y$.
We now evaluate the error pre-factor $M$ of~\cref{lem:trefthen_quadrature}. For any $y\in(-1,1)$,
\begin{align}\nonumber
\int_{\mathbb{R}}\left\|o(x+iy)\right\|\mathrm{d}x
&\le\frac{1}{\pi}\int_{\mathbb{R}}\left|\frac{e^{c}}{1+z^2}\right|\left|e^{-\frac{z^{2}+1}{4\gamma^{2}}-icz}\right|\|\mathcal{U}(t;z)\|\mathrm{d}x
\\\nonumber
&\le\frac{e^{|y|\|L\|_{L^1}}}{\pi}\int_{\mathbb{R}}\left|\frac{e^{c}}{1+z^2}\right|\left|e^{-\frac{z^{2}+1}{4\gamma^{2}}-icz}\right|\mathrm{d}x
\\\nonumber
&
=
\frac{e^{|y|\|L\|_{L^1}+c(1-y)}}{\pi}\int_{\mathbb{R}}\frac{e^{-\frac{x^2-y^2+1}{4 \gamma ^2}}}{x^2+\left(1-y\right)^2}\mathrm{d}x
\\\nonumber
&=\frac{e^{|y|\|L\|_{L^1}+c(1-y)}}{1-y} e^{-\frac{y \left(1-y\right)}{2 \gamma ^2}}\text{erfc}\left(\frac{1-y}{2 \gamma }\right)\\
&\le\frac{e^{|y|\|L\|_{L^1}+c(1-y)}}{1-y},
\end{align}
where the last equality follows from the identity $\mathrm{erfc}(a)=\frac{1}{\pi}e^{-a^2}\int_\mathbb{R}\frac{e^{-a^2x^2}}{1+x^2}$~\cite[Equation.7.7.1]{NISTDLMF}.
Hence the maximum over $y\in(-a,a)$ for $a\in(0,1)$ is bounded by
\begin{align}
\max_{y\in(-a,a)}\int_{\mathbb{R}}\left\|o(x+iy)\right\|\mathrm{d}x
\le
\max_{y\in(-a,a)}\frac{e^{|y|\|L\|_{L^1}+c(1-y)}}{1-y}
<\frac{e^{|a|\|L\|_{L^1}+c(1+a)}}{1-a}
=M_a
\end{align}
For our proof, it suffices to consider the strip $S_{1/2}$. Hence
\begin{align}
   M_{1/2}=2e^{(\|L\|_{L^1}+3c)/2}.
\end{align}
Now apply~\cref{lem:trefthen_quadrature} to obtain the quadrature error
\begin{align}\label{eq:lem:lchs_discretization_error_quad}
\left\|\int_{\mathbb{R}}o(x)\mathrm{d}x-h\sum_{j\in\mathbb{Z}}o(hj)\right\|\le\left.\frac{2M_a}{e^{2\pi a/h}-1}\right|_{a=1/2}=\frac{4e^{\|L\|_{L^1}/2+3c/2}}{e^{\pi /h}-1}=\epsilon.
\end{align}
For any $x\ge \log X>0$, observe that $e^x-1\ge\frac{X-1}{X}e^x$. By choosing that $h\le\frac{\pi}{\log16}$, the denominator $e^{\pi/h}-1\ge\frac{15}{16}e^{\pi/h}$. Hence, the error from discretizing the integral is at most $\epsilon$ with the choice
\begin{align}
h&=\frac{\pi}{\log\left(\frac{64e^{\|L\|_{L^1}/2+3c/2}}{15\epsilon}\right)}
=
\frac{\pi}{\|L\|_{L^1}/2+\log\left(\frac{64e^{3c/2}}{15\epsilon}\right)}.
\end{align}
For any $c>0$, the inequality $h\le\frac{\pi}{\log16}$ is satisfied following
\begin{align}
h=\frac{\pi}{\|L\|_{L^1}/2+\log\left(\frac{64e^{3c/2}}{15\epsilon}\right)}
\le
\frac{\pi}{\log\left(\frac{64e^{3c/2}}{15\epsilon}\right)}
\le
\frac{\pi}{\log\left(\frac{64}{15\epsilon}\right)}\le\frac{\pi}{\log16},
\end{align}
if $\epsilon\le\frac{64}{15\cdot16}=\frac{4}{15}=0.266...$.
\end{proof}

\begin{lemma}\label{lem:block_encoding_normalization}
For any $c>0$, any $\epsilon_{\text{lchs}}\in(0,1]$, any $\epsilon_\text{quad}\in(0,4/15]$ and any $\|L\|_{L^1}\ge0$, choose  $\gamma=\frac{1}{c}\sqrt{c+\log\frac{1+1/({2\pi})}{\epsilon_{\text{lchs}}}}$, $R=2c\gamma^2$, and any $h\le
\frac{\pi}{\|L\|_{L^1}/2+\log\left(\frac{64e^{3c/2}}{15\epsilon_\text{quad}}\right)}$ such that $R/h$ is an integer, and define
\begin{align}
\alpha_{\hat{f}_2}&\doteq\frac{1}{\sqrt{2\pi}}\int_{\mathbb{R}}|\hat{f}_2(k;\gamma,c)|\mathrm{d}k,
\quad
\alpha_{\hat{f}_2,R,h}\doteq\frac{h}{\sqrt{2\pi}}\sum_{j=-R/h}^{R/h}|\hat{f}_2(hj;\gamma,c)|.
\end{align}
Then
\begin{align}
|\alpha_{\hat{f}_2}-\alpha_{\hat{f}_2,R,h}|\le \frac{1}{1+2\pi}\epsilon_{\text{lchs}}+\frac{1}{e^{(\|L\|_{L^1}+c)/2}}\epsilon_\text{quad}.
\end{align}
\end{lemma}
\begin{proof}
Let us rewrite the normalization factors as
\begin{align}
\alpha_{\hat{f}_2,R,h}=\frac{h e^c}{\sqrt{2\pi}}\sum_{j=-R/h}^{R/h}\hat{f}_2(hj,\gamma,0),
\quad
\alpha_{\hat{f}_2}=\frac{e^c}{\sqrt{2\pi}}\int_{\mathbb{R}}\hat{f}_2(k;\gamma,0)\mathrm{d}k.
\end{align}
Hence, the difference $|\alpha_{\hat{f}_2}-\alpha_{\hat{f}_2,R,h}|$ is the error of approximating the integral over $\hat{f}_2(k;\gamma,0)$ by a uniform grid $k=hj$ and truncating to finite $R$.
By a triangle inequality,
\begin{align}
|\alpha_{\hat{f}_2,R,h}-\alpha_{\hat{f}_2}|\le
\left|\alpha_{\hat{f}_2,R,h}-\alpha_{\hat{f}_2,\infty,h}\right|+
\left|\alpha_{\hat{f}_2,\infty,h}-\alpha_{\hat{f}_2}\right|.
\end{align}
By~\cref{lem:lchs_discretization_error_2} in the case $\mathcal{U}(t;k)=1$, 
the first term
\begin{align}\nonumber
\left|\alpha_{\hat{f}_2,R,h}-\alpha_{\hat{f}_2,\infty,h}\right|=e^c\|I_{\infty,\gamma,0,h}(t)-I_{R,\gamma,0,h}(t)\|\le\frac{e^c}{\sqrt{2\pi}}\int_{I_R}|\hat{f}_2(k;\gamma,0)|\mathrm{d}k=\frac{1}{\sqrt{2\pi}}\int_{I_R}|\hat{f}_2(k;\gamma,c)|\mathrm{d}k.
\end{align}
By~\cref{lem:error_term_turncated}, this is further bounded by
\begin{align}
\frac{1}{\sqrt{2\pi}}\int_{I_R}|\hat{f}_2(k;\gamma,c)|\mathrm{d}k\le\frac{e^c}{2\pi}e^{-c^2\gamma^2}=\frac{1}{1+2\pi}\epsilon_\text{lchs}.
\end{align}
By~\cref{lem:lchs_discretization_error} in the case $H(s)=L(s)=0$, the second term
\begin{align}\nonumber
\left|\alpha_{\hat{f}_2,\infty,h}-\alpha_{\hat{f}_2}\right|&=e^c\|O_{\infty,\gamma,0}(t)-I_{\infty,\gamma,0,h}(t)\|\le\frac{4e^{c}}{e^{\pi /h}-1}\le\frac{64e^c}{15e^{\pi/h}}\le\frac{64e^c}{15e^{\|L\|_{L^1}/2+\log\left(\frac{64e^{3c/2}}{15\epsilon_\text{quad}}\right)}}
\\
&\le\frac{1}{e^{\|L\|_{L^1}/2+c/2}}\epsilon_\text{quad}.
\end{align}
\end{proof}

\section{Gate-efficient circuit implementation}\label{sec:circuit_implementation}
In this section, we prove our main result on an explicit LCHS quantum circuit implementation~\cref{thm:LCHS_algorithm} for block-encoding
\begin{align}
I_{R,\gamma,c,h}&\doteq \frac{h}{\sqrt{2\pi}}\sum_{j=-R/h}^{R/h }\hat{f}_2(hj;\gamma,c)\mathcal{U}(t;hj).
\end{align} 
A standard block-encoding construction is linear-combination-of-unitaries \cref{lem:lcu}. 
In prior art~\cite{an2023lchsoptimal}, it was already known that the query complexity of this block-encoding is no more than that of Hamiltonian simulation of $\mathcal{U}(t;\pm R)$. 
For completeness, we review this construction, which is used to synthesize the $\textsc{Sel}$ unitary, and then provide an explicit circuit construction of $\textsc{Prep}$ for our kernel function $\hat{f}_2$.
\begin{lemma}[Linear combination of unitaries~\cite{Berry2015Truncated,Low2016Qubitization}]\label{lem:lcu}
Let the operator $A=\sum_{j=0}^{M-1}a_jU_j$ be a linear combination of $M$ unitaries $U_j$ with complex coefficients $\vec{a}\in\mathbb{C}^M$. Define the unitaries
\begin{align}
\textsc{Prep}_{\ket{\vec{a}}}\ket{0}&=\sum_{j=0}^{M-1}\frac{\sqrt{|a_j|}}{\sqrt{|\vec{a}|_1}}\ket{j},
\quad
\overline{\textsc{Prep}}_{\ket{\vec{a}}}\ket{0}=\sum_{j=0}^{M-1}e^{i\mathrm{Arg}(a_j)}\frac{\sqrt{|a_j|}}{\sqrt{|\vec{a}|_1}}\ket{j},
\quad
\textsc{Sel}=\sum_{j=0}^{M-1}\ket{j}\bra{j}\otimes U_j,
\end{align}
where $|\vec{a}|_1=\sum_{j=0}^{M}|a_j|$. Then $A$ is block-encoded by
\begin{align}
\textsc{Be}\left[\frac{A}{|\vec{a}|_1}\right]=(\textsc{Prep}^\dagger\otimes \mathcal{I})\cdot\textsc{Sel} \cdot(\overline{\textsc{Prep}}\otimes \mathcal{I}).
\end{align}
\end{lemma}
In the following we represent the computational basis state $\ket{j}=\ket{|j|}_\text{abs}\ket{\mathrm{sign}[j]}_{\text{sgn}}$ as  an $n$-qubit register where $n=n_0+1$, $n_0=\lceil\log_2(\frac{R}{h}+1)\rceil$, where the first $n_0$ qubits store the binary representation of an integer $|j|=0,1,\cdots,\frac{R}{h}$, and the last qubit stores the sign $\mathrm{s}[j]\doteq\mathrm{sign}[j]$ of $j$.

In the query setting, the time-ordered evolution operator is constructed from either the block-encoding $\textsc{Be}[A/\alpha_A]$ in the time-independent case or the time-dependent block-encoding $\textsc{HamT}$~\cref{def:block_encoding_td}. 
Despite the additional controls by $j$, the query complexity of $\textsc{Sel}$~\cref{eq:I_selelect} is no more than that for $U_0(t)$.
\begin{definition}[Block-encoding of time-dependent matrices~\cite{Low2018IntPicSim}]
\label{def:block_encoding_td}
In the time-dependent case, $A(s):[0,t]\rightarrow\mathbb{C}^{2^n\times 2^n}$, is block-encoded with normalization factor $\alpha_A\ge\|A\|_{L^{\infty}}\doteq\max_{s\in[0,t]}\|A(s)\|$ and stepsize $h>0$, where $t/h\in\mathbb{Z}_{\ge0}$, by the unitary operator
\begin{align}
\textsc{Be}\left[\frac{A}{\alpha_A}\right]\doteq\sum_{j=0}^{t/h}\ket{j}\bra{j}\otimes\textsc{Be}\left[\frac{A(hj)}{\alpha_A}\right]\in\mathbb{C}^{2^{n+a+\lceil\log_2(t/h+1)\rceil}\times2^{n+a+\lceil\log_2(t/h+1)\rceil}}.
\end{align}
\end{definition}

\begin{proof}[Proof of~\cref{thm:LCHS_algorithm}]
We define the following unitary operators
\begin{align}
\label{eq:I_prep}
\textsc{Prep}_{\ket{\vec{a}}}\ket{0}&=\ket{\psi_{R,h,\gamma}}\doteq\sum_{j=-R/h}^{R/h}\frac{\sqrt{|a_j|}}{\sqrt{|\vec{a}|_1}}\ket{j},\quad
a_j\doteq h\frac{\hat{f}_2(hj;\gamma,c)}{\sqrt{2\pi}},
\\
\overline{\textsc{Prep}}_{\ket{\vec{a}}}\ket{0}&=
\overline{\ket{\psi_{R,h,\gamma}}}\doteq\sum_{j=-R/h}^{R/h}\frac{e^{-ikc}\sqrt{|a_j|}}{\sqrt{|\vec{a}|_1}}\ket{j},
\\\label{eq:I_selelect}
\textsc{Sel}&=\sum_{j=-R/h}^{R/h}\ket{j}\bra{j}\otimes \mathcal{U}(t;hj).
\end{align}
By linear-combination-of-unitaries, \cref{lem:lcu}, these block-encode $\textsc{Be}\left[\frac{I_h}{|\vec{a}|_1}\right]$. 
By definition, note that $|\vec{a}|_1=\alpha_{\hat{f},R,h}$, and by~\cref{thm:LCHS_quadrature}, $\alpha_{\hat{f}_2,R,h}=\alpha_{\hat{f}_2}+\mathcal{O}(\epsilon)$. 
The small difference between $\alpha_{\hat{f}_2,R,h}$ and $\alpha_{\hat{f}_2}$ only changes the block-encoding normalization of $I_h$ by a small factor.
We use the following lemmas to prove the claims gate and query complexity of this block-encoding.
\begin{itemize}
\item In~\cref{lem:kernel_superposition_state}, we prove that the unitaries $\textsc{Prep}_{\ket{\vec{a}}}$ and $\overline{\textsc{Prep}_{\ket{\vec{a}}}}$ may be approximated by a unitary quantum circuit to error $\epsilon_2$ for any $\epsilon_2>0$, any $c>0$, any $h>0$, any $\gamma>0$, and any $R=\Omega(\gamma\log^{1/2}\frac{1}{\epsilon_2})$, using $G=\mathcal{O}(R^{3/2}\log\frac{R}{\epsilon_2}\log\frac{R}{h})$ two-qubit gates.
\item In~\cref{lem:BE_sel_mul}, we prove that the query complexity $Q$ and gate complexity of $\textsc{Sel}$ is no more than that of $\mathcal{U}(t;hj)$, plus $\mathcal{O}(Q\log\frac{R}{h})$ two-qubit gates
\end{itemize}

Let us choose all errors $\epsilon_j=\mathcal{O}(\epsilon)$. 
In~\cref{thm:LCHS_entire,thm:LCHS_quadrature}, we chose $\gamma=\Omega(\frac{1}{c}\log^{1/2}\frac{1}{\epsilon_{\text{lchs}}})$, $R=\Omega(\frac{1}{c}\log\frac{1}{\epsilon_{\text{lchs}}})$, and $h^{-1}=\mathcal{O}(\|L\|_{L^1}+c+\log\frac{1}{\epsilon_\text{quad}})$. Assume that $c$ is a constant. Then under the above conditions, the total gate complexity is
\begin{align}\nonumber
&\mathcal{O}\left(\log^{3/2}\frac{1}{\epsilon}\log\frac{1}{\epsilon}\log\frac{\log\frac{1}{\epsilon}}{h}\right)+\mathcal{O}\left(Q\log\frac{\log\frac{1}{\epsilon}}{h}\right)
\\\nonumber
&=\mathcal{O}\left[\left(\log^{5/2}\frac{1}{\epsilon}+Q\right)\left(\log\frac{\log\frac{1}{\epsilon}}{h}\right)\right]
\\\nonumber
&=\mathcal{O}\left[\left(\log^{5/2}\frac{1}{\epsilon}+Q\right)\left(\log\left(\|L\|_{L^1}+\log\frac{1}{\epsilon}\right)+\log\log\frac{1}{\epsilon})\right)\right]
\\
&=\mathcal{O}\left[\left(\log^{5/2}\frac{1}{\epsilon}+Q\right)\left(\log\left(\|L\|_{L^1}+\log\frac{1}{\epsilon}\right)\right)\right].
\end{align}
\end{proof}

Above, we assumed block-encoding access of $L(s)$ and $H(s)$ separately. This is without loss of generality as these may always be constructed from the block-encoding of $A$, such as by taking a linear combination of unitaries $\frac{1}{2}(\textsc{Be}[A/\alpha_A]\pm \textsc{Be}[A/\alpha_A])^\dagger$.

\subsection{Multiplexed time-evolution operator}
Hamiltonian simulation algorithms for the time-evolution operator $U_{0}(t;k)\doteq\mathcal{T}e^{-i\int_0^tkL(s)+H(s)\mathrm{d}s}$ in either the time-independent or time-dependent case by quantum signal processing~\cite{Low2016HamSim} or by a truncated Dyson series~\cite{Low2018IntPicSim}  query a unitary $O$ that is a block-encoding $\textsc{Be}[(kL+H)/\alpha]$ or $\textsc{Be}[(kL(s)+H(s))/\alpha]$  respectively, for some block-encoding normalization $\alpha$.
These quantum algorithms intersperse oracle queries with arbitrary unitaries $V_j$ like
\begin{align}
U_{0}(t;k)=V_1\cdot O\cdot V_2\cdot O\cdot V_3\cdot O\cdots,
\end{align}
where the $V_j$ depend on the {\textit{a priori}} known constant $\alpha$, but not $kL(s)+H(s)$. Hence, replacing the oracle queries with controlled queries like $\textsc{Mul}=\sum_{j}\ket{j}\bra{j}\otimes \textsc{Be}[(k_jL(s)+H(s))/\alpha_A]$ immediately gives controlled time evolution $\textsc{Sel}=\sum_{j}\ket{j}\bra{j}\otimes U_{0}(t;k_j)$ such as in the following, and similarly for the time-dependent case.
\begin{lemma}\label{lem:BE_sel_mul}
There is a quantum circuit that block-encodes $\textsc{Be}[\textsc{Sel}']$ such that $\|\textsc{Sel}'-\textsc{Sel}\|\le\epsilon$ using $Q=\mathcal{O}((R\alpha_L+\alpha_H)t+\log\frac{1}{\epsilon})$ queries to the controlled block-encodings $\textsc{Be}[L/\alpha_L], \textsc{Be}[H/\alpha_H]\in\mathbb{C}^{2^{m}\times2^{m}}$, where $L,H\in\mathbb{C}^{2^{n}\times 2^{n}}$, and $\mathcal{O}(Q(m-n)+Q\log\frac{R}{h}))$ arbitrary two-qubit gates.
\end{lemma}
\begin{proof}
\cref{lem:mul_circuit} constructs the multiplexed block-encoding $\textsc{Mul}=\sum_{j=-R/h}^{R/h}\ket{j}\bra{j}\otimes\textsc{Be}\left[\frac{hjL+H}{\alpha}\right]$, where $\alpha=R\alpha_L+\alpha_H$ and $R/h\in\mathbb{Z}_{+}$ using one controlled query to the block-encodings of $L$ and $H$, and $\mathcal{O}(\log\frac{R}{h})$ two-qubit gates. By querying controlled-$\textsc{Mul}$ $\mathcal{O}(\alpha t+\log\frac{1}{\epsilon})$ times, then for each $j$, Hamiltonian simulation~\cref{lem:ham_sim_qsp} block-encodes $\textsc{Be}[U_j]\in\mathbb{C}^{2^{m+1}\times2^{m+1}}$ an operator $U_j$ such that $\|U_j-e^{-i(hjL+H)t}\|\le\epsilon$. Hence, we choose
\begin{align}
\textsc{Sel}'=\sum_{j=-R/h}^{R/h}\ket{j}\bra{j}\otimes\textsc{Be}[U_j].
\end{align}
The gate complexity of $\textsc{Sel}'$ follows from $Q$ multiplied by the gate complexity of $\textsc{Mul}$ and the $\mathcal{O}(m-n)$ two-qubit gates needed by Hamiltonian simulation.
\end{proof}

\begin{lemma}[Hamiltonian simulation by qubitization and quantum signal processing~\cite{Low2016HamSim,Low2016Qubitization,Berry2024HamSim}]\label{lem:ham_sim_qsp}
For any $n$-qubit Hermitian operator $H\in\mathbb{C}^{2^n\times 2^n}$, and any $\epsilon,t>0$, there is an operator $U$ that approximates the time-evolution operator $e^{-iHt}$ to error $\|U-e^{-iHt}\|\le\epsilon$ such that the quantum circuit block-encoding $\textsc{Be}[U]\in\mathbb{C}^{2^{m+1}\times 2^{m+1}}$ makes $Q=\mathcal{O}(\alpha_H t+\log\frac{1}{\epsilon})$ queries to the controlled block-encoding $\textsc{Be}[H/\alpha_H]\in\mathbb{C}^{m}$, and uses $\mathcal{O}(Q(m-n))$ arbitrary single-qubit gates. 
\end{lemma}

We now give an explicit construction of $\textsc{Mul}$.

\begin{lemma}\label{lem:mul_circuit}
For any integer $M=R/h>0$, the quantum circuit for the multiplexed block-encoding
\begin{align}
\textsc{Mul}\doteq\sum_{j=-R/h}^{R/h}\ket{j}\bra{j}\otimes\textsc{Be}\left[\frac{hjL+H}{R\alpha_L+\alpha_H}\right],\quad \textsc{Be}\left[\frac{hjL+H}{R\alpha_L+\alpha_H}\right]\in\mathbb{C}^{2^{m+b+2}\times 2^{m+b+2}},
\end{align}
where $b=\lceil\log_2(M+1)\rceil$,
uses $1$ query to controlled-$\textsc{Be}[L/\alpha_L]\in\mathbb{C}^{2^m\times 2^m}$ and controlled--$\textsc{Be}[H/\alpha_H]\in\mathbb{C}^{2^m\times 2^m}$, and $\mathcal{O}(\log\frac{R}{h})$ two-qubit gates.
\end{lemma}
\begin{proof}
Define the unitary $U_\text{r}$ to have the following action for any computational basis state $j\in[-M,M]$.
\begin{align}\label{eq:U_r}
U_\text{r}\ket{|j|}\ket{0}_c\ket{0}_\text{gb}=\ket{|j|}\left(\sqrt{\frac{|j|}{M}}\ket{0}_c\ket{\text{gb}_{|j|}}_\text{gb}+\cdots\ket{1}_c\right),
\end{align}
where $\ket{\text{gb}_j}$ is an arbitrary quantum state that depends on $j$. 
$U_\text{r}$ can be realized by black-box state preparation techniques~\cite{Sanders2018,Low2018SpectralNorm} as follows.
\begin{enumerate}
    \item Let $\textsc{Uni}\ket{0}=\ket{\text{u}}$ prepare the uniform superposition state $\ket{\text{u}}=\frac{1}{\sqrt{M}}\sum_{x=1}^{M}\ket{x}$ using $\mathcal{O}(b)$ two qubit gates.
    \item On the state $\ket{|j|}\ket{\text{u}}$, compute in a single ancilla qubit the comparison $\ket{x>|j|}_c$ to obtain
    \begin{align}\nonumber
    \textsc{Comp}\ket{|j|}\ket{0}_c\ket{\text{u}}_\text{gb}&=\frac{1}{\sqrt{M}}\ket{|j|}\left(\ket{0}_c\sum_{x=1}^{|j|}\ket{x}+\ket{1}_c\sum_{x=|j|+1}^{M}\ket{x}\right)\\
    &=\ket{|j|}\left(\frac{\sqrt{|j|}}{\sqrt{M}}\ket{0}_c\ket{\text{gb}_j}_\text{gb}+\cdots\ket{1}_c\right),
    \end{align}
    where $\ket{\text{gb}_j}_\text{gb}$ is some normalized quantum state.
    This comparison costs $\mathcal{O}(b)$ two-qubit gates~\cite{Takahashi2010Adder}.
    \item Hence, $U_\text{r}=\mathcal{I}\otimes\textsc{Uni}$ costs $\mathcal{O}(b)$ two-qubit gates and acts on $2b+1$ qubits.
\end{enumerate}
Let us define a few more unitaries
\begin{enumerate}
    \item The controlled version $\textsc{CU}_\text{r}=\ket{0}\bra{0}_\alpha\otimes U_\text{r}+\ket{1}\bra{1}_\alpha$, which has the same gate complexity.
    \item A state preparation unitary that prepares in the single-qubit register $\ket{\cdot}_\alpha$ the state
\begin{align}
U_{\alpha}\ket{0}_\alpha=\frac{\sqrt{R\alpha_L}}{\sqrt{\alpha}}\ket{0}_\alpha+\frac{\sqrt{\alpha_H}}{\sqrt{\alpha}}\ket{1}_\alpha,\quad\alpha=R\alpha_L+\alpha_H.
\end{align}
\item 
the block-encoding select unitary
\begin{align}
\textsc{Sel}_{L,H}\doteq\ket{0}\bra{0}_\alpha\otimes\textsc{Be}[L/\alpha_L]+\ket{1}\bra{1}_\alpha\otimes\textsc{Be}[H/\alpha_H].
\end{align}
\item 
the $\textsc{CZ}_{\alpha,\text{abs}}$ gate that, controlled on the state $\ket{0}_\alpha$, applies the Pauli $Z$ operator on single-qubit register $\ket{\cdot}_\text{abs}$.
\item The $\textsc{CSwap}_{\alpha,c,d}$ gate that controlled on the state $\ket{0}_\alpha$, swaps the states of single-qubit registers $\ket{\cdot}_c\ket{\cdot}_d$.
\end{enumerate}
We now prove that
\begin{align}
\textsc{Mul}=D_h\cdot U_\alpha^\dagger\cdot \textsc{CU}_\text{r}^\dagger\cdot \textsc{Sel}_{L,H}\cdot \textsc{CZ}_{\alpha,\text{abs}}\cdot\textsc{CSwap}_{\alpha,c,d}\cdot \textsc{CU}_\text{r}\cdot U_\alpha.
\end{align}
For any $\ket{j}$,
\begin{align}
&\ket{|j|}_\text{abs}\ket{\mathrm{s}[j]}_\text{sgn}\ket{0}^{b}_\text{gb}\ket{0}_\alpha\ket{0}_c\ket{0}_d
\\
&\rightarrow_{U_\alpha}\ket{|j|}_\text{abs}\ket{0}^{b}_\text{gb}\left(\frac{\sqrt{R\alpha_L}}{\sqrt{\alpha}}\ket{0}_\alpha+\frac{\sqrt{\alpha_H}}{\sqrt{\alpha}}\ket{1}_\alpha\right)\ket{0}_c\ket{0}_d\ket{\mathrm{s}[j]}_\text{sgn}
\\
&\rightarrow_{\textsc{CU}_\text{r}}\ket{|j|}_\text{abs}\left(\frac{\sqrt{R|j|\alpha_L}}{\sqrt{M\alpha}}\ket{0}_\alpha\ket{0}_c\ket{\text{gb}_j}_\text{gb}+\frac{\sqrt{\alpha_H}}{\sqrt{\alpha}}\ket{1}_\alpha\ket{0}^{b}_\text{gb}\ket{0}_c+\cdots\ket{0}_\alpha\ket{1}_c\right)\ket{0}_d\ket{\mathrm{s}[j]}_\text{sgn}
\\
&\rightarrow_{\textsc{CSwap}_{\alpha,c,d}}\ket{|j|}_\text{abs}\left(\frac{\sqrt{h|j|\alpha_L}}{\sqrt{\alpha}}\ket{0}_\alpha\ket{00}_{c,d}\ket{\text{gb}_j}_\text{gb}+\frac{\sqrt{\alpha_H}}{\sqrt{\alpha}}\ket{1}_\alpha\ket{0}^{b}_\text{gb}\ket{00}_{c,d}+\cdots\ket{0}_\alpha\ket{01}_{c,d}\right)\ket{\mathrm{s}[j]}_\text{sgn}
\\
&\rightarrow_{\textsc{CZ}_{\alpha,\text{abs}}}\ket{j}\left(\frac{(-1)^{\mathrm{s}[j]}\sqrt{h|j|\alpha_L}}{\sqrt{\alpha}}\ket{0}_\alpha\ket{00}_{c,d}\ket{\text{gb}_j}_\text{gb}+\frac{\sqrt{\alpha_H}}{\sqrt{\alpha}}\ket{1}_\alpha\ket{0}^{b}_\text{gb}\ket{00}_{c,d}+\cdots\ket{0}_\alpha\ket{01}_{c,d}\right)
\\
&\rightarrow_{\textsc{Sel}_{L,H}}\ket{j}\left(\left(\frac{(-1)^{\mathrm{s}[j]}\sqrt{h|j|\alpha_L}}{\sqrt{\alpha}}\ket{0}_{\alpha}\ket{\text{gb}_j}_\text{gb}\textsc{Be}\left[\frac{L}{\alpha_L}\right]+\frac{\sqrt{\alpha_H}}{\sqrt{\alpha}}\ket{0}^{b}_\text{gb}\ket{1}_{\alpha}\textsc{Be}\left[\frac{H}{\alpha_H}\right]\right)\ket{00}_{cd}+\cdots\ket{1}_{d}\right)
\\
&\rightarrow_{\textsc{CU}_\text{r}^\dagger}\ket{j}\left(\left(\frac{(-1)^{\mathrm{s}[j]}h|j|\sqrt{\alpha_L}}{\sqrt{\alpha}}\ket{0}_{\alpha}\textsc{Be}\left[\frac{L}{\alpha_L}\right]+\frac{\sqrt{\alpha_H}}{\sqrt{\alpha}}\ket{1}_{\alpha}\textsc{Be}\left[\frac{H}{\alpha_H}\right]\right)\ket{00}_{cd}\ket{0}^{b}_\text{gb}+\ket{0^\perp}_\text{gb}+\cdots\ket{1}_{d}\right)
\\
&\rightarrow_{\textsc{U}_\alpha^\dagger}\ket{j}\left(\left(\frac{hj\alpha_L}{\alpha}\textsc{Be}\left[\frac{L}{\alpha_L}\right]+\frac{\alpha_H}{\alpha}\textsc{Be}\left[\frac{H}{\alpha_H}\right]\right)\ket{0}_{\alpha cd,\text{gb}}+\cdots\ket{0^\perp}\right)
\end{align}
Note that the coefficient $R/M=h$. 
By the definition of block-encodings, there is an all-zero state $\ket{0}_f$ such that $(\bra{0}_f\otimes\mathcal{I})\textsc{Be}[L/\alpha_L](\ket{0}_{f}\otimes\mathcal{I})=L/\alpha_L$ and similarly for $\textsc{Be}[H/\alpha_H]$. Hence
\begin{align}
\textsc{Mul}(\ket{j}\ket{0}_{\alpha cdf,\text{gb}}\otimes\mathcal{I})=\ket{j}\left(\ket{0}_{\alpha cd,\text{gb}}\otimes\frac{hjL+H}{\alpha}+\cdots\ket{0^\perp}\otimes\cdots\right),
\end{align}
block-encodes $\frac{hjL+H}{\alpha}$ controlled on $\ket{j}$ as claimed.
\end{proof}

\subsection{Superposition of kernel function amplitudes}
In this section, we describe a unitary that prepares a superposition state whose coefficients are the kernel function. Our approach is based on the quantum singular value transformation, and is similar to~\cite{McArdle2022StatePrep}, except that we avoid needing to block-encode an inverse arcsine function. This leads to a simpler circuit description with a better normalization factor and better gate complexity, but increases the ancilla overhead from $\mathcal{O}(1)$ to $\log_2M+\mathcal{O}(1)$.
\begin{lemma}\label{lem:kernel_superposition_state}
For any $\epsilon>0,c>0,h>0,\gamma\ge0,R=\Omega(\gamma\log^{1/2}\frac{1}{\epsilon})$, let the states
\begin{align}
\ket{\psi_{R,h,\gamma}}&=\sum_{j=-R/h}^{R/h}\frac{\sqrt{|a_j|}}{\sqrt{\alpha}}\ket{j},
\quad
&\overline{\ket{\psi_{R,h,\gamma}}}&=\sum_{j=-R/h}^{R/h}e^{-ihc}\frac{\sqrt{|a_j|}}{\sqrt{\alpha}}\ket{j},
\\
a_j&=\hat{f}_2(hj;\gamma,c),
\quad
&\alpha &= \sum_{j=-R/h}^{R/h}|a_j|,
\end{align}
There is a quantum circuit $\textsc{Prep}'_{\vec{a}}$ and $\overline{\textsc{Prep}}'_{\vec{a}}$ that prepares a state $\ket{\psi'_{R,h,\gamma}}$ and $\overline{\ket{\psi'_{R,h,\gamma}}}$ with unit probability such that 
\begin{align}
|\ket{\psi_{R,h,\gamma}}-\ket{\psi'_{R,h,\gamma}}|\le\epsilon,\quad
|\overline{\ket{\psi_{R,h,\gamma}}}-\overline{\ket{\psi'_{R,h,\gamma}}}|\le\epsilon,
\end{align}
respectively, 
using 
\begin{align}
\mathcal{O}\left(R^{3/2}\log\frac{R}{\epsilon}\log\frac{R}{h}\right),
\end{align}
two-qubit gates and $\log_2\frac{R}{h}+\mathcal{O}(1)$ ancillae.

\end{lemma}
\begin{proof}
Let us define $g_2(k;\gamma)\doteq|\hat{f}_2(k;\gamma,c)|^{1/2}/e^{c/2}=\frac{1}{\sqrt{1+k^{2}}}e^{-\frac{k^{2}}{8\gamma^{2}}}$.
As any constant $e^{c/2}$ may be absorbed into the normalization $\alpha$, we prepare a superposition state over $g_2(jh;\gamma)$ for $j=-R/h,-R/h+1\cdots,R/h$.
We construct $\textsc{Prep}_{R,h,\gamma}$ as follows.
\begin{enumerate}
\item Prepare the uniform superposition state 
\begin{align}
\ket{\text{uniform}}=\frac{1}{\sqrt{2R/h+1}}\sum_{j=-R/h}^{R/h}\ket{j}.
\end{align}
This can be done exactly with unit success probability using $\mathcal{O}(\log\frac{R}{h})$ two-qubit gates~\cite{Sanders2018}.

\item Define the diagonal operator 
\begin{align}
D&=\sum_{j=-R/h}^{R/h}\frac{hj}{R}\ket{j}\bra{j}
=(\ket{0}_L\otimes \mathcal{I})^\dagger\textsc{Be}[D](\ket{0}_L\otimes \mathcal{I}), 
\\
\textsc{Be}[D]&=\textsc{U}_\text{r}^\dagger\cdot \textsc{Z}_{\text{abs}}\cdot\textsc{Swap}_{c,d}\cdot \textsc{U}_\text{r},
\quad
\ket{0}_L\doteq\ket{0}_c\ket{0}_d\ket{0\cdots 0}_\text{gb},
\end{align}
which is block-encoded using $\mathcal{O}(\log\frac{R}{h})$ two-qubit gates, where $U_\text{r}$ is from~\cref{eq:U_r}, $\textsc{Swap}_{c,d}$ swaps registers $\ket{\cdot}_c\ket{\cdot}_d$ and $\textsc{Z}_\text{abs}$ is the Pauli $Z$ gate on register $\ket{\cdot}_\text{abs}$. 

\item \cref{lem:bernstein_f5} defines a polynomial $p_n$ of degree $n=\mathcal{O}(R\log\frac{R}{\epsilon_\text{poly}})$ that approximates $g_2$ to error $\epsilon_{k/R}=p_n(k/R)-g_2(k;\gamma)$ where the error satisfies $\forall x\in[-1,1],\;|\epsilon_{x}|\le\epsilon_\text{poly}$. As $|g_2(k;\gamma)|\le 1$, we may always rescale the polynomial by a small factor $\frac{1}{1+\epsilon_\text{poly}}$ without changing the asymptotic complexity $n$ such that $\forall_{x\in[-1,1]}|p_n(x)|\le1$.
Using quantum singular value transformations~\cite{Gilyen2018singular}, we may block-encode $\textsc{Be}[p_n[D]]\in\mathbb{C}^{2^m\times 2^m}$, where $m=2\log\frac{R}{h}+\mathcal{O}(1)$.
\begin{align}
p_n[D]=\sum_{j=-R/h}^{R/h}p_n\left(\frac{hj}{R}\right)\ket{j}\bra{j},
\end{align}
using $\mathcal{O}(n)$ queries to $\textsc{Be}[D]$ and $\mathcal{O}(n\log\frac{R}{h})$ two-qubit gates.
\item Apply the block-encoding of $p_n[D]$ to $\ket{\text{uniform}}$ to obtain
\begin{align}
\textsc{Be}[p_n[D]]\ket{0\cdots 0}\ket{\text{uniform}}
&=
\sqrt{p}\ket{0\cdots 0}\frac{1}{\sqrt{2R/h+1}}\sum_{j=-R/h}^{R/h}p_n\left(\frac{hj}{R}\right)\ket{j}+\cdots\ket{0\cdots 0^\perp},
\\
\ket{\psi^{\prime}_{R,h,\gamma}}&=\frac{1}{\sqrt{p}}\frac{1}{\sqrt{2R/h+1}}\sum_{j=-R/h}^{R/h}p_n\left(\frac{hj}{R}\right)\ket{j},
\end{align}
where the success probability of obtaining the normalized quantum state $\ket{\psi^{\prime}_{R,h,\gamma}}$ is $p$. As $p$ may be computed classically to arbitrary precision from $p_n$, $\ket{\psi^{\prime}_{R,h,\gamma}}$ can be prepared by amplitude amplification with unit probability and $\mathcal{O}(1/\sqrt{p})$ queries of the unitary circuits in steps $1-4$ and their inverses.
\end{enumerate}

The success probability
\begin{align}
p&=\frac{\sum_{j=-R/h}^{R/h}p^2_n\left(\frac{hj}{R}\right)}{2R/h+1}
=
\frac{\sum_{j=-R/h}^{R/h}(g_2(hj;\gamma)+\epsilon_{hj/R})^2}{2R/h+1}
\ge
\frac{\sum_{j=-R/h}^{R/h}g^2_2(hj;\gamma)-2g_2(hj;\gamma)\epsilon_\text{poly}}{2R/h+1}.
\end{align}
As $g_2(k;\gamma)$ is symmetric about $k$ and monotonically decreasing with $|k|$, the sums are bounded by the integrals
\begin{align}\nonumber
h\alpha/e^c=h\sum_{j=-R/h}^{R/h}g^2_2(hj;\gamma)&\ge\int_{-R}^{R}g^2_2(k;\gamma)\mathrm{d}k=\int_{-\infty}^{-\infty}g^2_2(k;\gamma)\mathrm{d}k-2\int_{R}^{\infty}g^2_2(k;\gamma)\mathrm{d}k
\\\nonumber
&\ge \pi  e^{\frac{1}{4 \gamma ^2}} \text{erfc}\left(\frac{1}{2 \gamma }\right)-2\int_{R}^{\infty}\frac{k}{R^3}e^{-\frac{k^2}{4\gamma^2}}\mathrm{d}
\\
&\ge\pi  e^{\frac{1}{4 \gamma ^2}} \text{erfc}\left(\frac{1}{2 \gamma }\right)-\frac{4 \gamma ^2}{R^3}e^{-\frac{R^2}{4 \gamma ^2}}
\\
&=\Omega\left(1\right),\quad\text{using}\;R/\gamma=\Omega\left(\log^{1/2}\frac{1}{\epsilon}\right)\;\text{and}\;\gamma>0.
\end{align}
and
\begin{align}\nonumber
h\sum_{j=-R/h}^{R/h}g_2(hj;\gamma)&\le  h g_2(0;\gamma)+\int_{-\infty}^{\infty}g_2(k;\gamma)\mathrm{d}k
=h+e^{\frac{1}{16 \gamma ^2}} K_0\left(\frac{1}{16 \gamma ^2}\right)\\
&=\mathcal{O}\left(h+\log\gamma\right),
\end{align}
where $K_0$ is the modified Bessel function of the second kind.
Hence
\begin{align}
p=\Omega\left(\frac{1}{h}\frac{1-(h+\log\gamma)\epsilon_\text{poly}}{2R/h+1}\right)=\Omega(R^{-1}).
\end{align}

We now bound the error of $\ket{\psi'_{R,h,\gamma}}$. 
Let $a_j=g_2(\frac{hj}{R})$ and $b_j=p_n(hj/R)$.
Then by~\cref{lem:state_vector_error},
\begin{align}\nonumber
|\ket{\psi'_{R,h,\gamma}}-\ket{\psi_{R,h,\gamma}}|&=
\left|\frac{\vec{b}}{|\vec{b}|_2}-\frac{\vec{a}}{\sqrt{\alpha}}\right|\le
\frac{4}{\sqrt{\alpha}}\sqrt{2R/h+1}\left|\vec{b}-\vec{a}\right|_\infty
\le
\frac{4}{\sqrt{\alpha}}\sqrt{2R/h+1}\epsilon_\text{poly}
\\
&=\mathcal{O}\left(
\sqrt{R}\epsilon_\text{poly}\right).
\end{align}
We choose $\epsilon_\text{poly}=\mathcal{O}(\epsilon/\sqrt{R})$ and obtain the overall gate complexity by multiplying the inverse success amplitude $1/\sqrt{p}$ with the gates for block-encoding $\textsc{Be}[p_n[D]]$, which scales with the polynomial degree $n$ and $\log\frac{R}{h}$.

We apply the diagonal unitary $\textsc{D}_{hc}=\sum_{j=-R/h}^{R/h}e^{-ihjc}\ket{j}$ on $\textsc{Prep}'_{\vec{a}}$ to obtain $\overline{\textsc{Prep}}'_{\vec{a}}$.
$\textsc{D}_{hc}$ can be implemented using $b'=\log_2(\frac{R}{h}+1)$ controlled phase gates of the form
\begin{align}
\textsc{CP}_l=\ket{0}\bra{0}_\text{sign}\otimes\left[\begin{array}{cc}
     1&0  \\
     0&e^{-ihc2^l}
\end{array}\right] 
+
\ket{1}\bra{1}_\text{sign}\otimes\left[\begin{array}{cc}
     1&0  \\
     0&e^{ihc2^l}  
\end{array}\right], 
\end{align}
where each $\textsc{CP}_l$ is controlled by the sign qubit $\ket{\mathrm{s}[j]}_\text{sgn}$ and targets the $l^\text{th}$ qubit of $\ket{|j|}_\text{abs}$.
\end{proof}

\begin{lemma}\label{lem:state_vector_error}
For any $M>0,\epsilon\in[0,M^{-1/2}/2]$, let the vectors $\vec{p}\in\mathbb{R}^{M},\vec{f}\in\mathbb{R}^{M}$ differ by at most $\max_{j}|p_j-f_j|=|\vec{p}-\vec{f}|_\infty\le\epsilon$ where $|\vec{f}|_2=\sqrt{\sum_{j}|f_j|^2}=1$. Then
\begin{align}
\left|\frac{\vec{p}}{|\vec{p}|_2}-\vec{f}\right|_2\le 4\sqrt{M}\epsilon.
\end{align}
\end{lemma}
\begin{proof}
Let $\vec{\epsilon}=\vec{f}-\vec{p}$. Then
\begin{align}
|\vec{\epsilon}|_2&\le\sqrt{M}\epsilon,
\\
|\vec{p}|_2&=|\vec{f}-\vec{\epsilon}|_2\le|\vec{f}|_2+|\vec{\epsilon}|_2\le1+\sqrt{M}\epsilon,
\\
|\vec{p}|_2&\ge||\vec{f}|_2-|\vec{\epsilon}|_2|\ge1-\sqrt{M\epsilon},
\end{align}
By a triangle inequality,
\begin{align}
\left|\frac{\vec{p}}{|\vec{p}|_2}-\vec{f}\right|_2
&=
\left|\frac{\vec{p}-\vec{f}}{|\vec{p}|_2}+\vec{f}\left(\frac{1}{|\vec{p}|_2}-1\right)\right|_2
\le
\frac{\left|\vec{p}-\vec{f}\right|_2}{|\vec{p}|_2}+\left|\vec{f}\right|_2\left|\frac{1}{|\vec{p}|_2}-1\right|
\\
&\le
\frac{\left|\vec{\epsilon}\right|_2}{|\vec{p}|_2}+\frac{|1-|\vec{p}|_2|}{|\vec{p}|_2}
\le\frac{2\sqrt{M}\epsilon}{|\vec{p}|_2}\le\frac{2\sqrt{M}\epsilon}{1-\sqrt{M}\epsilon}\le4\sqrt{M}\epsilon.
\end{align}
\end{proof}

\begin{lemma}\label{lem:bernstein_f5}
For any $\gamma>0,\epsilon>0,R>0$, the degree-$n=\mathcal{O}(R\log\frac{R}{\epsilon})$ Chebyshev truncation $p_n(x)$ of $g_2(Rx;\gamma)$ has error
\begin{align}
\max_{x\in[-1,1]}|g_2(Rx;\gamma)-p_n(x)|\le\epsilon,
\end{align}

\end{lemma}
\begin{proof}
The function $g_2(k;\gamma)=\frac{1}{\sqrt{1+k^2}}e^{-\frac{k^2+1}{8\gamma^2}}$ is analytic in the complex plane except at the branch cuts $z\in\pm i[1,\infty]$. Hence, it is analytically continuable to the open Bernstein ellipse $\mathcal{E}_\rho$, where $\rho\in(1,1+\sqrt{2})$, and let $y=\frac{\rho-\rho^{-1}}{2}\in(0,1)$.
Using the parameterization $z=\frac{1}{2}(\rho e^{i\theta}+\rho^{-1}e^{-i\theta})$,
let us expand the absolute value of the denominator $|1+z^2|^2=\frac{2 \rho ^4 \cos (4 \theta )+12 \left(\rho ^6+\rho ^2\right) \cos (2 \theta )+\rho ^8+36 \rho ^4+1}{16 \rho ^4}$. As the derivative $\frac{\mathrm d}{\mathrm d\theta}|1+z^2|^2\propto\rho ^4 \sin (4 \theta )+3 \left(\rho ^6+\rho ^2\right) \sin (2 \theta )$, the stationary points are at $\theta=\pm\pi/2,0,\pi$. By taking a second derivative, the minimum is at $\theta=\pm\pi/2$ and the maximum at $\theta=0,\pi$. Hence,
\begin{align}
\max_{z\in\mathcal{E}_\rho}\frac{1}{|1+z^2|^{1/2}}\le \frac{2\rho}{\sqrt{6 \rho ^2-\rho ^4-1}}=\frac{1}{\sqrt{1-y^2}}.
\end{align}
The absolute value of exponential $\left|\exp{\left(-\frac{z^2+1}{8\gamma^2}\right)}\right|=\exp{\left(-\frac{\left(\rho ^4+1\right) \cos (2 \theta )+6 \rho ^2}{32 \gamma ^2 \rho ^2}\right)}$ is maximized at $\theta=\pi/2$ where
\begin{align}
\max_{z\in\mathcal{E}_\rho}|e^{-\frac{z^2+1}{8\gamma^2}}|\le e^{\frac{\rho ^4-6 \rho ^2+1}{32 \gamma ^2 \rho ^2}}=e^{-\frac{1-y^2}{8 \gamma ^2}}\le 1.
\end{align}
Using the sub-multiplicative property of norms, 
$\max_{z\in\mathcal{E}_\rho}|g_2(z;\gamma)|\le \frac{1}{\sqrt{1-y^2}}$.
Now rescaling with $z\rightarrow Rz$, $\max_{z\in\mathcal{E}_\rho}|g_2(Rz;\gamma)|\le \frac{1}{\sqrt{1-(Ry)^2}}=M$, for $y\in(0,1/R)$.
Let us choose $y=\frac{1}{2R}$, then $M=\sqrt{4/3}$ and $\rho = 1+\frac{1}{2R}+\mathcal{O}(R^{-2})$. Then
\cref{lem:chebyshev_truncation} on Chebyshev truncation states that the degree-$n$ approximation error is at most 
\begin{align}
\max_{x\in[-1,1]}|g_2(Rx;\gamma)-p_n(x)|\le\frac{2M\rho^{-n}}{\rho-1}=\epsilon=\Theta\left(\frac{\left(1+1/(2R)\right)^{-n}}{2R}\right).
\end{align}
Solving for $\epsilon$, it suffices to choose $n=\mathcal{O}(R\log\frac{R}{\epsilon})$.
\end{proof}

\begin{lemma}[Polynomial approximation by Chebyshev truncation~\cite{Trefethen2019Approximation}]\label{lem:chebyshev_truncation}
Let a function $f$ that is analytic on $[-1,1]$ be analytically continuable to the open Bernstein ellipse $\mathcal{E}_\rho=\{\frac{1}{2}(z+z^{-1}):|z|\le\rho\}$ where there is some $M>0$ such that $\max_{z\in\mathcal{E}_\rho}|f(z)|\le M$. Then for any $n\ge0$, truncating its Chebyshev expansion
\begin{align}
f(x)&=\sum_{j=0}^{\infty} a_j T_j(x),\quad T_j(\cos\theta)\doteq\cos(j\theta),
\end{align}
to degree $n$ approximates $f$ with error
\begin{align}
\max_{x\in[-1,1]}\left|f(x)-\sum_{j=0}^{n} a_j T_j(x)\right|\le\frac{2M\rho^{-n}}{\rho-1}.
\end{align}
\end{lemma}
\section{Constant factor optimization}\label{sec:constant_factors}
The query complexity of the LCHS is characterized by the normalization constant and the truncation radius.
In this section, we numerically optimize the query complexity cost prefactor $\mathrm{cost}=\alpha_{\hat{f}_{j,y},R}\cdot R$ with respect to $R,\gamma,c,j$ for a more general family of kernel functions $\hat{f}_{j,y}$~\cref{eq:generalized_kernel_intro}.
As the choice of numerical quadrature does not affect query complexity, it suffices to compute the error of LCHS directly by~\cref{thm:LCHS_general}.
For any target $\epsilon$, the best possible query complexity prefactor for this family of kernel functions may be obtained by numerically minimizing the product
\begin{align}\label{eq:optimization}
\text{Cost}=\min_{j,y,R,\gamma,c,y_0}\alpha_{\hat{f}_{j,y},R}R,\quad\text{s.t.}\quad 
\int_{\mathbb{R}\backslash[-R,R]}\frac{|\hat{f}_{j,y}(k;\gamma,c)|}{\sqrt{2\pi}}\mathrm{d}k+\int_{\mathbb{R}}\frac{|\hat{f}_{j,y}(k-iy_0;\gamma,c)|}{\sqrt{2\pi}}\mathrm{d}k=\epsilon.
\end{align}
For practical implementation, one should use the numerically optimized LCHS parameters.
In~\cref{fig:lchs_comparison}, we plot this prefactor resulting from solving~\cref{eq:optimization} with $j=2$ and $y=1$ fixed by gradient descent from the parameters proposed in~\cref{thm:LCHS_entire}
and find that the analytically derived upper bound is loose by only a factor $<1.3$ in the small $\epsilon$ limit. Note that the difference between $\alpha_{\hat{f},R}$ and $\alpha_{\hat{f}}$ in~\cref{fig:lchs_comparison} is visually indistinguishable.
We also tabulate~\cref{tab:generalized_solution} and plot the fits~\cref{fig:lchs_params_fit} of minimum cost solutions that additionally optimize over $j\ge1$ and $y>0$. Compared to the $\hat{f}_{2,1}$ case, we find roughly a $25\%$ improvement in the prefactor for $\epsilon\in[10^{-10},10^{-1}]$.
\begin{table}[t]
    \centering
    \begin{tabularx}{\textwidth}{Y|c|Y|Y|Y|Y|Y|Y|Y|Y}
    \hline\hline
     $\epsilon$ & Cost $\alpha_{\hat{f}_{j,y},R}\cdot R$ & $j$ & $y$ & $\alpha_{\hat{f}_{j,y},R}$ & $\gamma$ & $R$ & $c$ & $y_0$  \\
\hline
$10^{-1}$ & 3.32 & \multirow{10}{*}{2} & \multirow{10}{*}{1} & 1.178 & 1.749 & 2.82 & 0.586 & 5.58 \\ 
$10^{-2}$ & 9.34 &  & & 1.656 & 1.996 & 5.64 & 0.832 & 8.40 \\ 
$10^{-3}$ & 16.82 &  & & 1.921 & 2.314 & 8.75 & 0.928 & 11.67 \\ 
$10^{-4}$ & 25.25 &  & & 2.089 & 2.623 & 12.09 & 0.976 & 15.16 \\ 
$10^{-5}$ & 34.35 &  & & 2.203 & 2.916 & 15.59 & 1.003 & 18.79 \\ 
$10^{-6}$ & 43.93 &  & & 2.285 & 3.194 & 19.23 & 1.019 & 22.53 \\ 
$10^{-7}$ & 53.86 &  & & 2.345 & 3.457 & 22.96 & 1.029 & 26.37 \\ 
$10^{-8}$ & 64.06 &  & & 2.392 & 3.708 & 26.78 & 1.036 & 30.27 \\ 
$10^{-9}$ & 74.48 &  & & 2.429 & 3.948 & 30.66 & 1.041 & 34.23 \\ 
$10^{-10}$ & 85.05 &  & & 2.459 & 4.177 & 34.59 & 1.044 & 38.22 \\ 
\hline
$10^{-1}$ & 2.55 & 3.68 & 1.05 & 1.272 & \multirow{10}{*}{$\infty$} & 2.01 & -0.206 & 12.54 \\ 

$10^{-2}$ & 7.06 & 6.52 & 2.45 & 1.665 &  & 4.24 & -0.329 & 14.92 \\ 

$10^{-3}$ & 12.74 & 9.86 & 4.12 & 1.902 &  & 6.70 & -0.389 & 19.28 \\ 

$10^{-4}$ & 19.26 & 13.65 & 6.01 & 2.075 &  & 9.28 & -0.432 & 23.89 \\ 

$10^{-5}$ & 26.42 & 17.77 & 8.05 & 2.210 &  & 11.95 & -0.466 & 28.61 \\ 

$10^{-6}$ & 34.08 & 22.14 & 10.21 & 2.320 &  & 14.69 & -0.492 & 33.42 \\ 
$10^{-7}$ & 42.15 & 26.70 & 12.46 & 2.412 &  & 17.47 & -0.513 & 38.28 \\ 
$10^{-8}$ & 50.56 & 31.42 & 14.78 & 2.490 &  & 20.30 & -0.531 & 43.20 \\ 
$10^{-9}$ & 59.27 & 36.27 & 17.16 & 2.559 &  & 23.16 & -0.546 & 48.16 \\ 
$10^{-10}$ & 68.23 & 41.23 & 19.60 & 2.619 &  & 26.05 & -0.558 & 53.15 \\ 
\hline\hline
\end{tabularx}
    
    \caption{For each target error $\epsilon$, we minimize the LCHS cost prefactor $\alpha_{\hat{f}_{j,y},R}\cdot R$ using the generalized kernel functions $\hat{f}_{j,y}$~\cref{eq:generalized_kernel_intro} by solving the optimization program~\cref{eq:optimization} by gradient descent. The first ten row correspond to optimized solutions for $\hat{f}_2=\hat{f}_{2,1}$, and the next ten rows allow optimization of $\hat{f}_{j,y}$ over $j$ and $y$ and we find that the optimized solutions have arbitrarily large $\gamma$ which we round to $\infty$. The parameter $y_0$ does not affect the definition of the kernel function but is needed for the interested reader to verify the calculations in this table.}
    \label{tab:generalized_solution}
\end{table}

\begin{figure}
    \centering
    \includegraphics[width=0.48\linewidth]{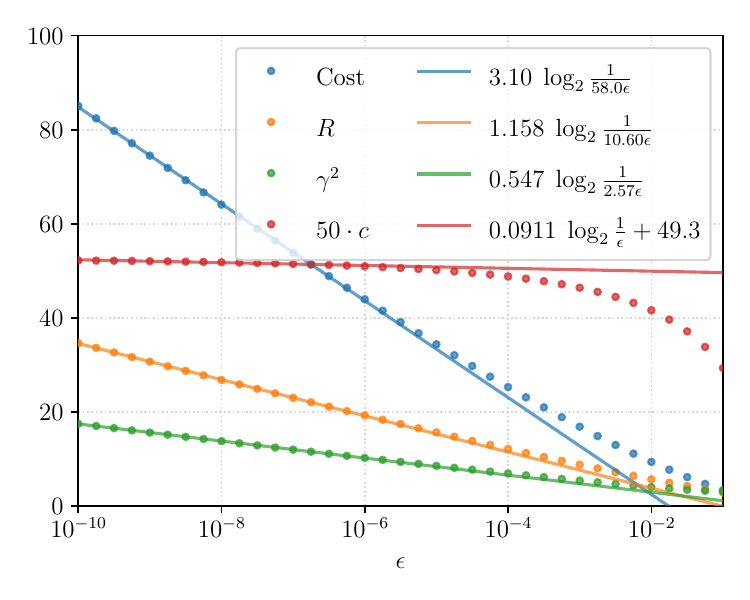}
    \includegraphics[width=0.48\linewidth]{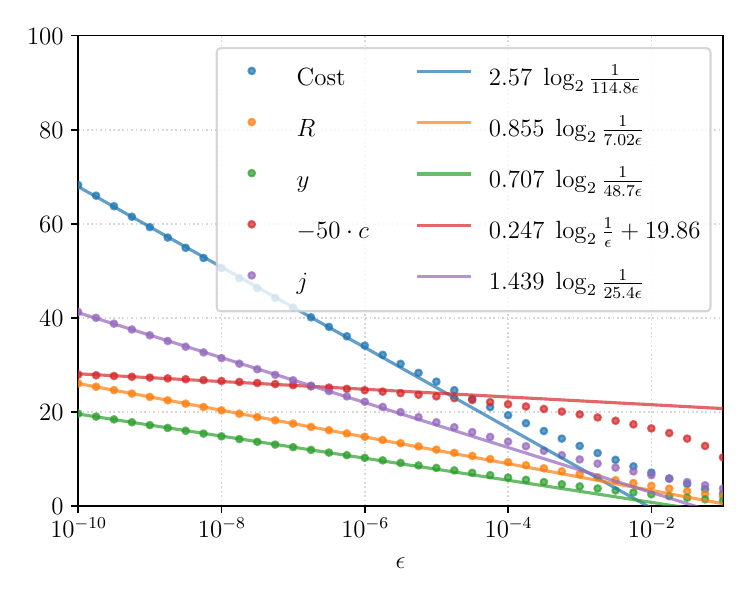}\vspace{-.5cm}

    \caption{Parameter choices of numerically optimized LCHS from~\cref{tab:generalized_solution} together with fits to points at $\epsilon\le10^{-6}$ for the kernel functions (left) $\hat{f}_2(k; \gamma,c)$ and (right) $\hat{f}_{j,y}(k;\gamma,c)$.}
    \label{fig:lchs_params_fit}
\end{figure}

\section{Query lower bounds}\label{sec:lower_bounds}
In this section, we present a $\Omega(u_{\text{lchs}}\alpha_At)$ lower bound~\cref{thm:lower_bound} on the query complexity for solving differential equations, and lower bounds on the query complexity of block-encoding $\textsc{Be}[e^{-At}/\alpha_{\hat{f}}]$ that suggest our method is optimal in all parameters and cannot be improved by more than a constant factor of $3$.
As LCHS block-encodes $e^{-At}$ for non-Hermitian $A=L+iH$, where $L$ and $H$ are Hermitian, any lower bound on block-encoding the special case $e^{-Lt}$ or $e^{-iHt}$ is a lower bound on LCHS.
The case $e^{-iHt}$ will not give useful lower bounds on LCHS as it corresponds to the Hamiltonian simulation problem, and the query complexity to the block-encoding $\textsc{Be}[H/\alpha_H]$ for $t=\Omega(\log\frac{1}{\epsilon})$ is the optimal value $\alpha_H t+\Theta(\log\frac{1}{\epsilon})$~\cite{Low2016HamSim,Low2016Qubitization,Berry2024HamSim}. 
Hence, we focus on the case $e^{-Lt}$, which implements imaginary time evolution~\cite{Suzuki2025Grover} and Gibbs sampling~\cite{Chowdhury2016quantum}.
Although a $\Omega(\alpha_Lt)$ lower bound on imaginary time evolution can be proven in a variety of ways~\cite{Suzuki2025Grover}, there is no known lower bound that is multiplicative in $\alpha_Lt$ and $\log\frac{1}{\epsilon}$, which is needed to prove the optimality of our LCHS.

\begin{theorem}[An overlap-dependent LCHS lower bound]\label{thm:lower_bound}
    Let $K \ge 1$ and $N \ge 1$ be integers, and $\alpha >0$. There exists a family of PSD Hamiltonians $\{L_{N,K} \in \mathbb C^{N \times N}\}$ satisfying $\|L_{N,K}\|\le \alpha$, evolution times $\{t_{N,K} \ge 0\}$, and normalized states $\{\ket{\phi_N} \in \mathbb C^N\}$, such that preparing the states
    \begin{align}
        \ket{\psi_{N,K}}:= \frac{e^{-t_{N,K}L_{N,K}}\ket{\phi_N}}{\|e^{-t_{N,K}L_{N,K}}\ket{\phi_N}\|}
    \end{align}
    within constant additive error requires
    \begin{align}
        \tilde \Omega \left ( 
        \frac{1}{\|e^{-t_{N,K}L_{N,K}}\ket{\phi_N}\|} \alpha t_{N,K}   \right)
    \end{align}
    queries to the block-encoding of $L_{N,K}/\alpha$, including its inverse and controlled version. The tilde notation hides a factor of $K/\log N$ assumed to be smaller than 1. 
\end{theorem}
\begin{proof}
    Without loss of generality we assume $\alpha$ to be a constant. Our proof considers a reduction from SEARCH with a single marked element, for which preparing $\ket{\psi_{N,K}}$ within small, constant additive error would solve a SEARCH instance. We then use the well-known lower bound $\Omega(\sqrt N)$~\cite{bennett1997strengths} to prove our result.

    In the following we drop the indices $K$ and $N$ for simplicity.  The instances we consider are such that
    \begin{align}
        \ket{\phi}:=\frac 1 {\sqrt N} \sum_{j=0}^{N-1} \ket j \in \mathbb C^N
    \end{align}
    is the uniform superposition state.

  The Hamiltonians are
   \begin{align}
       L:=-(1+\delta) \ketbra{x} - (1-\delta) \ketbra \phi + E(\delta),
   \end{align}
where $x \in \{0,\ldots,N-1\}$ is the marked element, $\delta >0$ is a small parameter that will depend on $K$ and $N$, and $E(\delta)$ is such that the lowest eigenvalue  $E_0$ of $L \succeq 0$ is exactly zero. In particular, it is possible to show
\begin{align}
    E(\delta)=1 + \sqrt{\frac 1 N -\frac{\delta^2}N+\delta^2} .
\end{align}
The  first nonzero eigenvalue is
\begin{align}
    E_1 = 2 \sqrt{\frac 1 N -\frac{\delta^2}N+\delta^2} .
\end{align}
All other eigenvalues are $E(\delta) \approx 1$. 

We can also compute the relevant eigenvectors, obtaining in the subspace spanned by $\{\ket x , \ket{x^\perp}\}$,
where $\ket{x^\perp}=\frac 1 {\sqrt{N-1}}\sum_{j \ne x}\ket j$,
\begin{align}
    \ket{v_0} \propto \begin{pmatrix} \frac{-(1-\delta)\sqrt{N-1}} {1+\delta(N-1)-\sqrt{N(1+\delta^2(N-1))}}
        \cr 1
    \end{pmatrix} ,
\end{align}
and
\begin{align}
    \ket{v_1} \propto \begin{pmatrix} \frac{-(1-\delta)\sqrt{N-1}} {1+\delta(N-1)+\sqrt{N(1+\delta^2(N-1))}}
        \cr 1
    \end{pmatrix} .
\end{align}
We are interested in the case where $\delta \rightarrow 0$ becomes asymptotically small but $\delta^2 N$ becomes asymptotically large
as $N$ becomes asymptotically large.
This will allow us to set different scalings for the eigenvalues and eigenvectors in terms of $N$, 
obtaining a more general case than the traditional one for solving SEARCH where $\delta$ is set exactly to zero.
In the asymptotic limit and under the approximations, we obtain, for example,
\begin{align}
    \frac{-(1-\delta)\sqrt{N-1}} {1+\delta(N-1)-\sqrt{N(1+\delta^2(N-1))}} = 2 \delta \sqrt N (1+\cO(\delta)).
\end{align}
We will also set the parameters so that $\delta^2 \sqrt N$ is asymptotically small.
This implies
\begin{align}
   \ket{v_0}=
   \begin{pmatrix} 1 \cr \frac 1{2 \delta \sqrt N}
    \end{pmatrix}   + \cO(\delta) ,
\end{align}
and using orthogonality,
\begin{align}
   \ket{v_1}=
   \begin{pmatrix} - \frac 1{2 \delta \sqrt N} \cr 1
    \end{pmatrix}   + \cO(\delta) .
\end{align}
Furthermore, we can express
\begin{align}
    \ket{\phi} =  \left(\frac 1 {2 \delta \sqrt N} +\cO(\delta)\right)  \ket{v_0} + \left( 1+ \cO\left( \frac 1 {2\delta \sqrt N}\right)\right)\ket{v_1} .
\end{align}

Since in this limit $E_1 = 2 \delta (1+\cO(1/(\delta^2 N)))\approx 2 \delta$, applying the exponential now gives
\begin{align}
    e^{-t L } \ket{\phi} \approx 
     \frac 1 {2 \delta \sqrt N}  \ket{v_0} + e^{-2t \delta}\ket{v_1} ,
\end{align}
and there exists $t = \cO(\frac 1 \delta \log(\delta^2 N))$ such that
the amplitudes of the evolved state in the eigenbasis are all of order $\cO(1/(\delta \sqrt N))$, being asymptotically small.
This readily guarantees that,
if we normalize the state,
a projective measurement in the computational basis will produce the marked element with constant probability, since $\ket{v_0}$ has large overlap with $\ket x$. For this approach, we note
\begin{align}
   \delta \sqrt N t =\cO( \sqrt N \log(\delta^2 N)).
\end{align}
Also note that $\|e^{-tL} \ket \phi\|=\cO(1/(\delta \sqrt N))$.
Hence, the lower bound from SEARCH implies that the query cost of this approach obeys
\begin{align}
    \Omega \left ( \frac 1 {\|e^{-tL}\ket \phi\|} \cdot t \cdot \frac 1 {\log(\delta^2 N)}\right) .
\end{align}
According to our definition, the precision is such that $\epsilon = \cO(1/(\delta \sqrt N))$.

We now set our parameters; for example,
we can choose
$\delta_{N,K}:= N^{\frac 1 2 (-1+1/K)}$ for  $K >4$. Note that $(\delta_{N,K})^2 N$ is asymptotically large and also $(\delta_{N,K})^2 \sqrt N$ is asymptotically small.
This also implies $t_{N,K}=\cO(N^{\frac 1 2 (1-1/K)} \cdot \frac 1 K \log (N))$. We can express our lower bound as
\begin{align}
    \Omega \left(  \frac 1 {\|e^{-t_{N,K}L_{N,K}}\ket {\phi_N}\|} \cdot t_{N,K} \cdot \frac 1 {\log(1/\epsilon_{N,K})}\right) 
\end{align}
for precision $\epsilon_{N,K}=\cO(1/N^{1/(2K)})$.

Note that the factor in precision 
is inverse to what we have wanted. However,
since $\epsilon_{N,K}$ can decrease mildly  as both $K$ and $N$ become large (e.g., by setting $K \sim (\log N)^{1-c}$ for $c >0$), then we can drop it in the expression and establish the lower bound
\begin{align}
    \tilde \Omega \left(  \frac 1 {\|e^{-t_{N,K}L_{N,K}}\ket {\phi_N}\|} \cdot t_{N,K}  \right) .
\end{align}
The tilde notation hides the factor $1/\log(1/\epsilon_{N,K})\sim K/\log N$, which is asymptotically small under the assumptions, but where the decay can be, e.g., sublogarithmic.

\end{proof}

Lastly, we note that, since $N$ and $K$ are independent, so are the norm of the evolved state and evolution time. For example, we can set $K$ and $N$ such that the norm $\sim 1/(2\delta \sqrt N)= 1/(2 N^{1/K})$ remains constant, while $t_{N,K}$ still grows with $\sqrt N$. This justifies the use of two independent integers, so that the lower bound has an explicit dependence in the norm of the state and evolution time. Additionally, multiplicative scaling in $\alpha_At$ with a tunable factor of any integer $\alpha_A=d>0$ can be included by taking the tensor product of $L_{N,K}$ with the $d\times d$ matrix where all entries are unity, which reproduces the exact same evolution in $t/d$ time with the uniform superposition ancilla state~\cite{Berry2014Exponential}.

We now compute lower bounds on block-encoding $e^{-Lt}$ by the approximate polynomial degree of an appropriate scalar function.
For any function $f:[-1,1]\rightarrow[-1,1]$, define the approximate degree $n=\widetilde{\mathrm{deg}}_{\epsilon^{\prime}}(f)$ of $f$ as the minimum degree of the best uniform polynomial approximation $p_n^{\prime}$ such that
\begin{align}
\max_{x\in[-1,1]}|f(x)-p^{\prime}_n(x)|\le\epsilon^{\prime}.
\end{align}
By the Quantum Singular Value Transformation (QSVT) algorithm~\cite{Gilyen2018singular}, one may block-encode 
\begin{align}
\textsc{Be}\left[\frac{p_n^{\prime}[L/\alpha_L]}{\alpha_++\alpha_-}\right],
\quad
\alpha_\pm\doteq\frac{1}{2}\max_{x\in[-1,1]}\left|p^{\prime}(x)\pm p^{\prime}(-x)\right|.
\end{align}
using exactly $n$ queries to the block-encoding $\textsc{Be}[L/\alpha_L]$. Note that $\alpha_++\alpha_-\le 2(1+\epsilon)$ is at most a constant. This is also the lower bound, following a $\Omega(\widetilde{\mathrm{deg}}_{\epsilon^{\prime}}(f))$ bound on the query complexity of computing matrix entries of any $f(L/\alpha_L)$~\cite{Montanaro2024MatrixQueryComplexity}. By the adversary method~\cite{Laneve2025QSPAdversary}, $\widetilde{\mathrm{deg}}_{\epsilon^{\prime}}(f)$ is in fact exactly the lower bound. Hence, QSVT is optimal, up to a constant factor of $\alpha_++\alpha_-$.
We may thus lower bound the query complexity of LCHS by minimizing $\widetilde{\mathrm{deg}}_{\epsilon^{\prime}}(f)$ over all real scalar functions $f$ that block-encode $f(L/\alpha_L)=e^{-Lt}/\alpha$ for some $\alpha\ge1$ and any $\alpha_L\mathcal{I}\succeq L\succeq0$, and is otherwise bounded by $1$.

Equivalently, for any given $\tau=\alpha_Lt$, and polynomial $p_n$ of degree $n$, we wish to find the polynomial $p_n^*$ that minimizes the error 
\begin{align}\label{eq:polynomial_target}
\epsilon=\underset{p_n}{\mathrm{inf}}\max_{x\in[0,1]}\left|e^{-\tau x}-p_n(x)\right|={\epsilon^{\prime}}\alpha,\quad\text{such that}
\quad \max_{x\in[-1,1]}|p_n(x)|\le \alpha.
\end{align}
We note that the degree scaling is sensitive to the value $\alpha$ of the uniform boundedness constraint as shown in the following known results.
\begin{itemize}
    \item $\alpha\rightarrow\infty$: There is an approximating polynomial of degree $\mathcal{O}(\sqrt{\tau+\log\frac{1}{\epsilon}}\sqrt{\log\frac{1}{\epsilon}})$~\cite{Sachdeva2013Approximation}.
    \item $\alpha=e^{\tau}$: There is an approximating polynomial of degree $\mathcal{O}(\tau+\log\frac{1}{\epsilon})$~\cite{Sachdeva2013Approximation}.
    \item $\alpha=1$: The approximating polynomial has degree $\Theta(\tau/\epsilon)$ as $f^{(1)}$ is discontinuous at $x=0$~\cite{Trefethen2019Approximation}.
\end{itemize}
For general classes of functions, such as $C^l$-smooth functions, or analytic functions, the optimal scaling of $n$ is known~\cite{Trefethen2019Approximation}.
However, the optimal scaling when $f$ is defined only piecewise on sub-intervals with a uniform boundedness constraint like $\max_{x\in[-1,1]}|p_n(x)|\le \alpha$ in~\cref{eq:polynomial_target} is poorly understood. 
We know of only two optimal results for the case $\alpha=\Theta(1)$ -- the sign function~\cite{Eremenko2011Sign} and the inverse function~\cite{Gilyen2018singular}, both of which are characterized by the maximum gradient $\tau$ between two sub-intervals and are approximated by polynomials of optimal degree $\Theta(\tau\log\frac{1}{\epsilon})$. Similarly, we expect the optimal degree for~\cref{eq:polynomial_target} to be $\Theta(\tau\log\frac{1}{\epsilon})$.

Using well-known numerical techniques~\cite{Hettich1993semiinfiniteLP}, we may rewrite~\cref{eq:polynomial_target} as a semi-infinite program that we discretize into some $K$ points.
Any real polynomial of degree $n$ is uniquely characterized by the coefficients $\vec{a}\in\mathbb{R}^{n+1}$ in the Chebyshev expansion
\begin{align}
p_n(x)=\sum_{j=0}^n a_jT_j(x),
\end{align}
into Chebyshev polynomials of the first kind. The Chebyshev basis is well-suited for this problem as it is nearly optimal for finding uniform polynomial approximations.
For some set of evaluation points $\mathcal{X}=\{x_k\}_{k=0}^{K-1}$, where $x_k=\cos\frac{\pi k}{K-1}$ are Chebyshev nodes, let the Chebyshev-Vandermonde matrix $V\in\mathbb{R}^{K\times (n+1)}$ have entries
\begin{align}
V_{kj}=T_j(x_k).
\end{align}
Then $p_n(x)$ evaluated at all points $\mathcal{X}$ is 
\begin{align}
\vec{p}_n=V\cdot\vec{a}\in\mathbb{R}^{K}.
\end{align}
Let $\vec{p}_{n,\text{right}}$ be the subsets of $\vec{p}_n$ evaluated at all points $x\ge 0$. 
Then~\cref{eq:polynomial_target} becomes the $K\rightarrow\infty$ limit of a linear program over $\vec{a}$, which, for finite $K$, is solvable in polynomial time.
\begin{align}\label{eq:polynomial_target_LP}
\underset{\vec{a}}{\min
\epsilon}
\quad
\text{such that}
\quad
|e^{-\tau x}-\vec{p}_{n,\text{right}}|_\infty\le \epsilon=\epsilon^{\prime}\alpha
\quad
\text{and}\quad|\vec{p}_n|_\infty\le \alpha.
\end{align}
For any finite $K$, the constraints in~\cref{eq:polynomial_target_LP} are a subset of those in~\cref{eq:polynomial_target}. Hence the degree $n$ that achieves a target error $\epsilon$ will be a lower bound on the degree computed by the semi-infinite program. In practice, one should choose $K\gg n$, which will still lead to solutions $p_n$ where some points exceed the desired error tolerance, but only by a small multiplicative constant close to unity. As the polynomial of best approximation is characterized by the equioscillation theorem~\cite{Trefethen2019Approximation}, one may iteratively refine the solution of~\cref{eq:polynomial_target_LP} by finding local maxima in error that do not satisfy~\cref{eq:polynomial_target_LP} and adding those points to $\mathcal{X}$.
This combined approach augments the Remez exchange algorithm to automatically handle cases of multiple sub-intervals or discontinuous non-uniform error targets where the sign of the error $p_n-f$ may not be strictly alternating.

We solve the linear program~\cref{eq:polynomial_target_LP} to establish an empirical scaling law on the approximate degree.
We find that $\min_f\widetilde{\mathrm{deg}}_{\epsilon/\alpha}(f)\approx a_\text{fit}\tau\log_2\frac{1}{c_\text{fit}\epsilon'}$ in the limit of small $\epsilon'$, where $f$ is any function that satisfies~\cref{eq:polynomial_target} for a given value of $\alpha$ and $\tau$. 
For $\tau\in\{1,2,4,8,16\}$, and $\alpha=e$, our results plotted in~\cref{fig:optimal_polynomial} indicate a fit of $a_\text{fit}\approx0.310, c_\text{fit}\approx31.2$, and we also find for small $\epsilon$ that $\alpha_++\alpha_-\approx 1.08$. By evaluating $a_\text{fit}$ for $\alpha\in\{e^{j/10}\}_{j=6}^{11}$, we also find empirically that the cost function $\alpha\cdot a_\text{fit}$ at $\alpha=e$ is within roughly $1\%$ of the minimum value. 
This choice of $\alpha$ is also consistent with our previous results.
In the limit of large time, the constant prefactor in LCHS of $\alpha_{\hat{f},R}\cdot R=\mathcal{O}(\alpha_{\hat{f}}\log\frac{1}{\epsilon})$ is directly comparable to $\frac{\alpha}{\tau}\min_f\widetilde{\mathrm{deg}}_{\epsilon/\alpha}(f)\approx \alpha a_\text{fit}\log_2\frac{\alpha}{c_\text{fit}\epsilon}$, plotted in~\cref{fig:lchs_comparison}, as optimal Hamiltonian simulation costs $\tau+\mathcal{O}(\log\frac{1}{\epsilon})$ queries.
We call this the `QSVT lower bound' as it omits the $\alpha_++\alpha_-$ factor, which is at least $1$ for small $\epsilon$.
In~\cref{thm:LCHS_entire}, we obtain the upper bound on the cost of $\alpha_{\hat{f}_2,R}R\le \frac{2}{c}e^c\mathrm{erfc}(\frac{1}{2\gamma})(c+\log\frac{1}{\epsilon}+\mathcal{O}(1))$, which is minimized at $\alpha_{\hat{f}_2,R}=e$ in the limit $\epsilon\rightarrow 0$. Similarly, the optimized solutions in~\cref{tab:generalized_solution} to our generalized kernel function suggest that $\alpha_{\hat{f}_{j,y},R}$ could converge to a value to close $e$.

\bibliographystyle{alphaUrlePrint}
\bibliography{main}

\appendix

\section{LCHS by an asymmetric block-encoding}\label{sec:schrodingerization}
In this section, we elucidate the connections between LCHS and Schr\"odingerization~\cite{Jin2024Schrodingerization}.
To simplify prose, we focus on the case of time-independent $L$ and $H$ as the following results generalize in a straightforward manner to the time-dependent case.
The exact LCHS from prior art~\cite{an2023lchs,an2023lchsoptimal} is the linear combination from~\cref{eq:exponential_decay_matrix}, which in the time-independent case is
\begin{align}\label{eq:LCHS_continuous_appendix}
\forall t\ge0,\;L\succeq0,\quad O(t)&=\frac{1}{\sqrt{2\pi}}\int_{\mathbb{R}}\hat{f}(k)e^{-i(kL+H)t}\mathrm{d}k=e^{-(L+iH)t}=e^{-At}.
\end{align}
As discussed in the main text, this may be relaxed to an approximate LCHS by~\cref{thm:LCHS_general}.
For the purposes of this discussion though, it suffices to ignore any approximation errors.
The block-encoding of $O(t)$ may be implemented by controlled-time-evolution
\begin{align}\label{eq:LCHS_appendix_controlled}
O(t)=\frac{\|\hat{f}_\text{l}\|_{L^2}\|\hat{f}_\text{r}\|_{L^2}}{\sqrt{2\pi}}(\bra{\hat{f}_\text{l}}\otimes\mathcal{I})\left(\int_{\mathbb{R}}\ket{k}\bra{k}\otimes e^{-i(kL+H)t}\mathrm{d}k\right)(\ket{\hat{f}_\text{r}}\otimes\mathcal{I}),
\end{align}
for any ancillae quantum states $\ket{\hat{f}_\text{l}},\ket{\hat{f}_\text{r}}$ such that
\begin{align}\label{eq:left_right_normalization}
\bra{\hat{f}_\text{l}}k\rangle\langle k|\hat{f}_\text{r}\rangle\doteq\frac{\hat{f}(k)}{\|\hat{f}_\text{l}\|_{L^2}\|\hat{f}_\text{r}\|_{L^2}},
\quad
\ket{\hat{f}_\text{r}}\doteq\frac{1}{\|\hat{f}_\text{r}\|_{L^2}}\int_\mathbb{R}\hat{f}_\text{r}(k)\ket{k}\mathrm{d}k,
\quad\ket{\hat{f}_\text{l}}\doteq\frac{1}{\|\hat{f}_\text{l}\|_{L^2}}\int_\mathbb{R}\hat{f}_\text{l}(k)\ket{k}\mathrm{d}k.
\end{align}
To simplify notation, we will drop the identity $\mathcal{I}$ symbol where it is clear.
Equivalent to~\cref{eq:LCHS_appendix_controlled} is defining the controlled-time-evolution unitary in terms of uncontrolled time-evolution by the dilated Hamiltonian~\cite{Jin2024Schrodingerization}
\begin{align}
\int_{\mathbb{R}}\ket{k}\bra{k}\otimes e^{-i(kL+H)t}\mathrm{d}k=
e^{-i(D\otimes L+\mathcal{I}\otimes H)t},
\quad D\doteq\int_{\mathbb{R}}k\ket{k}\bra{k}\mathrm{d}k.
\end{align}
Then 
\begin{align}\label{eq:LCHS_appendix_Hamiltonian}
O(t)=\frac{\|\hat{f}_\text{l}\|_{L^2}\|\hat{f}_\text{r}\|_{L^2}}{\sqrt{2\pi}}\bra{\hat{f}_\text{l}}e^{-i(D\otimes L+\mathcal{I}\otimes H)t}\ket{\hat{f}_\text{r}}=e^{-At}.
\end{align}
Let the state preparation unitaries $\textsc{Prep}_{\hat{f}_\text{l}}\ket{0}=\ket{\hat{f}_\text{l}}$ and $\textsc{Prep}_{\hat{f}_\text{r}}\ket{0}=\ket{\hat{f}_\text{r}}$.
Then we have the block-encoding
\begin{align}\label{eq:asymmetric_BE}
\textsc{Be}\left[\frac{e^{-At}}{\alpha_{\hat{f}_\text{l},\hat{f}_\text{r}}}\right]&=\textsc{Prep}_{\hat{f}_\text{l}}^\dagger\cdot e^{-i(D\otimes L+\mathcal{I}\otimes H)t}\cdot\textsc{Prep}_{\hat{f}_\text{r}}, \quad
\alpha_{\hat{f}_\text{l},\hat{f}_\text{r}}\doteq\frac{\|\hat{f}_\text{l}\|_{L^2}\|\hat{f}_\text{r}\|_{L^2}}{\sqrt{2\pi}}.
\end{align}
For any given kernel function $\hat{f}$, we now prove in~\cref{thm:symmetric_is_optimal} that the choice made in the main text of $|\hat{f}_\text{l}|\propto|\hat{f}_\text{r}|\propto{|\hat{f}(k)|^{1/2}}$ is optimal in that it minimizes the block-encoding normalization factor at $\alpha_{\hat{f}_\text{l},\hat{f}_\text{r}}=\alpha_{\hat{f}}=\frac{\|\hat{f}\|_{L^1}}{\sqrt{2\pi}}$.
\begin{theorem}[Optimality of the symmetric block-encoding]\label{thm:symmetric_is_optimal}
Let $|\hat{f}|:\mathbb{R}\rightarrow\mathbb{R}$ be any function such that $\|\hat{f}\|_{L^1}<\infty$.
Let $|\hat{f}_\text{l}|, |\hat{f}_\text{r}|$ be any function such that $\|\hat{f}_\text{l}\|_{L^2},\|\hat{f}_\text{r}\|_{L^2}<\infty$ and $\forall k\in\mathbb{R}$, $|\hat{f}_\text{l}(k)||\hat{f}_\text{r}(k)|=|\hat{f}(k)|$. Then
\begin{align}
\|\hat{f}\|_{L^1}\le\|\hat{f}_\text{l}\|_{L^2}\|\hat{f}_\text{r}\|_{L^2},
\end{align}
where equality is achieved if $|\hat{f}_\text{l}(k)|\propto|\hat{f}_\text{r}(k)|\propto|\hat{f}(k)|^{1/2}$.
\end{theorem}
\begin{proof}
Let us solve the optimization problem
\begin{align}
\min\|\hat{f}_\text{l}\|_{L^2}^2\|\hat{f}_\text{r}\|_{L^2}^2\quad\text{such that}\quad\forall k\in\mathbb{R},\;|\hat{f}_\text{l}(k)||\hat{f}_\text{r}(k)|=|\hat{f}(k)|.
\end{align}
Using the method of Lagrange multipliers, the derivative
\begin{align}\nonumber
0&=\frac{\partial}{\partial|\hat{f}_\text{l}(k)|}\left[\|\hat{f}_\text{l}\|_{L^2}^2\|\hat{f}_\text{r}\|_{L^2}^2-\int_{\mathbb{R}}\lambda(k)(|\hat{f}_\text{l}(k)||\hat{f}_\text{r}(k)|-|\hat{f}(k)|)\mathrm{d}k\right]
\\\nonumber
&=\frac{\partial}{\partial|\hat{f}_\text{l}(k)|}\left[\int_{\mathbb{R}}|\hat{f}_\text{l}(k)|^2\mathrm{d}k\|\hat{f}_\text{r}\|_{L^2}^2-\int_{\mathbb{R}}\lambda(k)(|\hat{f}_\text{l}(k)||\hat{f}_\text{r}(k)|-|\hat{f}(k)|)\mathrm{d}k\right]
\\\nonumber
&=2|\hat{f}_\text{l}(k)|\|\hat{f}_\text{r}\|_{L^2}^2-\lambda(k)|\hat{f}_\text{r}(k)|
\\
&=2|\hat{f}_\text{l}(k)|^2\|\hat{f}_\text{r}\|_{L^2}^2-\lambda(k)|\hat{f}_\text{l}(k)||\hat{f}_\text{r}(k)|.
\end{align}
By taking derivatives with respect to $|\hat{f}_\text{r}(k)|$, we also obtain $0=2|\hat{f}_\text{r}(k)|^2\|\hat{f}_\text{l}\|_{L^2}^2-\lambda(k)|\hat{f}_\text{l}(k)||\hat{f}_\text{r}(k)|$.
Subtracting these two stationary conditions implies that the stationary solution satisfies 
\begin{align}
|\hat{f}_\text{l}(k)|^2\|\hat{f}_\text{r}\|_{L^2}^2=|\hat{f}_\text{r}(k)|^2\|\hat{f}_\text{l}\|_{L^2}^2\quad\rightarrow\quad|\hat{f}_\text{l}(k)|\propto|\hat{f}_\text{r}(k)|.
\end{align}
By assumption, we have the condition
\begin{align}
|\hat{f}(k)|&=|\hat{f}_\text{l}(k)||\hat{f}_\text{r}(k)|
=
|\hat{f}_\text{l}(k)|^2\frac{\|\hat{f}_\text{r}\|_{L^2}}{\|\hat{f}_\text{l}\|_{L^2}}=|\hat{f}_\text{r}(k)|^2\frac{\|\hat{f}_\text{l}\|_{L^2}}{\|\hat{f}_\text{r}\|_{L^2}},
\\
\Rightarrow\|\hat{f}\|_{L^1}&=\int_\mathbb{R}|\hat{f}(k)|\mathrm{d}k=\int_\mathbb{R}|\hat{f}_\text{r}(k)|^2\mathrm{d}k\frac{\|\hat{f}_\text{l}\|_{L^2}}{\|\hat{f}_\text{r}\|_{L^2}}=\|\hat{f}_\text{l}\|_{L^2}\|\hat{f}_\text{r}\|_{L^2}.
\end{align}
To show that this stationary solution is the minimum and not the maximum, it suffices to exhibit any example where $\|\hat{f}\|_{L^1}<\|\hat{f}_\text{l}\|_{L^2}\|\hat{f}_\text{r}\|_{L^2}$. For example, let $\hat{f}_\text{l}(k)=\frac{1}{(1+k^2)^2}$, and $\hat{f}_\text{l}(k)=\frac{1}{1+k^2}$. Then $\|\hat{f}\|_{L^1}=\frac{3}{8}\pi=1.178...<1.241...=\sqrt{\frac{5}{32}}\pi=\|\hat{f}_\text{l}\|_{L^2}\|\hat{f}_\text{r}\|_{L^2}$.
\end{proof}

\subsection{Schr\"odingerization}
Schr\"odingerization for $L\succeq0$ as described in prior art~\cite{Jin2024Schrodingerization,Jin2025inhomogeneous,Jin2025Schrodingerization}, chooses
\begin{align}\label{eq:schro_kernel_states}
\ket{\hat{f}_\text{r}}&=\ket{\hat{f}}
\doteq\frac{1}{\|\hat{f}\|_{L^2}}\int_\mathbb{R}\hat{f}(k)\ket{k}\mathrm{d}k.
\end{align}
The ``recovery map'' state $\ket{\hat{f}_\text{l}}$ is usually expressed in the position basis $\ket{x}_\text{pos}=\frac{1}{\sqrt{2\pi}}\int_\mathbb{R}e^{ikx}\ket{k}\mathrm{d}k$ and chosen to be
\begin{align}\label{eq:recovery_state}
\ket{\hat{f}_\text{l}}=\int_0^\infty \sqrt{2}e^{-x}\ket{x}_\text{pos}=\frac{1}{\|\hat{f}_\text{l}\|_{L^2}}\int_{\mathbb{R}}\hat{f}_\text{l}\ket{k}\mathrm{d}k,
\quad
\hat{f}_\text{l}(k)=\frac{2}{1-i k},\quad \|\hat{f}_\text{l}\|_{L^2}=2\sqrt{\pi}.
\end{align}
By~\cref{eq:LCHS_appendix_Hamiltonian},
this block-encodes a LCHS with the effective kernel function $\hat{f}_\text{eff}=\frac{2\hat{f}(k)}{1+ik}$.
Note that the normalization of $\hat{f}_\text{l}$ is chosen so that the residual at $z=-i$ of $\hat{f}_\text{eff}$ and $\hat{f}$ are equal.
As the pole is in the upper complex plane, only the $x<0$ component of $f_\text{eff}(x)$ is affected. Hence,
\begin{align}\label{eq:asymmetric_block_encoding_O}
O(t)=\frac{\|\hat{f}_\text{l}\|_{L^2}\|\hat{f}\|_{L^2}}{\sqrt{2\pi}}\bra{\hat{f}_\text{l}}e^{-i(D\otimes L+\mathcal{I}\otimes H)t}\ket{\hat{f}_\text{r}}=\frac{1}{\sqrt{2\pi}}\int_{\mathbb{R}}\hat{f}_\text{eff}(k)e^{-i(kL+H)t}\mathrm{d}k=e^{-At}.
\end{align}
Note that other recovery maps are possible, along with different choices of $D$ and $\ket{\hat{f}_\text{r}}$~\cite{Li2025MomentMatching}.
These will define different $\hat{f}_\text{eff}$ that can be reduced to evaluating~\cref{thm:LCHS_general}.
However, by~\cref{thm:symmetric_is_optimal}, the normalization factor of any such asymmetric block-encoding will be suboptimal compared to the symmetric choice $|\hat{f}_\text{l}(k)|\propto|\hat{f}_\text{r}(k)|\propto|\hat{f}_\text{eff}(k)|^{1/2}$.
For example, the normalization factors for the following kernel functions are
\begin{align}
\hat{f}&=\sqrt{\frac{2}{\pi}}\frac{e^{c(1-ik)}}{1+k^2}=\hat{f}_2(k;\infty,c),
&\hat{f}_\text{eff}&=\hat{f}^*_\text{l}(k)\hat{f}(k)=\sqrt{\frac{8}{\pi}}\frac{e^{c(1-ik)}}{(1+ik)(1+k^2)}=\hat{f}_{3,1}(k;\infty,c),
\\
\alpha_{\hat{f}}&=e^c
 \qquad\qquad\qquad<
&\alpha_{\hat{f}_\text{eff}}&=\frac{4}{\pi}e^c
\qquad\qquad\qquad<\qquad\qquad\quad\alpha_{\hat{f}_\text{l},\hat{f}}=\sqrt{2}e^c.
\end{align}
One may also show that $\alpha_{\hat{f}_2}\le\frac{1}{\sqrt{2}}\alpha_{\hat{f}_\text{l},\hat{f}_2}$ for our optimal kernel $\hat{f}_2(k;\gamma,c)$ with any $\gamma,c>0$.

\subsection{Position basis analysis}
In contrast to the asymmetric block-encoding interpretation of~\cref{eq:asymmetric_block_encoding_O}, 
Schr\"odingerization in prior art is usually analyzed in the position basis $\ket{x}_\text{pos}$.
For completeness, we review this analysis and demonstrate that both pictures are equivalent.
To simplify notation, we drop the position subscript $_\text{pos}$ where it is clear.

Schr\"odingerization defines the time-evolved state 
\begin{align}\nonumber
\ket{w(t;\vec{u}_0)}&\doteq e^{-i(D\otimes L+\mathcal{I}\otimes H)t}\ket{\hat{f}_\text{r}}\ket{\vec{u}_0}=
\frac{1}{\|\hat{f}\|_{L^2}}\left(\int_{\mathbb{R}}\ket{k}\bra{k}\otimes e^{-i(kL+H)t}\mathrm{d}k\right)\left(\int_\mathbb{R}\hat{f}(k)\ket{k}\mathrm{d}k\right)\otimes\ket{\vec{u}_0}
\\
&=\frac{1}{\|\hat{f}\|_{L^2}}\int_{\mathbb{R}}\hat{f}(k)\ket{k}\otimes \left(e^{-i(kL+H)t}\ket{\vec{u}_0}\right)\mathrm{d}k.
\end{align}
Using resolution of identity $\mathcal{I}=\int_\mathbb{R}\ket{x}\bra{x}_\text{pos}\mathrm{d}x$ in the position basis, where $(\langle x|_\text{pos})\ket{k}=\frac{1}{\sqrt{2\pi}}e^{-ikx}$,
\begin{align}\nonumber
\ket{w(t;\vec{u}_0)}&=\int_\mathbb{R}\ket{x}\bra{x}w(t;\vec{u}_0)\rangle\mathrm{d}x
=\frac{1}{\|\hat{f}\|_{L^2}}\int_\mathbb{R}\ket{x}\int_{\mathbb{R}}\hat{f}(k)\langle x\ket{k}\otimes \left(e^{-i(kL+H)t}\ket{\vec{u}_0}\right)\mathrm{d}k\;\mathrm{d}x
\\
&=\frac{1}{\sqrt{2\pi}\|\hat{f}\|_{L^2}}\int_\mathbb{R}\ket{x}\otimes\int_{\mathbb{R}}\hat{f}(k)e^{-i(k(L+x/t)+H)t}\ket{\vec{u}_0}\mathrm{d}k\;\mathrm{d}x.
\end{align}
As $L+x/t\succeq0$ for any $x,t\ge0$,~\cref{eq:LCHS_continuous_appendix} implies
\begin{align}\nonumber
\ket{w(t;\vec{u}_0)}&=\frac{1}{\|\hat{f}\|_{L^2}}\int_0^\infty\ket{x}\otimes e^{-(L+x/t+iH)t}\ket{\vec{u}_0}\mathrm{d}x
+
\int_{(-\infty,0)}\cdots \ket{x}\otimes\cdots\mathrm{d}x
\\
&=
\frac{1}{\|\hat{f}\|_{L^2}}\int_0^\infty e^{-x}\mathrm{d}x\ket{x}\otimes e^{-At}\ket{\vec{u}_0}+
\int_{(-\infty,0)}\cdots \ket{x}\otimes\cdots\mathrm{d}x.
\end{align}

Hence, a position basis measurement on the ancilla register returns $\ket{x}_\text{pos}\ket{\vec{u}(t)}$ if $x\ge0$ is measured, with some probability density $p_x$. The desired time-evolved state $\ket{\vec{u}(t)}$ is obtained with overall probability $P$ where
\begin{align}\label{eq:Schrodingerization_success}
\forall x\ge0,\quad p_x=|\langle x\ket{w(t;\vec{u}_0)}|^2=\frac{e^{-2x}|\vec{u}(t)|^2}{\|\hat{f}\|_{L^2}^2},
\quad
P=\int_0^\infty p_x\mathrm{d}x=\frac{|\vec{u}(t)|^2}{2\|\hat{f}\|_{L^2}^2}.
\end{align}
By collecting the $\ket{x}$ components for $x\ge0$ into a normalized quantum state, which turns out to be $\ket{\hat{f}_\text{l}}$ from~\cref{eq:recovery_state},
\begin{align}
\ket{w(t;\vec{u}_0)}&=\frac{1}{\sqrt{2}\|\hat{f}\|_{L^2}}\ket{\hat{f}_\text{l}}\otimes e^{-At}\ket{\vec{u}_0}+
\int_{(-\infty,0)}\cdots \ket{x}\otimes\cdots\mathrm{d}x.
\end{align}
In other words $\ket{w(t;\vec{u}_0)}$ on the $x\ge0$ subspace is a product state $\ket{\hat{f}_{\text{l}}}\ket{\vec{u}(t)}$ with amplitide $\frac{|\vec{u}(t)|}{\sqrt{2}\|\hat{f}\|_{L^2}}$ for any $\ket{\vec{u}_0}$. Hence the recovery map $\ket{\hat{f}_{\text{l}}}$ minimizes the block-encoding normalization in
\begin{align}\label{eq:Schro_BE}
\bra{\hat{f}_\text{l}}e^{-i(D\otimes L+\mathcal{I}\otimes H)t}\ket{\hat{f}_\text{r}}=\frac{e^{-At}}{\sqrt{2}\|\hat{f}\|_{L^2}}=\frac{e^{-At}}{\frac{1}{\sqrt{2\pi}}\|\hat{f}_\text{l}\|_{L^2}\|\hat{f}\|_{L^2}}=\frac{e^{-At}}{\alpha_{\hat{f}_\text{l},\hat{f}}},
\end{align}
which, as expected, is identical to~\cref{eq:asymmetric_BE}. 

When applied to any initial state $\ket{\vec{u}_0}$, this produces $\ket{\vec{u}(t)}$ with the same success probability $P$ as~\cref{eq:Schrodingerization_success}. 
Hence, the interpretation of performing amplitude amplification on the $x\ge0$ space of~\cref{eq:Schrodingerization_success} or applying the block-encoding~\cref{eq:Schro_BE} on an initial state $\ket{\vec{u}_0}$ followed by amplitude amplification have equal cost.

\end{document}